%% file: main.tex
\documentclass[a4paper,twocolumn,11pt,accepted=2021-07-08]{quantumarticle}
\pdfoutput=1

\usepackage[utf8]{inputenc}
\usepackage{amsmath}
\usepackage{amssymb}
\usepackage{bbold}
\usepackage{graphicx}
\usepackage{braket}
\usepackage{leftidx}
\usepackage{tikz, circuitikz}
\usepackage{mathtools}
\usepackage{ stmaryrd }
\usepackage{circuitikz}
\usepackage{textcomp}
\usepackage{accents}
\usepackage{comment}
\usepackage{ mathrsfs }
\usepackage{breakurl,amsmath,amsthm,amssymb,xspace,xypic,enumerate,color,tikzfig}
\usepackage{float}
\usepackage[numbers]{natbib}

\input{tikzstyles.tex}

\newcommand{\disc}{%
\,\begin{tikzpicture}[yshift=1.5mm]
\node [disc] (a) at (0,0) {};
\node [style=none] (0) at (0, -0.7) {}; 
\node [style=none] (1) at (0, -0.15) {}; 
\draw [style=none] (0.center) to (1.center);
\end{tikzpicture}\,}
\usetikzlibrary{backgrounds}
\usetikzlibrary{arrows}
\usetikzlibrary{shapes,shapes.geometric,shapes.misc}
\usetikzlibrary{positioning,chains,fit,shapes,calc}

\tikzstyle{tikzfig}=[baseline=-0.25em,scale=4]

\pgfkeys{/tikz/tikzit fill/.initial=0}
\pgfkeys{/tikz/tikzit draw/.initial=0}
\pgfkeys{/tikz/tikzit shape/.initial=0}
\pgfkeys{/tikz/tikzit category/.initial=0}

\pgfdeclarelayer{edgelayer}
\pgfdeclarelayer{nodelayer}
\pgfsetlayers{background,edgelayer,nodelayer,main}

\tikzstyle{none}=[inner sep=0mm]

\tikzstyle{every loop}=[]

\newcommand\abs[1]{\left|#1\right|}

\newcommand{\tr}{\mathrm}

\def\be{\begin{equation}}
\def\ee{\end{equation}}
\def\ba{\begin{align}}
\def\ea{\end{align}}
\newcommand{\cat}[1]{\ensuremath{\mathbf{#1}}}
\newcommand{\FHilb}{\cat{FHilb}}

\newcommand{\CPM}{\ensuremath{\mathrm{CPM}}\xspace}

\DeclareMathOperator{\Tr}{Tr}

\newcommand{\ca}{\mathcal A}
\newcommand{\cb}{\mathcal B}
\newcommand{\cc}{\mathcal C}
\newcommand{\cd}{\mathcal D}
\newcommand{\ce}{\mathcal E}
\newcommand{\cf}{\mathcal F}

\newcommand{\ch}{\mathcal H}

\newcommand{\ck}{\mathcal K}
\newcommand{\cl}{\mathcal L}

\newcommand{\calp}{\mathcal P}

\newcommand{\cs}{\mathcal S}
\newcommand{\ct}{\mathcal T}
\newcommand{\cu}{\mathcal U}

\newcommand{\cw}{\mathcal W}
\newcommand{\cx}{\mathcal X}

\newcommand{\cz}{\mathcal Z}

\def\tl{\tilde}

\usepackage{bbm}

\newcommand{\lalg}[1]{\mathfrak{#1}}  

\renewcommand{\sl}{\lalg{sl}}

\def\tl{\tilde}

\newtheorem{definition}{Definition}

\newtheorem{theorem}{Theorem}

\newtheorem{lemma}{Lemma}

\newcommand{\correction}[1]{{#1}}

\begin{document} 

\title{Routed quantum circuits}

\author{Augustin Vanrietvelde}
\orcid{0000-0001-9022-8655}
\email{a.vanrietvelde18@imperial.ac.uk}
\affiliation{Quantum Group, Department of Computer Science, University of Oxford}
\affiliation{Department of Physics, Imperial College London}
\affiliation{HKU-Oxford Joint Laboratory for Quantum Information and Computation}

\author{Hl\'er Kristj\'ansson}
\orcid{0000-0003-4465-2863}
\email{hler.kristjansson@st-annes.ox.ac.uk}
\affiliation{Quantum Group, Department of Computer Science, University of Oxford}
\affiliation{HKU-Oxford Joint Laboratory for Quantum Information and Computation}

\author{Jonathan Barrett}
\orcid{0000-0002-2222-0579}
\email{jonathan.barrett@cs.ox.ac.uk}
\affiliation{Quantum Group, Department of Computer Science, University of Oxford}

\begin{abstract}
      We argue that the quantum-theoretical structures studied in several recent lines of research cannot be adequately described within the standard framework of quantum circuits. This is in particular the case whenever the combination of subsystems is described by a nontrivial blend of direct sums and tensor products of Hilbert spaces. We therefore propose an extension to the framework of quantum circuits, given by \textit{routed linear maps} and \textit{routed quantum circuits}. We prove that this new framework allows for a consistent and intuitive diagrammatic representation in terms of circuit diagrams, applicable to both pure and mixed quantum theory, and exemplify its use in several situations, including the superposition of quantum channels and the causal decompositions of unitaries. We show that our framework encompasses the `extended circuit diagrams' of Lorenz and Barrett [arXiv:2001.07774 (2020)], which we derive as a special case, endowing them with a sound semantics. 
\end{abstract}

\maketitle

\tableofcontents

\section{Introduction}

Finite-dimensional quantum theory \correction{ is often} described through the framework of quantum circuits \cite{deutsch1989quantum,aharonov1998quantum,abramsky2004categorical,nielsen2000quantum,coecke_kissinger_2017}. This framework is built on the possibility of both sequential and parallel compositions, with the latter represented by tensor products of completely positive trace-preserving  maps. This enables quantum theory to be expressed  graphically through circuit diagrams, the intuitive nature of which is an important reason for the success of the framework of quantum circuits.

Yet, two different lines of research recently led to the introduction of scenarios and structures which, as we shall argue in the present work, cannot be described in a fully satisfactory way when using the sole framework of standard quantum circuits\footnote{This has to be distinguished from the recent surge in interest in indefinite causal order (ICO) \cite{hardy2005probability,chiribella2009beyond,chiribella2013quantum,oreshkov2012quantum}, which also goes beyond standard quantum circuits in a different way. Even though the framework presented here would have natural applications to the study of ICO \cite{barrett2020cyclic}, it is important to emphasise that the features that differentiate it from standard quantum circuits are not of the same kind as those that give rise to ICO; all the constructions under study in the present work display a definite causal order.}. These scenarios and structures share two deeply interlinked key features. The first of these features is the presence of quantum channels satisfying specific constraints on their ability to relate given sectors (i.e.\ orthogonal subspaces) of their input and output spaces (we shall call these constraints `sectorial constraints'). The second is the existence of parallel compositions of such channels which consist in non-trivial blends of tensor products and direct sums (we shall characterise them as featuring `sectorial correlations' between parallel systems).

\correction{ The first of the lines of research exhibiting these features is that of coherent control, or superpositions, of quantum channels \cite{aharonov1990superpositions,aaberg2004subspace,aaberg2004operations,oi2003interference}. The coherent control of quantum channels generalises the `if' clause in classical computation to the quantum level, where the use or not of a given quantum channel is determined by the value of another quantum system -- enabling the possibility of using multiple quantum channels in a superposition. This idea can be applied equally to computation, where it is often referred to as control of unknown unitaries (or channels) \cite{zhou2011adding, soeda2013limitations, araujo2014, friis2014, gavorova2020topological,thompson2018, dong2019controlled} and communication, where it takes the form of communication in a superposition of trajectories \cite{gisin2005error,abbott2020communication,chiribella2019shannon,kristjansson2020single,kristjansson2020resources,lamoureux2005experimental,rubino2020experimental}.  We shall describe superpositions of quantum channels in depth in Section \ref{sec:CommSuperpos}, focusing on the communication perspective of Ref.\ \cite{chiribella2019shannon}. In this line of research, in particular from the communication point of view,} sectorial constraints and non-trivial parallel compositions have a direct physical meaning, as they serve to describe a physical setup in which a particle is sent in a coherent superposition of two transmission lines. The sectorial constraints correspond to the fact that transmission lines are required not to modify the number of particles going through them. The non-trivial parallel composition of the two transmission lines corresponds to the fact that there is exactly one particle in total.

A second line of research, that of causal decompositions \cite{Allen2017, barrett2019, lorenz2020,barrett2020cyclic} (described more comprehensively in Section \ref{sec:CausalDecs}), aims at probing the equivalence between the compatibility of a unitary channel with a given causal structure, and the existence of a decomposition for this channel along a given unitary circuit in which this causal structure is made obvious. In this context, it was found that, in order to prove such an equivalence in some cases, one has to consider decompositions in which some of the unitary channels satisfy sectorial constraints and are parallelly composed in a non-trivial way. In other words, describing the full scope of possible compositional structures of channels in quantum theory requires the introduction of sectorial constraints and non-trivial parallel compositions -- something which cannot be done through the use of standard quantum circuits. Here, in contrast with the previous example, these new structures are not introduced to depict physical specifications, but to unlock a more flexible mathematical description.

It is striking that two very different lines of research led to the introduction of the same structures; this might be an indication that such structures are common in the study of quantum phenomena. Indeed, they can be seen as being the proper tools to depict the breaking down of a finite-dimensional C* algebra into (possibly non-factor) commuting subalgebras \cite{sinclair_smith_2008}, a decomposition which is central to powerful theorems and techniques in quantum theory, such as the Schur-Weyl duality and its generalisations \cite{marvian2014, Harrow2005ApplicationsOC}, or the existence of decoherence-free subspaces \cite{palma1996quantum, duan1997preserving, zanardi1997noiseless, lidar1998decoherence, beige2000quantum, kwiat2000experimental}. Possibly non-factor sub-C* algebras have also been argued to yield the proper mathematical notion of subsystems in quantum theory (for an overview, see Ref. \cite{chiribella2018agents} and the references therein), \correction{and have recently been studied for the properties of their entanglement entropy \cite{Ma_2016, lin2018comments, Bianchi_2019}.}
Finally, these structures are common in quantum optical setups, in which photons are usually sent in a superposition of trajectories; it has already been argued in \cite{araujo2014} that such setups cannot in general be suitably described by standard quantum circuits \correction{ (the discussion in \cite{friis2014} also relates to this point, see footnote \ref{foot:araujoFriis}).}

Even though they were shown to be fairly natural and physically meaningful in the two lines of research in which they have been introduced, the new structures corresponding to sectorial constraints and correlations were described, and had their consistency proven, only at the level of the specific situations considered in each case; no systematic description and formalisation has been undertaken. As we shall argue in Section \ref{sec:motivation}, such a systematic description cannot be faithfully obtained when solely using the framework of quantum circuits, as this framework lacks the tools both to hardcode sectorial constraints on quantum channels, and to describe non-trivial parallel compositions. An extension of this framework is therefore needed.

Such an extension would need to satisfy several requirements. First, given the strong compositional flavour of the structures at hand, it would need to be appropriately compositional (i.e., emphasising the sequential and parallel composition operations on channels and their structures). Second, drawing lessons from recent successful re-expressions of quantum theory \cite{hardy2016reconstructing, chiribella2011informational, masanes2011derivation, selby2018reconstructing, abramsky2004categorical}, this extension would have to form a \textit{process theory} \cite{coecke_kissinger_2017}, focusing on dynamical processes (i.e., channels) and obtaining states as specific cases of the latter: the emphasis on compositional structure is more natural in process theories. Third, it should yield intuitive, higher-level (i.e.\ diagrammatic) expressions, from which the sectorial constraints and correlations can be read off. Finally, this extension should be general enough, encompassing all possible situations in which sectorial constraints and correlations could appear, and should yield a theory of pure states and unitaries as well as one of mixed states and general quantum channels.

The aim of this paper is to present such an extension, obtained through the introduction of \textit{routed maps}. In addition to a linear map, a routed map includes a boolean matrix, called the \textit{route}, which encodes constraints on the possibility for the linear map to relate given subspaces of its input and output spaces (sectorial constraints). Sectorial correlations are then obtained as a contextual feature, specified by the global circuit in which a parallel composition is present (and, more specifically, by the routes in this circuit). We prove that routed maps form a consistent theory which faithfully describes the structures previously mentioned, and which possesses the compositional features required to accept a consistent diagrammatic expression. In both the pure and mixed theories, we single out physically meaningful routed maps and provide rules for the (non-trivial) problem of their sequential composition. We present and exemplify a general and user-friendly diagrammatic representation for routed maps, and show how it allows for faithful and intuitive depictions of physical scenarios.

Finally, we use our framework to provide a solid mathematical basis for the `extended quantum circuits' of Ref.\ \cite{lorenz2020}, previously introduced in a qualitative and non-systematic way for the study of causal decompositions of unitary channels. We show here how a sub-framework of routed maps, that of \textit{index-matching routed maps}, allows to give a precise meaning to the `extended quantum circuits' of Ref.\ \cite{lorenz2020} (which we shall here denote as `index-matching quantum circuits'). We use this framework to provide a systematic account of index-matching quantum circuits and of their interpretations, and in particular to prove a simple rule which singles out the well-indexed index-matching quantum circuits, which correspond to physical circuits. 

One of the main achievements of this paper is to prove that the frameworks we build can be expressed in a fully diagrammatic language, in terms of circuits. The proper mathematical characterisation of this feature is given by category theory, and more specifically by the fact that these frameworks form dagger symmetric monoidal categories (dagger SMCs). In order not to burden the paper with formal notions, we leave our discussion of these categorical features to the appendices; in the main text, we will be describing the structural qualities of our frameworks in more natural terms. An account of the categorical constructions leading to routed theories will also be provided in an upcoming paper \cite{vanrietvelde2020categorical}.

The structure of this paper is as follows. In Section \ref{sec:motivation}, we argue for the necessity of going beyond standard quantum circuits, by presenting and analysing the results of two recent lines of research. In Section \ref{sec:PureRoutedMaps}, we present the framework of routed maps in the case of pure quantum theory, prove that it allows for a consistent diagrammatic representation, and show it in action by means of an example. In Section \ref{sec:RoutedDiagrams}, we describe the kinds of diagrams this framework leads to, and their interpretation. In Section \ref{sec:MixedRoutedMaps}, we extend our framework to encompass general quantum channels, and illustrate with an example how it leads to natural decoherence computations. Finally, in Section \ref{sec:IndexMatching}, we show how a sub-framework of routed maps gives rise to the index-matching quantum circuits of Ref.\ \cite{lorenz2020}, endowing them with a sound and systematic semantics.

\section{Motivation} \label{sec:motivation}

In order to motivate the introduction of routed maps, we will first describe the objects of study and the results of two recent lines of research, and argue why they cannot be described in a fully satisfactory manner within the sole framework of quantum circuits\footnote{\correction{In technical terms, this means that they cannot be adequately described using the categories \FHilb\,or CPM[\FHilb]. In Appendix \ref{app:CP*}, we extend this discussion to the case of other standard categorical constructions such as CP*[\FHilb] and the Karoubi envelope of CPM[\FHilb].}}.

\subsection{Communication and computation in a superposition of trajectories} \label{sec:CommSuperpos}

The study of communication in a superposition of trajectories was recently proposed as a framework which extends standard quantum Shannon theory  \cite{chiribella2019shannon,kristjansson2020resources}. \correction{In this framework, not only is the information carried by a particle allowed to be in a quantum state, but also the trajectory of the quantum information carrier is allowed to be in a coherent superposition of different transmission lines. The goal of Refs.\ \cite{abbott2020communication,chiribella2019shannon} was to study the communication advantages that using such scenarios can lead to. In a similar spirit, several recent works have investigated the use in quantum  computation of controlling unknown `black box' quantum operations, where a quantum control system determines if one unknown operation is used, or another, or both in a superposition \cite{zhou2011adding, soeda2013limitations, araujo2014, friis2014, gavorova2020topological,thompson2018, dong2019controlled}.
Here, we will describe the challenges that these works raise in terms of mathematical formalisation.} To do so, we will focus on a paradigmatic example of communication in a superposition of trajectories, the `one particle in two trajectories' scenario \cite{chiribella2019shannon,abbott2020communication}, focusing specifically on the perspective of Ref.\ \cite{chiribella2019shannon}. 

Even though we will focus on this specific example in order to present sharp arguments on a well-defined situation, the following analysis applies to a wide range of quantum-optical or interferometric setups to show that, in general, they cannot be faithfully described by standard quantum circuits. Indeed, the features of this example are common in quantum optics, in which sending a photon in a superposition of trajectories, and applying operations which do not modify the number of photons (such as waveplates) are standard protocols. This point has, in particular, been encountered in different guises in the literature discussing the possibility of coherently controlling unknown unitaries \cite{zhou2011adding, araujo2014, friis2014, gavorova2020topological} or quantum channels \cite{aharonov1990superpositions,aaberg2004subspace,aaberg2004operations,oi2003interference, thompson2018, abbott2020communication,chiribella2019shannon, dong2019controlled}: no-go theorems forbid such coherent control in standard quantum circuits, yet it is achievable in simple quantum optical implementations. This is because these quantum optical implementations  of coherent control of quantum channels cannot be correctly described by standard quantum circuits\footnote{\label{foot:araujoFriis}On this point, see in particular the discussions in \cite{araujo2014} and \cite{friis2014}. Ref.\ \cite{araujo2014} concludes that `the language of quantum circuits should be extended in order to capture all information processing possibilities allowed by quantum physics': the present work precisely describes such an extension. \correction{Ref.\,\cite{friis2014} argues that implementations of coherent control can be represented using standard quantum circuits, but that `the  correct  circuit  representation  of  a  schematic  may  not  resemble  it  visually'; in terms of our requirements for `adequate representations', this entails that this representation in standard quantum circuits will be inadequate as it does not neatly distinguish the subsystem on which the operation to be controlled is acting.}}.

\correction{In this work, we will not discuss further the specific problem of coherently controlling an unknown unitary; we will provide a detailed analysis and an account in terms of routed circuits in a separate work \cite{vanrietvelde2021coherent}. In general, however, we note that the coherent control of a unitary is a special case of a superposition of trajectories through two independent unitaries, discussed below, where one of the two unitaries is the identity.}

Going back to our main example, let us describe it in physical terms. There is a sender $S$ and a receiver $R$, and there are two different communication lines from $S$ to $R$. These communication lines, $A$ and $B$, can be thought of as being under the control of two agents, Alice and Bob. The sender wants to transmit a qudit message (also called a particle) $M$, and there is an additional control qubit $C$, whose value coherently determines whether the message goes through Alice's or Bob's communication lines. The agent who does not get the message is instead handed a `vacuum state', orthogonal to the possible states of the message. The agents' operations do not modify the number of particles in their communication lines. Afterwards, the two communication lines are merged back, yielding again a message and a control qubit for the receiver to analyse.

A mathematical description of this scenario is the following. The communication lines are represented as quantum channels (i.e.\ completely positive trace-preserving (CPTP) maps) $\ca$ and $\cb$, acting on $\cl (\ch_A)$ and $\cl (\ch_B)$, where $\ch_A$ and $\ch_B$ are Hilbert spaces of dimension $d+1$, and $\cl (\ch)$ is the set of linear operators on $\ch$. Moreover, each of these Hilbert spaces is partitioned into two orthogonal subspaces (also called sectors): a one-dimensional subspace corresponding to the vacuum state, and a $d$-dimensional subspace corresponding to the possible states of the message. This partitioning is written as $\ch_A = \ch_A^0 \oplus \ch_A^1, \, \ch_B= \ch_B^0 \oplus \ch_B^1$. $\mathcal{A}$ and $\mathcal{B}$ are constrained to map the vacuum sector to itself, and the one-particle sector to itself; this is called the \textit{no-leakage condition} in Ref.\ \cite{chiribella2019shannon}. Finally, initialisation of the trajectories' superposition is described by a unitary channel $\mathcal{U}$ from $\cl (\ch_M \otimes \ch_C)$ to $\cl ( \widetilde{\ch}_{AB} )$, where $\widetilde{\ch}_{AB} := \ch_A^1 \otimes \ch_B^0 \,\, \oplus \,\, \ch_A^0 \otimes \ch_B^1$ is the one-particle subspace of $\ch_A \otimes \ch_B$; the termination of the superposition is given by $\mathcal{U}^\dagger$ \correction{ (it could also be given by any other unitary channel from $\cl ( \widetilde{\ch}_{AB} )$ to $\cl (\ch_M \otimes \ch_C)$)}. Note that this scenario is consistent because of the crucial requirement that the vacuum and one-particle sectors are preserved by $\mathcal{A}$ and $\mathcal{B}$.

Let us now elaborate on why the sole framework of \correction{standard} quantum circuits cannot provide a fully satisfactory description of this scenario. Given that Alice and Bob's channels are applied in parallel and in different regions of spacetime, a natural diagrammatic representation of this scenario should have the following form (here, $\ce$ stands for `encoding' and $\cd$ for `decoding'):

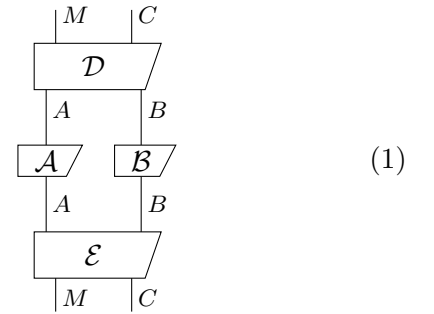
\be \label{eq:SatisfactoryRep} %
\InputIfFileExists{SuperpositionOfPaths.tikz}{}{\input{./figures/SuperpositionOfPaths.tikz}} \ee

If this diagram is to be understood as a quantum circuit, its boxes correspond to CPTP maps, and its wires correspond to spaces of linear operations on Hilbert spaces. The conjunction of wires $A$ and $B$ should correspond to the space $\cl (\ch_A \otimes \ch_B)$. This description as a standard quantum circuit depicts neither sectorial correlations -- here, the fact that $B$ gets the vacuum state if $A$ gets the particle, and reciprocally -- nor sectorial constraints -- here, the fact that $\ca$ and $\cb$ have to preserve the number of particles.

A first, general source of dissatisfaction with such a description is therefore its looseness: it does not include key elements of the scenario. For example, $\ce$ is a map to $\cl (\ch_A \otimes \ch_B)$, and not to its subspace $\cl (\widetilde{\ch}_{AB})$, even though the latter is the sector which will always be `used' in this scenario, no matter Alice and Bob's choices of operations; in general, one is forced to manipulate Hilbert spaces that are `too big' in comparison with the Hilbert spaces in actual use. Such a looseness implies, for example, that $\cd$ cannot be taken to be equal to $\cu^\dagger$, as the latter has input space $\cl (\widetilde{\ch}_{AB})$ and is therefore not trace-preserving if extended to $\cl (\ch_A \otimes \ch_B)$. To make $\cd$ trace-preserving, one will have to incorporate to it elements that specify how it acts on the other sectors of $\cl (\ch_A \otimes \ch_B)$, even though this is irrelevant information as far as the scenario is concerned.

This leads to a second, sharper critique: if we restrict to the scenario where $\ca$ and $\cb$ are unitary channels, then a description in standard quantum circuits will have to use the non-unitary channels $\ce$ and $\cd$, even though the scenario itself is then fully unitary. In other words, this is then an example of a fully unitary scenario which cannot be described in unitary quantum circuits.

The absence of a depiction of sectorial constraints on $\ca$ and $\cb$ is also a source of looseness, as the quantum circuit description will be missing a key aspect of the scenario, one that is pivotal not only to its physical interpretation -- as the constraints correspond to a natural physical requirement, that of not modifying the number of  -- but also to its consistency at the mathematical level, as $\ca$ and $\cb$ do not map $\cl(\widetilde{\ch}_{AB})$ to itself if they do not respect the sectorial constraints.

The mathematical consistency of a scenario can of course always be proven `by hand', by explicitly stating the sectorial correlations and constraints and showing that they lead to an overall consistent description; Refs.\ \cite{chiribella2019shannon,kristjansson2020single} proceed in this way. Yet, such an unsystematic account has two drawbacks. First, it yields no diagrammatic depiction of the existence of crucial sectorial correlations and constraints. Second, it will not scale up nicely to more elaborate situations: for more particles, for more channels, for more intricate sectorial constraints (where Alice, for example, is allowed to destroy a particle but not to create one), and for sectorial constraints which also include coherence conditions (for example, if Alice's action on one sector is restricted to be decoherent with her action on another sector). Proving the consistency of such scenarios by hand would quickly become tedious. In contrast, a systematic and consistent framework, such as the framework of routed maps, will not only enable a faithful and intuitive depiction of their structural features, but also include general consistency theorems ensuring that they are well-defined.

\subsection{Causal decompositions of unitary channels} \label{sec:CausalDecs}

We now move on to our second example. The study of causal decompositions \cite{Allen2017, barrett2019, lorenz2020,barrett2020cyclic} aims at exploring the connection between the \textit{causal structure} and the \textit{compositional structure} of unitary channels. Given a unitary channel with several inputs and several outputs (i.e., such that their input and output Hilbert spaces have preferred factorisations into tensor products), its causal structure is described by a set of no-influence relations. A no-influence relation is the fact that modifications of a given input cannot affect a given output. Formally, it is defined in the following way: if $\cu$ is a unitary channel from inputs $A$ and $B$ to outputs $C$ and $D$, one says that $A$ cannot influence $D$ (written $A \not\to D$) if there exists a channel $\cc$ such that

\be \label{Causal} %
\InputIfFileExists{CausalDec1.tikz}{}{\input{./figures/CausalDec1.tikz}} \,  = \, %
\InputIfFileExists{CausalDec2.tikz}{}{\input{./figures/CausalDec2.tikz}} \, , \ee
where the symbol $\disc$ denotes the trace-out channel. 

Compositional structure, on the other hand, corresponds to the existence of decompositions of a given unitary channel into several unitary channels along a given graph. For example, there exist unitary channels from inputs $A$ and $B$ to outputs $C$ and $D$ which admit a decomposition of the following form:
 
\be \label{Comp} %
\InputIfFileExists{CompDec1.tikz}{}{\input{./figures/CompDec1.tikz}} \,  = \, %
\InputIfFileExists{CompDec2.tikz}{}{\input{./figures/CompDec2.tikz}} \, . \ee

Clearly, for a given unitary channel, the existence of a decomposition of the form (\ref{Comp}) implies that this channel satisfies the no-influence relation $A \not\to D$, as described in (\ref{Causal}). Interestingly, the converse is also true \cite{Eggeling_2002}: if a unitary channel satisfies $A \not\to D$, then it admits a decomposition of the form (\ref{Comp}). The central conjecture that the research programme of \textit{causal decompositions} aims to probe is that this equivalence between causal and compositional structures for unitary channels, which we just illustrated in a simple example, holds in general: for any number of inputs and outputs, a given set of no-influence relations is equivalent to the existence of a decomposition along a given graph, in which these no-influence relations are made obvious. So far, no counter-example to this conjecture has been found, and it has been proven in numerous cases; yet, in some of these cases, the compositional structure had to be expressed by going beyond standard quantum circuits.

The paradigmatic example of this is the case of unitary channels with three inputs $A_I$, $E_I$, $B_I$ and three outputs $A_O$, $E_O$, $B_O$, obeying the no-influence relations $A_I \not\to B_O$ and $B_I \not\to A_O$. The causal decomposition corresponding to this pair of no-influence relations should have the following graph:

\be \label{Diamond} %
\InputIfFileExists{Diamond1.tikz}{}{\input{./figures/Diamond1.tikz}} \,  = \, %
\InputIfFileExists{Diamond2.tikz}{}{\input{./figures/Diamond2.tikz}} \, . \ee

However, if we take this graph to represent a standard unitary quantum circuit (i.e.\ if we interpret each box as a unitary channel), it is not equivalent to the causal structure: there exist unitary channels which obey the two previously mentioned no-influence relations, but which cannot be decomposed in the form of (\ref{Diamond}) \correction{(see Ref.\,\cite{lorenz2020} for an example)}. Yet there exists a slightly different kind of circuit for which the theorem holds \cite{lorenz2020}, namely:

\be \label{eq:DiamondDot} %
\InputIfFileExists{Diamond1.tikz}{}{\input{./figures/Diamond1.tikz}} \,  = \, %
\InputIfFileExists{Diamond3.tikz}{}{\input{./figures/Diamond3.tikz}} \, . \ee

The decomposition on the right-hand side is written in terms of \textit{index-matching quantum circuits}, with distinctive `$k$' superscripts written on some of the wires. Such diagrams were first introduced in Ref.\ \cite{lorenz2020}, where they were called `extended circuit diagrams'. In the present work, we shall adopt the more specific name `index-matching quantum circuits', and use somewhat different notation to that in the original work \cite{lorenz2020}\footnote{In comparison with Ref.\ \cite{lorenz2020}, we drop the practice of writing superscripts in boxes, writing them only on the wires. Our presentation will also not include `nested indices', i.e.\ superscripts of the form $k l_k$. This is because such nested indices are unnecessary for our needs.}. The superscripts can be interpreted in the following way (where this is again slightly different from the presentation in Ref.\ \cite{lorenz2020}). There exists a finite set of indices $K$; the Hilbert space $\ch_R$ admits a partition into orthogonal subspaces $\ch_R^k$, with indices $k \in K$; so do the other wires with `$k$' superscripts. $U_3$ is a unitary map from $\ch_R \otimes \ch_{B_I}$ to $\ch_{R'} \otimes \ch_{B_O}$, which satisfies the additional requirement that it preserves each of the sectors, in the sense that for any $k$, one has $U_3 (\ch_R^k \otimes \ch_{B_I}) = \ch_{R'}^k \otimes \ch_{B_O}$; a similar condition holds for $U_2$ (to put it in simple terms, $U_2$ and $U_3$ are block-diagonal). $U_1$ is a unitary map from $\ch_{E_I}$ to $\bigoplus_{k \in K} \ch_L^k \otimes \ch_R^k$, which is a subspace of $\ch_L \otimes \ch_R = (\bigoplus_{k \in K} \ch_L^k) \otimes (\bigoplus_{l \in K} \ch_R^l) $; similarly, $U_4$ has input space $\bigoplus_{k \in K} \ch_{L'}^k \otimes \ch_{R'}^k$.

Any unitary satisfying $A_I \not\to B_O$ and $B_I \not\to A_O$ admits a decomposition of the form (\ref{eq:DiamondDot}). It is therefore natural to argue that index-matching quantum circuits are the proper tool to study causal decompositions. Indeed, beyond this paradigmatic example, several other cases of causal decompositions have been proven which involve index-matching quantum circuits, and do not hold if one limits oneself to standard quantum circuits (see Ref.\ \cite{lorenz2020} for an overview). As in the previous example, the inability of standard quantum circuits to represent these compositional structures is due to their incapacity to encode sectorial correlations and constraints; the repetition of indices in index-matching quantum circuits is precisely used as a graphical depiction of such correlations and constraints.

In Section \ref{sec:IndexMatching}, we will show how index-matching quantum circuits can be seen as a handy diagrammatic representation of a sub-framework of routed maps: the framework of \textit{index-matching routed maps}. This will allow us to properly formalise such diagrams in full generality, to provide them with rigorous semantics, and to prove that simple rules single out the physically meaningful index-matching quantum circuits, such as the one above.

\section{Routed maps for pure quantum theory} \label{sec:PureRoutedMaps}

\subsection{Routed linear maps} \label{sec:routedlinearmaps}

We will first present \textit{routed linear maps} (or routed maps for short), which can be used to describe pure quantum theory (i.e.\ for pure states and isometric channels). Routed maps come from the introduction of routes, which are mathematical objects encoding sectorial constraints on linear maps. To talk about sectorial constraints, one first has to work with spaces which are formally partitioned into sectors.

We want to define a \textit{partitioned space} as a (finite-dimensional) Hilbert space $\ch_A$ partitioned into orthogonal subspaces, i.e.\ $\ch_A = \bigoplus_{k \in \cz_A} \ch_A^k$, with $\cz_A$ a finite set. A practical way to define partitions is the following: we say that a family $(\pi^k_A)_{k \in \cz}$ of orthogonal projectors on a finite-dimensional Hilbert space $\ch_A$ defines an orthogonal partition of $\ch_A$ if $ \forall k,l, \, \pi^k_A \circ \pi^l_A = \delta^{kl} \pi^k_A$ and $\sum_k \pi^k_A = \tr{id}_A $. The length of this partition is the size of $\cz_A$.

A partitioned space can then be formally defined:
\begin{definition}[Partitioned Hilbert space]
A (finite-dimensional) partitioned Hilbert space is a tuple  $(\ch_A, \cz_A, (\pi^k_A)_{k \in \cz_A})$ where $\ch_A$ is a finite-dimensional Hilbert space, $\cz_A$ is a finite set of indices, and $(\pi^k_A)_{k \in \cz_A}$ is an orthogonal partition of $\ch_A$.
\end{definition}
We will colloquially refer to such a partitioned space by the name $A^k$. It is important to stress that in the framework we are building, partitions are hardcoded: formally, two different partitions of the same space give rise to two different partitioned spaces, which should not be confused with each other. From two partitioned spaces $A^k = (\ch_A, \cz_A, (\pi^k_A)_{k \in \cz_A})$ and $B^l = (\ch_B, \cz_B, (\mu^l_B)_{l \in \cz_B})$, one can form their tensor product $A^k \otimes B^l := (\ch_A \otimes \ch_B, \cz_A \times \cz_B, (\pi^k_A \otimes \mu^l_B)_{(k,l) \in \cz_A \times \cz_B})$.

Given two partitioned spaces $A^k$ and $B^l$, sectorial constraints on linear maps $A^k \to B^l$ can be encoded by \textit{relations}. Relations are a way of modelling connections between elements of two sets; they can be thought of as generalisations of functions, in which a given element of the input set can be mapped to any number of elements of the output set (including the possibility that it is mapped to no element at all) \cite{maclane71}. For example, $\lambda$, as represented by the following graph, is a relation from $\cz_A$ to $\cz_B$.

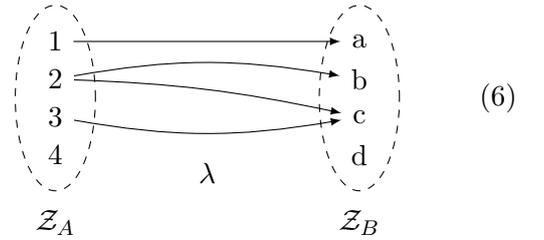
\be \label{eq:relationGraph1} %
\InputIfFileExists{RelationGraph1.tikz}{}{\input{./figures/RelationGraph1.tikz}} \ee

The fact that $\lambda$ relates $k \in \cz_A$ to $l \in \cz_B$ is denoted $k \overset{\lambda}{\sim} l$: here, for instance, $2 \overset{\lambda}{\sim} b$, $2 \overset{\lambda}{\sim} c$, etc. Relations can be sequentially composed in a natural way, following the rule that two elements are related if there exists at least one path between them:

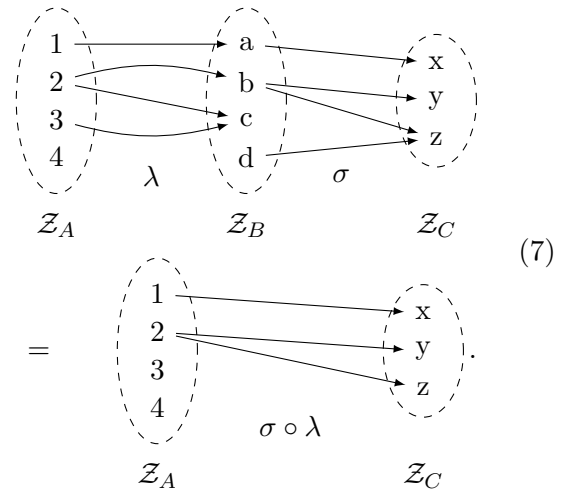
\be \begin{split}
\InputIfFileExists{RelationGraph2.tikz}{}{\input{./figures/RelationGraph2.tikz}} \\
    = \quad \quad %
\InputIfFileExists{RelationGraph3.tikz}{}{\input{./figures/RelationGraph3.tikz}} \, .
\end{split} \ee

Relations can be parallelly composed using cartesian products, with the rule: $(k_1,k_2) \overset{\lambda_1 \times \lambda_2}{\sim}(l_1,l_2)$ if and only if $k_1 \overset{\lambda_1}{\sim} l_1$ and $k_2 \overset{\lambda_2}{\sim} l_2$. Finally, from a relation $\lambda : \cz_A \to \cz_B$, one can define the opposite relation $\lambda^\top: \cz_B \to \cz_A$, given by reversing the arrows in $\lambda$'s graph. A relation $\lambda$ can equivalently be represented by a matrix\footnote{To improve clarity, we will follow the convention of writing input indices in subscript and output indices in superscript.} $(\lambda^l_k)_{k \in \cz_A, l \in \cz_B}$ with coefficients in the semiring of booleans: $\lambda^l_k= 1$ if $k \overset{\lambda}{\sim} l$, and 0 otherwise. In this picture, sequential composition is given by matrix products, parallel composition by tensor products of matrices, and taking the opposite relation corresponds to matrix transposition. In the rest of this paper, we will predominantly work with the representation of relations as boolean matrices, and refer to a relation $\lambda$ by its boolean components $\lambda^l_k$.

A route is a relation which represents a set of sectorial constraints: the constraint that a sector $\ch_A^k$ of the input space is forbidden from being connected to a sector $\ch_B^l$ of the output space will be denoted by the fact that $\lambda$ features no arrow from $k$ to $l$, or equivalently, $\lambda_k^l = 0$. For example, if we take the partitioned Hilbert spaces $\ch_A = \ch_A^0 \oplus \ch_A^1$ and $\ch_B = \ch_B^0 \oplus \ch_B^1$, the constraint that a linear map $f: \ch_A \to \ch_B$ satisfies $f(\ch_A^1) \subseteq \ch_B^1$ (i.e.\ $f$ does not connect $\ch_A^1$ to $\ch_B^0$) will be represented by the route $\lambda = \begin{pmatrix} 1 & 0\\ 1 & 1 \end{pmatrix}$. Formally, we have:

\begin{definition}[Routes]
Let $(\ch_A, \cz_A, (\pi^k_A)_{k \in \cz_A})$ and $(\ch_B, \cz_B, (\mu^l_B)_{l \in \cz_B})$ be two partitioned spaces, and $\lambda : \cz_A \to \cz_B$ a relation. A linear map $f : \ch_A \to \ch_B$ \textbf{follows the route} $\lambda$ if

\be \label{eq:RoutedCondition} f = \sum_{lk} \lambda^l_{k} \, \cdot \, \mu^l_B \circ f \circ \pi^k_A \, .\ee
One also says that $\lambda$ \textbf{is a route} for $f$.
\end{definition}
An equivalent condition to (\ref{eq:RoutedCondition}), proven in Appendix \ref{app:proofRoutedMapsSMC}, is
\be \forall k,l, \, \lambda^l_k = 0 \implies  \mu_B^l \circ f \circ \pi_A^k = 0 \,. \ee 

This yields an intuitive interpretation of routes: the 1's in the matrix of a route $\lambda$ can be thought of as designating the blocks that are allowed to be non-zero in the block decomposition of $f$.

On partitioned Hilbert spaces, compatibility of linear maps with routes plays well with sequential composition, parallel composition, and hermitian adjoints: if $f$ follows $\lambda$ and $g$ follows $\sigma$, then $g \circ f$ follows $\sigma \circ \lambda$, $f \otimes g$ follows $\lambda \times \sigma$, and $f^\dagger$ follows $\lambda^\top$ (this is proven in Appendix \ref{app:proofRoutedMapsSMC}).

We are now in a position to define routed maps, which, together  with partitioned spaces, form the basic components of our framework:

\begin{definition}[Routed maps]
A \textbf{routed linear map} (or routed map for short) from $ A^k$ to $B^l$ is a pair $(\lambda, f)$ where $\lambda: \cz_A \to \cz_B$ is a relation, and $f: \ch_A \to \ch_B$ is a linear map which follows $\lambda$.
\end{definition}

For example, if we go back to the scenario described in Section \ref{sec:CommSuperpos} (restricting for now to the pure version, in which all operations are unitary operators), Alice and Bob's wires and operations can be modelled, respectively, as partitioned spaces of the form $\ch = \ch^0 \oplus \ch^1$, and as routed maps of the form $(\delta, V)$, where $\delta = \begin{pmatrix} 1 & 0\\ 0 & 1 \end{pmatrix}$, and $V$ is a unitary map (which, by definition, has to follow the route $\delta$).

Routed maps themselves can be sequentially and parallelly composed, through pairwise composition in both cases\footnote{Note that, because partitions are hardcoded into partitioned spaces, one can sequentially compose two routed maps only if the first map's output space is equal to the second map's input space, \textit{including the partition}.}, and one can take their hermitian adjoint: $(\lambda, f)^\dagger := (\lambda^\top, f^\dagger)$. These features are all encompassed formally by Theorem \ref{th:RoutedMaps=dagSMC}, stated and proven in Appendix \ref{app:proofRoutedMapsSMC}, which, more generally, characterises the framework of partitioned Hilbert spaces and routed linear maps as a $\dagger$-symmetric monoidal category. \correction{This has another important practical consequence: routed maps can be represented graphically in a well-defined way using circuit diagrams.} \correction{Circuit diagrams where the wires are interpreted as partitioned Hilbert spaces and the boxes are interpreted as routed linear maps shall be referred to as routed quantum circuits, or simply, routed  circuits.}

Formally, we have:

\begin{theorem}\label{th:PureLinearMapsSoundForCircuits}
The framework of partitioned Hilbert spaces and routed linear maps admits a sound representation in terms of circuit diagrams. 
\end{theorem}

\begin{proof}
This follows from Theorem \ref{th:RoutedMaps=dagSMC}, of Appendix \ref{app:proofRoutedMapsSMC}, and the fact that symmetric monoidal categories are suitable for diagrammatic representation in terms of circuit diagrams \cite{joyal1991, coecke_kissinger_2017}.
\end{proof}

\correction{Soundness in Theorem \ref{th:PureLinearMapsSoundForCircuits} means the following. We take two circuit diagrams whose wires are interpreted as partitioned Hilbert spaces and whose boxes are interpreted as routed linear maps; if these are provably equal as diagrams (i.e. if one can be obtained by simply deforming the other), then the routed linear maps they represent are provably equal \cite{coecke_kissinger_2017}. Essentially, this should be understood as ensuring that circuit diagrams can be used without second thoughts when dealing with routed linear maps, just as they could be used without second thoughts when dealing with standard linear maps.}

\correction{Note that in Appendix \ref{app:conceptual}, we provide additional comments on how the `route' part of a routed map can be understood at the conceptual level.}

\subsection{Practical isometries}

An important task is to single out those routed maps which correspond to physical evolutions. In the standard framework of quantum circuits for pure quantum theory, they are given by isometries; for routed maps, we will instead coin the notion of \textit{practical isometries}. The reasons why this notion is necessary can be understood by considering two examples.

A first example shows that the naive notion of isometricity cannot be directly applied to routed maps. A map $U$ is an isometry if $U^\dag \circ U = \textrm{id}$. But there can be routed maps $(\lambda, U)$ such that $(\lambda, U)^\dagger \circ (\lambda, U) = (\lambda^\top \circ \lambda, U^\dagger \circ U) = (\lambda^\top \circ \lambda, \textrm{id}_{\ch_A}) \neq (\textrm{id}_{\cz_A}, \textrm{id}_{\ch_A})$: i.e., $U$ is an isometry but $\lambda$ is not. This will be, for example, the generic case for routed maps which `delete' an index, i.e.\ when $\lambda$ is a relation with the singleton as output set. As the latter maps are necessary components of our framework, we will have to relax the condition of isometricity, by not imposing it on the `route' part of a routed map.

There is a second fact to take into consideration. For a given routed map $(\lambda,U) : A^k \to B^l$, there can be elements of $\cz_A$ that $\lambda$ relates to no element of $\cz_B$ at all (this is for example the case of `4' in (\ref{eq:relationGraph1})). The fact that $U$ follows $\lambda$ then entails that it is null on the corresponding sectors of $\ch_A$. Thus we should only ask $U$ to be an isometry when restricted to the sectors corresponding to the other elements of $\cz_A$, those that $\lambda$ relates to at least one element of $\cz_B$. This yields the notion of \textit{practical isometries}.

For a routed map $(\lambda,U) : A^k \to B^l$, we define the \textit{practical input set} of $\lambda$, $\cs_\lambda$, as the subset of $\cz_A$ whose elements are related by $\lambda$ to at least one element of $\cz_B$ (for example, for $\lambda$ as defined in (\ref{eq:relationGraph1}), we have $\cs_\lambda = \{1,2,3\})$. We define the corresponding \textit{practical input space} of $(\lambda,U) $ as $\ch_A^{\cs_\lambda}:= \bigoplus_{k \in \cs_\lambda} \ch_A^k$. The symmetric notions of practical output set $\ct_\lambda$ and practical output space $\ch_B^{\ct_\lambda}$ are defined in the same way. This enables us to define practical isometries and practical unitaries:

\begin{definition}
Let $(\lambda, U)$ be a routed map from $A^k$ to $B^l$, with practical input space $\ch^{\cs_\lambda}_A$. $(\lambda, U)$ is a \textbf{practical isometry} if $U$ is an isometry when restricted to $\ch^{\cs_\lambda}_A$.
$(\lambda, U)$ is a \textbf{practical unitary} if both $(\lambda, U)$ and $(\lambda, U)^\dagger$ are practical isometries.
\end{definition}

Equivalently, $(\lambda, U)$ is a {practical isometry} if $U$ is a partial isometry with initial domain $\ch^{\cs_\lambda}_A$. Similarly, $(\lambda, U)$ is a {practical unitary}  if $U$ a partial isometry with initial domain $\ch^{\cs_\lambda}_A$ and with range $\ch_B^{\ct_\lambda}$.

In the routed maps framework for pure quantum theory, the physically meaningful routed maps are the practical isometries. One, however, has to be careful on one point: the sequential composition of two practical isometries is not necessarily a practical isometry itself. The badly-behaved compositions of practical isometries correspond to situations in which we have complete descriptions for two individual gates, but where these descriptions are not sufficient to specify a complete description of the sequential composition of these two gates. Parallel compositions of practical isometries, on the other hand, always yield practical isometries, as proven in Appendix \ref{app:practicalIsos}.

We therefore need to single out the sequential compositions which are well-behaved for practical isometries. We can do so with a condition which, crucially, depends solely on the maps' routes:

\begin{theorem} \label{th:IsoCompos}
Let $\lambda: \cz_A \to \cz_B$ and $\sigma: \cz_B \to \cz_C$  be two routes satisfying 

\be \label{eq:IsoComposition} (\lambda \circ \lambda^\top) [\cs_\sigma] \subseteq \cs_\sigma \, . \ee
Then, for any practical isometries $(\lambda, U): A^k \to B^l$ and $(\sigma, V) : B^l \to C^m$, their composition $(\sigma, V) \circ (\lambda, U)$ is a practical isometry.
\end{theorem}

\begin{proof}
See Appendix \ref{app:practicalIsos}.
\end{proof}

When the condition \eqref{eq:IsoComposition} is satisfied, we say that the sequential composition of $\lambda$ and $\sigma$ is \textit{proper for practical isometries}.

Theorem \ref{th:IsoCompos} and its forthcoming generalisation to quantum channels are the crucial consistency theorems for routed quantum circuits.  In routed quantum circuits made of practical isometries, all sequential compositions have to satisfy (\ref{eq:IsoComposition}), in order to ensure that the global map they form is also a practical isometry. For the case of practical unitaries, sequential composition is well-behaved if the routes satisfy both (\ref{eq:IsoComposition}) and a symmetric condition:

\begin{theorem} \label{th:UniCompos}
Let $\lambda: \cz_A \to \cz_B$ and $\sigma: \cz_B \to \cz_C$  be two routes satisfying 

\begin{subequations}
\be (\lambda \circ \lambda^\top) [\cs_\sigma] \subseteq \cs_\sigma \, , \ee
\be (\sigma^\top \circ \sigma) [\ct_\lambda] \subseteq \ct_\lambda \, . \ee
\end{subequations}
Then, for any practical unitaries $(\lambda, U): A^k \to B^l$ and $(\sigma, V) : B^l \to C^m$, their composition $(\sigma, V) \circ (\lambda, U)$ is a practical unitary.
\end{theorem}

Under these conditions, we say that the sequential composition of $\lambda$ and $\sigma$ is \textit{proper for unitaries}.

\subsection{An example: superposition of two trajectories} \label{sec:example2Superp}

Before we move on to general quantum channels, let us provide a first didactic example of a routed circuit, by showing how our framework allows to properly formalise the scenario of a superposition of two trajectories described in Section $\ref{sec:CommSuperpos}$ -- restricting for now to the case of unitary channels. Diagrammatically, we represent partitioned spaces as wires, and routed maps as boxes: we write the linear map within the box, and depict the matrix elements of its route floating next to it. Spaces with a trivial partition are denoted without superscripts. The scenario is then formalised as the following \correction{routed} circuit:

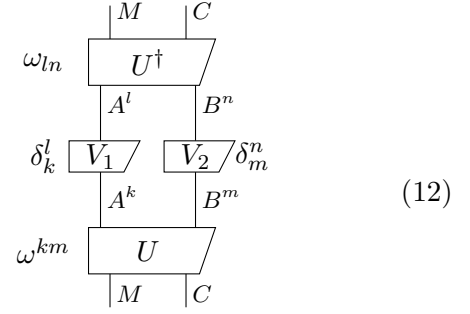
\be \label{eq:ChannelSuperposition} %
\InputIfFileExists{ChannelSuperposition.tikz}{}{\input{./figures/ChannelSuperposition.tikz}}\ee
where $\omega = \begin{pmatrix} 0 & 1 &  1 & 0 \end{pmatrix}$, and $\delta =  \begin{pmatrix} 1 & 0\\ 0 & 1 \end{pmatrix}$.

As we can see, Alice and Bob's channels are now represented as routed maps of the form ($\delta, V$), in which the sectorial constraints inherent to the scenario have been included. Let us also carefully break down the meaning of the encoding operation, ($\omega, U$). As the input wires of this operation are not partitioned, $\omega$ is a relation from a trivial set of indices -- represented by the singleton $\{*\}$ -- to the set $\cz_A \times \cz_B = \{(0,0),(0,1),(1,0),(1,1)\}$. $\omega$'s graph is

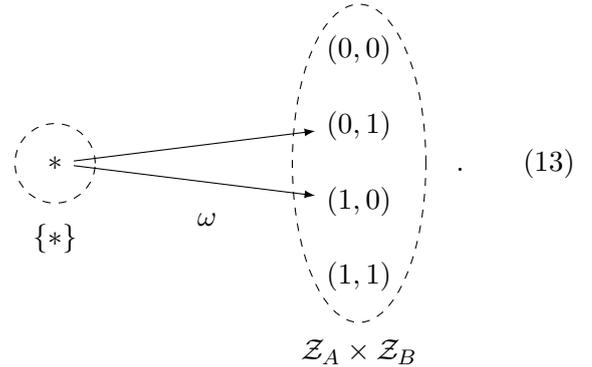
\be %
\InputIfFileExists{omegaGraph.tikz}{}{\input{./figures/omegaGraph.tikz}} \quad . \ee

As we can see, $\omega$ is characterised by its practical output set $\ct_\omega = \{(0,1),(1,0)\}$. $(\omega, U)$'s practical output space, $\ch_{AB}^{\ct_\omega}$, is thus precisely the $\widetilde{\ch}_{AB}$ introduced in section \ref{sec:CommSuperpos}: this enforces the sectorial correlations present in the scenario. It is important to stress that within the framework of routed maps, these sectorial correlations are obtained as a contextual feature \correction{(see footnote \ref{foot:contextual} about our use of the word `contextual')}: they are specified by the routes present in the global \correction{routed} circuit. This point is elaborated upon in Section \ref{sec:accessiblespace}.

All the routed maps in this \correction{routed} circuit are practical unitaries, and it is easy to check that all the compositions are proper for practical unitaries. This means that the framework of routed maps allowed us to provide a suitable unitary description of this unitary scenario -- something which, as we argued in section \ref{sec:CommSuperpos}, is not possible within the sole framework of \correction{standard} quantum circuits. More generally, it is straightforward to see that the use of routed maps solves all the points of dissatisfaction we had with the description of this scenario in standard quantum circuits.

One can reduce graphical clutter in the above \correction{routed} circuit by using a handy graphical convention, that of contracting Kronecker deltas. The idea is, not to write explicitly the $\delta$-routes $\delta^l_k$ between a wire $A^k$ and a wire $A^l$, and to instead just write down these wires with the same superscript\footnote{Contraction of Kronecker deltas can sometimes lead to ambiguities about the routes in a given diagram (this will be elaborated upon in Section \ref{sec:IndexMatching}); one should therefore keep in mind that it is simply a graphical shorthand, and that the rigorous diagrammatic representation of routed quantum circuits is properly done through explicitly writing down all the routes.}. The above \correction{routed} circuit, for example, then becomes:

\be %
\InputIfFileExists{ChannelSuperposition2.tikz}{}{\input{./figures/ChannelSuperposition2.tikz}}\ee

\section{Diagrammatic representation: routed \correction{circuits}} \label{sec:RoutedDiagrams}

Let us provide a thorough description of the diagrammatic representation that we just introduced in a simple example. Indeed, one of the objectives of the introduction of routed maps is to use them as a mathematical basis for a faithful and systematic diagrammatic representation, in which the routes can be read in an intuitive way. We provide such a representation in the form of so-called \textit{routed \correction{circuits}}. The well-definition of these diagrams as a faithful representation of the mathematical framework is guaranteed by Theorem 1. A specific challenge in this context is to give clear rules on the physical interpretation of slices in a given routed \correction{circuit}: we will do so by establishing a distinction between \textit{formal space} and \textit{accessible space}.

\subsection{Diagrams for routed linear maps}\label{sec:routediagrams}

An example of a routed \correction{circuit} is (\ref{eq:ChannelSuperposition}). We represent the objects $X^k := (\ch_X, \cz_X,(\pi^k_X)_{k \in \cz_X})$ by wires. We represent the morphisms $(\lambda,f)$ by boxes: we write the linear maps $f$ inside the boxes, whereas the matrix elements of the route, $\lambda_{k,\dots, l}^{m,\dots, n}$, are drawn as numbers floating next to the boxes. In general, we ask that no two wires bear the same superscript (except in the case of the shorthand notation given by contractions of Kronecker deltas).
When  $|\cz_X|=1$, we simply write $X$ in place of $X^k$.

Remember that a diagram composed of practical isometries represents a practical isometry itself if and only if the sequential compositions of routes in this diagram are suitable for isometries, as per Theorem \ref{th:IsoCompos}. In this case, the diagram is called an iso-diagram. In the same way, a diagram in which sequential compositions of routes are suitable for unitaries is called a uni-diagram.

\subsection{How to interpret slices}

An important question is that of the interpretation of slices in a diagram. By slices, we mean horizontal combinations of wires. For example, if we take the slice comprising wires $A^k$ and $B^m$ in (\ref{eq:ChannelSuperposition}), a simple formulation of the question at hand would be: `What is the Hilbert space corresponding to this slice?'. As we will show, the answer depends on whether one is asking from a mathematical or physical perspective. This will lead us to distinguishing two spaces corresponding to a slice: the formal space, and the accessible space.

A first possible answer comes from strictly sticking to the mathematical formalism. As is clear from its definition in Section \ref{sec:routedlinearmaps}, the tensor product $A^k \otimes B^m$ has Hilbert space $\ch_A \otimes \ch_B = \bigoplus_{k,m} \ch_A^k \otimes \ch_B^m$. We will define this as the \textit{formal space} corresponding to the slice. In contrast to what will come later, the formal space corresponding to a slice is non-contextual\footnote{\label{foot:contextual}Here, we use the word `contextual' in a colloquial sense; this should not be confused with its use in discussions of non-contextuality as a quantum feature, in which `contextuality' has a different, more technical meaning.}, in the sense that it only depends on the slice itself. As we can see, the formal space is the `big' Hilbert space: it contains all the sectors of the partition.

It is clear, however, that some of these sectors are in general forbidden from being populated, due to the sectorial constraints imposed by the routes. Thus, in the context of the diagram at hand, one can give a more refined meaning to the slice. This is formalised by the notion of the \textit{accessible space} corresponding to a slice: we define it as the subspace of the formal space in which states will be constrained to lie due to the routes. The accessible space corresponds to a more physical understanding of the situation, and encodes physical correlations between sectors in each of the wires which compose it. For example, the accessible space corresponding to the slice comprising wires $A^k$ and $B^m$ in (\ref{eq:ChannelSuperposition}) is $\ch_A^1 \otimes \ch_B^0 \oplus \ch_A^0 \otimes \ch_B^1$, a strict subspace of its formal space. We give the general recipe for computing the accessible Hilbert space corresponding to a slice in section \ref{sec:accessiblespace}.

Before that, it is important to emphasise that the accessible space is a \textit{contextual} notion: it depends on the whole diagram (more specifically, on the routes thereof) and not only on the slice itself. This somewhat counter-intuitive feature should not come as a surprise. To see why, it is enlightening to take the view in which routed \correction{circuits} are understood as representing a physical setup and the sectorial constraints this setup implies (this is for example the case in (\ref{eq:ChannelSuperposition})). In this context, the whole point of the notion of accessible space is to take into account the fact that some setups lead to only a subspace of a given formal space being populated.  It is thus natural that the whole setup should be taken into account when computing the accessible space.

That the notion of accessible space is a contextual one entails another important consequence: the accessible space of a given slice can get  modified (and, more specifically, reduced) when additional maps are adjoined to a diagram. For example, if one considers a diagram containing only the middle layer of (\ref{eq:ChannelSuperposition}), then the accessible space of the slice comprising wires $A^k$ and $B^m$ is equal to its formal space; but when one adjoins the other layers to recover the diagram above, this accessible space gets reduced to $\ch_A^1 \otimes \ch_B^0 \oplus \ch_A^0 \otimes \ch_B^1$. Once again, this is in fact natural: in the physical interpretation of routed maps, adding more routed maps means specifying a setup further - which could mean that we are adding new constraints on the possible physical correlations in a given slice\footnote{There is a specific case of interest, however, in which the previous comment will not apply: when one is considering a diagram whose global input and output wires (the inputs and outputs of the whole diagram) are not partitioned (i.e.\ bear no indices).
It is easy to see that adjoining more  routed maps to such a diagram will not modify the accessible space of a given slice inside it; one can thus consider the accessible space to be the exact subspace which will be populated. 
For example, this is the case for the slice comprising wires $A^k$ and $B^m$ in (\ref{eq:ChannelSuperposition}), as the global inputs and outputs in (\ref{eq:ChannelSuperposition}) bear no indices.}.

\subsection{Computing the accessible space corresponding to a slice}  \label{sec:accessiblespace}

An example of a more elaborate routed \correction{circuit}, with slices explicitly drawn out, is given below; we will use this as an example to illustrate the general procedure which yields the accessible space corresponding to a slice. Two possible slices are drawn in blue and magenta, and their respective accessible Hilbert spaces are written out on the side:

\be %
\InputIfFileExists{slices1a.tikz}{}{\input{./figures/slices1a.tikz}} \ee

One can compute the accessible space corresponding to a given slice with the following formal procedure (using the blue slice above as an example). We call $\ck$ the set of indices present in this slice [e.g.\ $\ck = \{ i,j \}$]. \correction{This procedure is justified more formally in Appendix \ref{app:accessibleSpace}.}

\begin{enumerate}
	\item write down the matrix components of all the routes featuring an index in $\ck$ [e.g.\ $\delta^i_m \gamma^j_n$];
	\item write down the matrix components of all the routes (both above and below the slice under consideration\footnote{It might sound surprising that the routes \textit{above} (i.e.\ after) the slice should be taken into account as well, but this is in fact necessary: the consistency of the process formed by the whole diagram forces one to restrict the states which can populate a given slice so that they do not lie out of the practical input space of a subsequent routed map.}) featuring an index among those already present [e.g.\ $\lambda^{mn}_k$];
	\item iterate until there are no matrix components left to add under the previous rule;
	\item sum over all indices present, except the ones in $\ck$; this yields the components of a boolean matrix $\eta$ with indices in $\ck$ [e.g.\ $\eta^{ij} := \sum_{mnkp} \delta_m^i \gamma_n^j \lambda_{k}^{mn} \alpha^k_p$]
	\item write down the explicit direct sum of all the sectors in the slice with $\eta$ [e.g.\ $ \bigoplus_{i,j}
	\eta^{ij} \ch_E^i \otimes \ch_F^j$]: this yields the accessible Hilbert space corresponding to the slice.
\end{enumerate}

As one can see, this displays formal similarities with the Einstein summation convention of linear algebra. One should not take these similarities too seriously, however, as some conventions are different. In particular, in the formula for the accessible space, the indices in $\ck$ appear three times, and indices born by input or output wires of the diagram appear one time yet are still summed over.

\section{Routed maps for mixed quantum theory} \label{sec:MixedRoutedMaps}

\subsection{Routed quantum channels}

Let us now show how to extend the theory of routed maps to encompass \correction{mixed states and general quantum channels}. In the same way that one goes from linear maps to completely positive linear maps, we will be going from relations to completely positive relations. These will be used to encode sectorial constraints which can not only forbid connections between some sectors, but also forbid some of the allowed connections to be coherent with one another.

The broad idea is to generalise (\ref{eq:RoutedCondition}) to the case where we take a completely positive linear map $\cc: \cl\left(\bigoplus_{k \in \cz_A} \ch_A^k \right) \to \cl\left(\bigoplus_{l \in \cz_B} \ch_B^l\right)$ between linear operators on partitioned Hilbert spaces. A natural way to do so is to use relations of the form $\Lambda: \cz_A \times \cz_A \to \cz_B \times \cz_B$ as routes, leading to the following definition:

\begin{definition}[Routes for CPMs]
Let $(\ch_A, \cz_A, (\pi^k_A)_{k \in \cz_A})$ and $(\ch_B, \cz_B, (\mu^l_B)_{l \in \cz_B})$ be two partitioned spaces, and $\Lambda: \cz_A \times \cz_A \to \cz_B \times \cz_B$ a relation. 
A completely positive map $\cc: \cl\left(\ch_A \right) \to \cl\left( \ch_B\right)$  \textbf{follows} $\Lambda$ if

\be \label{eq:FollowingCPM} \forall \rho, \quad  \mathcal{C}(\rho) =  \sum_{ll'kk'} \Lambda^{ll'}_{kk'}    \,\cdot\, \mu_B^l \circ \mathcal{C} \left( \pi_A^k \circ \rho \circ \pi_A^{k'} \right) \circ \mu_B^{l'} \, .\ee
One also says that $\Lambda$ \textbf{is a route} for $\cc$.
\end{definition}

However, we need not use all of the possible $\Lambda$'s; some are superfluous. Take, for instance, a $\Lambda$ which is not symmetric, in the sense that, for some $k,k',l,l'$, $1 = \Lambda_{kk'}^{ll'} \neq \Lambda_{k'k}^{l'l} = 0$. As completely positive maps are symmetric, any $\cc$ following $\Lambda$ will also follow $\tilde{\Lambda}$ defined from $\Lambda$ by setting $\tilde{\Lambda}_{kk'}^{ll'}$ to 0. In other words, a non-symmetric $\Lambda$ expresses a set of constraints which could be expressed just as suitably by a symmetric one.

In addition, let us define the diagonal of $\Lambda$ as the relation $\dot{\Lambda} : \cz_A \to \cz_B$ defined by $\dot{\Lambda}_k^l := \Lambda_{kk}^{ll}$. It is easy to see that if there exist $k,l$ such that $\dot{\Lambda}_k^l = 0$, then any completely positive $\cc$ that follows $\Lambda$ will also follow the route $\tilde{\Lambda}$ obtained from the former by setting  $\forall k',l', \, \tilde{\Lambda}_{kk'}^{ll'} = \tilde{\Lambda}_{k'k}^{l'l} = 0$. Let us define diagonally dominant relations $\Lambda$ as the ones satisfying for any $k,l$, $\dot{\Lambda}_{k}^{l} = 0 \implies \forall k', l', \, \,\Lambda_{kk'}^{ll'} = \Lambda_{k'k}^{l'l} = 0$. This entails that a non diagonally dominant $\Lambda$ expresses a set of constraints which could be expressed just as suitably by a diagonally dominant one. One can thus, without loss of generality, work only with symmetric and diagonally dominant $\Lambda$'s.

Remarkably, the symmetric and diagonally dominant $\Lambda$'s can be recovered in another way: they are exactly the \textit{completely positive relations} that one can obtain by mimicking, on relations, the procedure that leads from linear maps to completely positive linear maps, through `doubling then tracing out'. Indeed, one of the several equivalent definitions of completely positive linear maps is the following: $\cc : \cl(\ch_A) \to \cl(\ch_B)$ is a completely positive linear map if and only if it is of the form $\cc: \rho \mapsto \Tr_E (M \rho M^\dagger)$, where $\ch_E$ is an auxiliary Hilbert space and $M: \ch_A \to \ch_B \otimes \ch_E$ is a linear map \cite{coecke_kissinger_2017}. If, in an analogous way, we say that $\Lambda: \cz_A \times \cz_A \to \cz_B \times \cz_B$ is completely positive if there exists a set $\cz_E$ and a relation $\lambda: \cz_A \to \cz_B \times \cz_E$ such that $\Lambda^{ll'}_{kk'}  = \sum_m \lambda^{lm}_k \lambda^{l'm}_{k'}$, then it can be found that a given $\Lambda$ is a completely positive relation if and only if it is symmetric and diagonally dominant\footnote{For a proof, see Ref.\ \cite{Mohindru2015}, Proposition 3.1.}. Sequential and parallel compositions of completely positive relations are completely positive relations\footnote{This follows directly from the universal CPM construction of Selinger \cite{Selinger2007}.}.

Completely positive relations will thus be used to express sectorial constraints for quantum channels, providing completely positive routes. The diagonal $\dot{\Lambda}_k^l$ of a completely positive route encodes constraints on whether a channel is allowed to connect the sectors $k$ and $l$; and the off-diagonal coefficients $\Lambda_{kk'}^{ll'}$ encode constraints on whether the connections between sectors $k$ and $l$ on the one hand, and sectors $k'$ and $l'$ on the other hand, are allowed to be coherent with each other (these will be called \textit{coherence constraints}). 

In analogy with the construction for the pure case, routed CPMs are defined as follows, with partitioned spaces colloquially written as $\cl\left(\bigoplus_{k \in \cz_A} \ch_A^k \right) = \bigoplus_{k,k' \in \cz_A} \cl\left( \ch_A^k \to \ch_A^{k'} \right) =: A^{kk'}$:

\begin{definition}[Routed CPMs]
A \textbf{routed completely positive map (CPM)} from $ A^{kk'}$ to $B^{ll'}$ is a pair $(\Lambda, \cc)$ where $\Lambda: \cz_A \times \cz_A \to \cz_B \times \cz_B$ is a completely positive relation, and $\cc: \cl\left(\bigoplus_{k \in \cz_A} \ch_A^k \right) \to \cl\left(\bigoplus_{l \in \cz_B} \ch_B^l\right)$ is a completely positive map which follows $\Lambda$.
\end{definition}

The framework of partitioned Hilbert spaces and routed CPMs satisfies the exact analogue of Theorem \ref{th:PureLinearMapsSoundForCircuits}:

\begin{theorem}\label{th:MixedLinearMapsSoundForCircuits}
The framework of partitioned Hilbert spaces and routed CPMs admits a sound representation in terms of circuit diagrams.
\end{theorem}

\begin{proof}
This follows from Theorem \ref{th:RoutedCPMs=dagSMC}, of Appendix \ref{app:proofRoutedCPMsSMC}, and from the fact that symmetric monoidal categories are suitable for diagrammatic representation in terms of circuit diagrams \cite{joyal1991, coecke_kissinger_2017}.
\end{proof}

\correction{Just as Theorem \ref{th:PureLinearMapsSoundForCircuits}, this should be understood as ensuring that circuit diagrams can be used without second thoughts when dealing with routed CPMs.}

\textit{Practically trace-preserving} routed CPMs are defined in the same way as practical isometries: 

\begin{definition}[Routed quantum channels]
A routed CPM $(\Lambda, \cc): A^{kk'} \to B^{ll'}$ is \textbf{practically trace-preserving} if it is trace-preserving when restricted to act on its practical input space $\cl(\ch_A^{\cs_{\dot{\Lambda}}})$, defined by the practical input set of $\Lambda$'s diagonal, $\cs_{\dot{\Lambda}}$. $(\Lambda, \cc)$ is then called a \textbf{routed quantum channel}.
\end{definition}

Finally, the condition for a composition of routed quantum channels to be proper (i.e., to always yield a routed quantum channel) is similar to that for practical isometries, and solely depends on their routes' diagonals: 

\begin{theorem} \label{th:IsoCompos_cpm}
Let $\Lambda: \cz_A \times \cz_A \to \cz_B \times \cz_B$ and $\Sigma: \cz_B \times \cz_B \to \cz_C \times \cz_C$ be two routes satisfying 

\be \label{eq:RQChannelCompos} (\dot{\Lambda} \circ \dot{\Lambda}^\top) [\cs_{\dot{\Lambda}}] \subseteq \cs_{\dot{\Sigma}} \, . \ee
Then, for any routed quantum channels $(\Lambda, \cc): A^{kk'} \to B^{ll'}$ and $(\Sigma, \ce) : B^{ll'} \to C^{mm'}$, their composition $(\Sigma, \ce) \circ (\Lambda, \cc)$ is a routed quantum channel.
\end{theorem}

The proof of this is similar to the one for practical isometries. When the condition \eqref{eq:RQChannelCompos} is satisfied, we say that the sequential composition of $\Lambda$ and $\Sigma$ is \emph{proper for routed quantum channels}.

\subsection{Link with Kraus representations}

A question of interest is whether the condition (\ref{eq:FollowingCPM}), expressing that a CP map $\cc$ follows a completely positive route $\Lambda$, can be translated in terms of the Kraus representations $\{K_i\}_i$ of $\cc$. A first answer is that the sectorial constraints expressed by $\Lambda$'s diagonal $\dot{\Lambda}$ have to be satisfied by each of the Kraus operators in any Kraus decomposition of $\cc$. This is also a sufficient condition when $\Lambda$ has \textit{full coherence}, i.e.\ when it is a route of the form $\Lambda_{kk'}^{ll'} = \dot{\Lambda}_k^l \dot{\Lambda}_{k'}^{l'}$ which includes no constraints on coherence:

\begin{theorem}\label{th:KrausFullCoherence}
Let $\Lambda: \cz_A \times \cz_A \to \cz_B \times \cz_B$ be a completely positive route, and $\cc: \cl \left( \bigoplus_{k \in \cz_A} \ch_A^k \right) \to \cl \left( \bigoplus_{l \in \cz_B} \ch_B^l \right)$ a completely positive linear map, with a Kraus representation given by the set of operators $\{K_i\}_i$, where $\forall i, K_i: \ch_A \to \ch_B$.

If $\cc$ follows $\Lambda$, then each of the $K_i$'s follow its diagonal $\dot{\Lambda}$. For a $\Lambda$ with full coherence, the reverse implication holds as well.
\end{theorem}

The other situation in which one can give conditions equivalent to (\ref{eq:FollowingCPM}) in terms of Kraus representations is the opposite extremal case: the one in which $\Lambda$ is a route with \textit{full decoherence}, i.e.\ is of the form $\Lambda_{kk'}^{ll'} = \delta_{kk'} \delta^{ll'} \dot{\Lambda}_k^l$. We say that a given Kraus decomposition $\{K_i\}_i$ is adapted to a completely positive route $\Lambda$ with full decoherence if for each $i$, there exists a unique pair $(k,l)$ such that $K_i = \mu_B^l \circ K_i \circ \pi_A^k$, i.e.\ $K_i$ only maps from $\ch_A^k$ to $\ch_B^l$.

\begin{theorem}\label{th:KrausFullDecoherence}
Let $\Lambda: \cz_A \times \cz_A \to \cz_B \times \cz_B$ be a completely positive route with full decoherence, and $\cc: \cl \left( \bigoplus_{k \in \cz_A} \ch_A^k \right) \to \cl \left( \bigoplus_{l \in \cz_B} \ch_B^l \right)$ a completely positive linear map. Then $\cc$ follows $\Lambda$ if and only if there exists a Kraus representation of $\cc$ adapted to $\Lambda$.
\end{theorem}

Theorems \ref{th:KrausFullCoherence} and \ref{th:KrausFullDecoherence} are proven in Appendix \ref{app:Kraus}.

\subsection{Diagrammatic representation}

Let us quickly elaborate on how the diagrammatic constructions and notions of Section \ref{sec:RoutedDiagrams} generalise to the case of mixed quantum theory. As exemplified in the next subsection, routed diagrams for \correction{mixed states and general} quantum channels are in simple analogy with those for \correction{pure states and isometries}: one simply switches to writing wires with doubled superscripts, of the form $A^{kk'}$, and to writing the completely positive routes with these same doubled indices. All compositions in such diagrams have to be suitable for routed quantum channels. The well-definition of these diagrams as a faithful representation of the mathematical framework is guaranteed by Theorem \ref{th:MixedLinearMapsSoundForCircuits}.

\subsection{Two examples: Superposition of three trajectories and decoherence of copied information}

We will exemplify the routed maps framework for general quantum channels \correction{with two examples} to show how completely positive routes can include constraints on the coherence between sectors, and how this leads to easy decoherence computations.

In order to also present a somewhat more involved use of the routed \correction{circuits} framework, let us extend the scenario we already formalised in Section \ref{sec:example2Superp}, and consider now the superposition of three trajectories. This scenario is the same as before, except that the control system is now a qutrit, and the message can now go in a superposition of three different channels $\ca$, $\cb$, and $\cc$, which once again preserve the number of particles \cite{chiribella2019shannon}.

Now that we are working with general quantum channels, there are in fact two different routes that one could attribute to $\ca$ (and, in the same way, to $\cb$ and $\cc$); the choice between them depends on the features of the physical scenario we want to describe. On the one hand, we could be asking only that $\ca$ preserves the number of particles, without setting constraints on the coherence between the vacuum sector and the one-particle sector. In this case, the route constraining $\ca$ will be $\delta^l_k \delta^{l'}_{k'}$. But one could also be considering a more restrictive situation, in which $\ca$ not only acts separately on the two sectors, but also acts incoherently on each \cite{chiribella2019shannon,kristjansson2020resources}. The choice of route corresponding to this situation is then $\delta^{ll'}_{kk'}$\footnote{The Kronecker delta here means that $\delta^{ll'}_{kk'}=1$ if and only if $l=l'=k=k'$, else $\delta^{ll'}_{kk'}=0$.}. The use of completely positive routes, therefore, allows to neatly distinguish between the two different scenarios.

Let us, for example, look at the scenario in which each channel is allowed to preserve coherence between the sectors. The routed circuit representing such a scenario (using contractions of Kronecker deltas) is then:

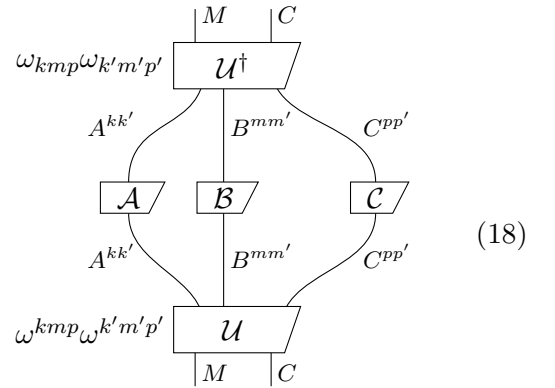
\be \label{eq:3ChannelSuperposition} %
\InputIfFileExists{ChannelSuperposition3.tikz}{}{\input{./figures/ChannelSuperposition3.tikz}} \ee
where, for given $k$, $m$ and $p$, $\omega^{kmp} = 1$ if and only if $k+m+p = 1$. Remember that, because of the convention of contracting Kronecker deltas, writing the names of $\ca$'s input and output wires with the same superscripts implicitly means that we are considering the routed quantum channel $(\delta^l_k \delta^{l'}_{k'}, \ca)$; the same goes with the routed quantum channels corresponding to Bob's and Charlie's actions.

On the other hand, in the scenario where the one-particle and vacuum sectors evolve incoherently with each other, the systems corresponding to Alice, Bob and Charlie would have repeated indices $A^{kk}, B^{mm}$ and $C^{pp}$.

Another example shows how routes can help derive some immediate consequences of discardings on the coherence between sectors. Suppose we have a routed channel from one system $A$ to two partitioned systems $B^{kk'}$ and $C^{ll'}$, which features perfect (possibly coherent) sectorial correlations between $B^{kk'}$ and $C^{ll'}$ -- i.e.\ this routed channel is of the form $(\delta^{kl} \delta^{k'l'}, \cc)$. This can be understood as a channel which, in particular, sends copies of the same information to two agents, Bob and Charlie. Indeed, if Bob measures in which of the sectors $B^k$ his system is, and Charlie does the same with the sectors $C^l$, they will find the same result.

Let us now look at what happens if Bob discards his system (or, more generally, if Bob's system never reaches  Charlie, as the latter's description of his own system is then the one obtained by discarding Bob's part). The discarding on $B^{kk'}$ is given by the routed quantum channel $(\delta_{kk'}, \disc_B)$, where $\disc_B$ is the trace on $\cl(\ch_B)$. The quantum channel this yields is therefore

\be %
\InputIfFileExists{Copying1.tikz}{}{\input{./figures/Copying1.tikz}} \quad = \quad  %
\InputIfFileExists{Copying2.tikz}{}{\input{./figures/Copying2.tikz}} \, .\ee

In the equation above, some simple calculus on routes alone yielded an important physical theorem. Indeed, the routed quantum channel obtained by discarding Bob's system is of the form $(\delta^{ll'}, \cc')$: its route means that it yields states that are completely decohered with respect to the partition $l$. Thus we proved in a natural way a well-known feature of quantum theory: copying information and then discarding one of the copies necessarily leads to the loss of any coherence in the other copy, between the sectors which encoded this information.

What is remarkable is that the use of completely positive routes allows to derive such a theorem from very simple calculus on boolean matrices, and without having to know anything specific about the channel $\cc$, except its crucial structural features. Moreover, the systematic nature of our framework means that one will be able to scale up such proofs easily: in any scenario in which information is copied in some way between any number of subsystems, for any number of subsystems being discarded, calculus on routes will yield direct consequences on the coherence between sectors for the other subsystems.

\subsection{Computing the accessible space corresponding to a slice}

The discussion of interpretations of slices in a routed diagram can also be generalised to routed quantum channels. The formal space corresponding to a slice will, once again, be the `big' space of linear operators corresponding to it: for example, the one corresponding to the slice comprising wires $A^{kk'}$, $B^{mm'}$ and $C^{pp'}$ in (\ref{eq:3ChannelSuperposition}) is $\cl(\ch_A \otimes \ch_B \otimes \ch_C)$. Accessible spaces will be defined as solely depending on the routes' diagonals: indeed, including the information on coherence encoded by the routes' off-diagonal elements would not yield satisfactory operator spaces. The accessible space corresponding to the previously mentioned slice, for example, is $\cl [ ( \ch_A^1 \otimes \ch_B^0 \otimes \ch_C^0) \oplus (\ch_A^0 \otimes \ch_B^1 \otimes \ch_C^0) \oplus (\ch_A^0 \otimes \ch_B^0 \otimes \ch_C^1) ]$. Since we are only using the routes' diagonals, the accessible space corresponding to the same slice in the incoherent case will be the same.  The general procedure of Section \ref{sec:accessiblespace} for computing the accessible space can easily be accommodated to general quantum channels: one follows it using the routes' diagonals, thus ending up with a Hilbert space $\ch^\tr{acc}$; the accessible space of linear operators corresponding to the slice is  then $\cl(\ch^\tr{acc})$.

\section{Index-matching quantum circuits} \label{sec:IndexMatching}

A drawback of general routed quantum circuits is that the sectorial correlations and constraints they feature are not represented in a completely graphical way; the routes are denoted by abstract symbols which do not depict graphically the possible connections between sectors\footnote{A possible systematic way to depict such connections would be to add a third dimension to diagrams, in order to represent direct sums\correction{, as is done with the sheet diagrams of Ref.\ \cite{comfort2020}. Sheet diagrams, however, can quickly grow unwieldy to write down and to decipher as the number of summands in a direct sum increases, in particular because their representations will necessarily be 2d projections of their 3d structures. Note also that 1) they can only be used to represent a subset of routed maps (for instance, a routed map with route $\lambda = \begin{pmatrix} 1 & 0\\ 1 & 1 \end{pmatrix}$ cannot be represented using a sheet diagram, as it is not a direct sum of maps on the sectors); and 2) in a theory of completely positive maps, they are unable to encompass routed maps which feature coherence between the sectors.}}.

Nevertheless, there is a sub-framework of routed quantum circuits which encompasses a fair share of scenarios (though not all of them), and in which sectorial correlations and constraints can be represented in a fully graphical way: this is the framework of \textit{index-matching quantum circuits}. As we will see, index-matching quantum circuits correspond to the `extended quantum circuits' introduced for the study of causal decompositions in Refs.\ \cite{Allen2017,barrett2019, lorenz2020, barrett2020cyclic}, and presented in Section \ref{sec:CausalDecs}. The theory of index-matching quantum circuits will thus also serve to provide a sound and systematic mathematical foundation to the use of the diagrams introduced in these earlier works.

The simple idea behind index-matching circuits is to make the most out of the graphical trick of Kronecker delta contraction, which we described earlier in a simple example. Thus, in this framework, one restricts the partitions to be indexed by a rigid combination of several indices, and only considers routes built from Kronecker deltas between such indices. This allows to represent these routes directly on a diagram, by repeating indices to denote the Kronecker deltas. The conditions for suitable composition also take a particularly simple form in this context, making it easier to ensure that \correction{an index-matching} circuit is suitable for practical isometries, practical unitaries, or practically trace-preserving maps. A typical example of an index-matching circuit is the causal decomposition (\ref{eq:DiamondDot}).

We formalise thoroughly the framework of index-matching routed maps in Appendix \ref{app:IndexMatching}; here, we will present it in a more accessible way. At the level of pure states and operators, it has two major components. The first one is partitioned Hilbert spaces whose partitions are labelled by several indices, i.e.\ which are of the form $A^{k_1 \ldots k_m}$. Each index $k_i$ has a length $\abs{k_i}$, denoting the number of values it can take. The second components is index-matching routed maps, which are routed maps whose route is solely written in terms of Kronecker deltas (in order to make sense, these Kronecker deltas necessarily have to relate indices of the same length). Examples of possible index-matching routes from $A^{k_1 k_2}$ to $B^{l_1 l_2}$ are $\delta^{l_1}_{k_1} \delta^{l_2}_{k_2}$, $\delta^{l_1 l_2}_{k_1 k_2}$, $\delta^{l_1 l_2}$, $\delta^{l_2}_{k_1 k_2}$, $1$, etc.

Using graphical Kronecker delta contractions, the routes of index-matching routed maps can therefore be represented in a fully graphical way. For example, if we look at maps of type $A^{k_1} \otimes B^{k_2} \to C^{l_1} \otimes D^{l_2}$, the maps $(\delta^{l_1 l_2}_{k_1 k_2},U_1)$, $(\delta^{l_1 l_2},U_2)$ and $(1, U_3)$ will respectively be represented as

\be %
\InputIfFileExists{IMmap1.tikz}{}{\input{./figures/IMmap1.tikz}} , %
\InputIfFileExists{IMmap2.tikz}{}{\input{./figures/IMmap2.tikz}} , %
\InputIfFileExists{IMmap3.tikz}{}{\input{./figures/IMmap3.tikz}} \, . \ee

Contractions of Kronecker deltas, however, can sometimes lead to ambiguities about the routes which are associated to each map in an index-matching circuit. For example, if we composed the channels $(\delta^{l_1}_{k_1} \delta^{l_2}_{k_2}, U): A^{k_1} \otimes B^{k_2} \to C^{l_1} \otimes D^{l_2}$ and  $(\delta^{m_1 m_2}_{l_1 l_2}, V) : C^{l_1} \otimes D^{l_2} \to E^{m_1} \otimes F^{m_2}$, this would lead to the index-matching circuit

\be %
\InputIfFileExists{IMmap4.tikz}{}{\input{./figures/IMmap4.tikz}} \, , \ee
in which one now cannot properly read the first channel's route anymore, as it has been \correction{`overwritten'} by the route of the second channel. Such issues entail that, if we want to make sure we will be able to give an unambiguous meaning to index-matching circuits, we will need to define the theory of such circuits in a more restrictive way. This is done formally in Appendix \ref{app:IndexMatchingDiagrams}; here, we will stress the main features of the theory thus obtained.

The idea is to go in the opposite direction: instead of starting with maps and defining graphs to represent their compositions, we shall start with abstract graphs, then interpret their wires and nodes as spaces and maps (this idea was loosely inspired by the approach developed in Ref.\ \cite{pinzani2020} for the formalisation of superpositions of causal order). We thus define indexed open directed acyclic graphs (IODAGs) as abstract open graphs made of nodes and directed wires, with the wires additionally bearing indices, and with an equivalence relation on indices (indicating which ones are `the same index'). Figure \ref{fig:ODAGexamples} shows some examples of IODAGs.

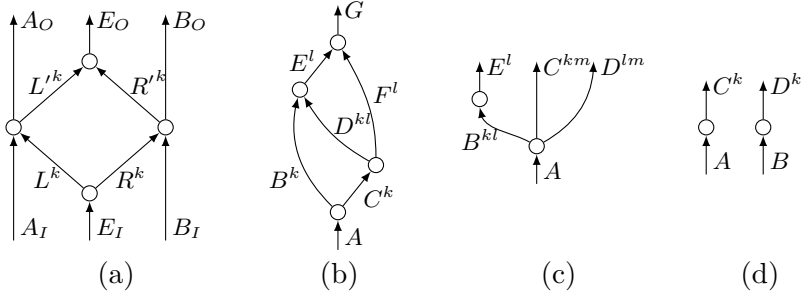
\begin{figure*}
    \centering
    \begin{tabular}{c c c c}
\InputIfFileExists{DotExample1.tikz}{}{\input{./figures/DotExample1.tikz}} & %
\InputIfFileExists{DotExample2.tikz}{}{\input{./figures/DotExample2.tikz}} & %
\InputIfFileExists{DotExample3.tikz}{}{\input{./figures/DotExample3.tikz}} & %
\InputIfFileExists{DotExample5.tikz}{}{\input{./figures/DotExample5.tikz}} \\
(a) & (b) & (c) & (d) \end{tabular}
    \caption{Examples of IODAGS.}
    \label{fig:ODAGexamples}
\end{figure*}

One can, as a second step, \textit{interpret} an IODAG by assigning partitioned Hilbert spaces to wires and maps to nodes, where the maps follow the routes specified by the index-matching in the diagram; this provides an index-matching circuit. This procedure allows to attribute proper semantics to index-matching quantum circuits.

IODAGs can be composed sequentially and in parallel. To prevent the appearance, in their interpretations, of ambiguities such as the ones described above, the possibility of sequential composition will be restricted: sequentially composing $\Gamma_1$ and $\Gamma_2$ is allowed only if $\Gamma_1$'s output wires are the same as $\Gamma_2$'s input wires, and if the equivalence classes among $\Gamma_1$'s outputs' indices are the same as those among $\Gamma_2$'s inputs' indices. For instance, taking the following IODAGs:
$$ \begin{tabular}{c c c}
\InputIfFileExists{DotCompExample1.tikz}{}{\input{./figures/DotCompExample1.tikz}} & \quad \quad \quad \quad & %
\InputIfFileExists{DotCompExample2.tikz}{}{\input{./figures/DotCompExample2.tikz}} \\ (e) & \quad \quad \quad \quad & (f1) \\ %
\InputIfFileExists{DotCompExample3.tikz}{}{\input{./figures/DotCompExample3.tikz}} & \quad \quad \quad \quad & %
\InputIfFileExists{DotCompExample4.tikz}{}{\input{./figures/DotCompExample4.tikz}}\\
 (f2) & \quad \quad \quad \quad & (f3) \end{tabular} \, ,$$
(e) cannot be composed with (f1) or with (f2), but it can be composed with (f3). In a sequential composition, equivalence classes which appear in the interface wires are merged; at the graphical level, this can lead to some relabelling. For instance, the composition of (e) and (f3) defined above yields
$$ %
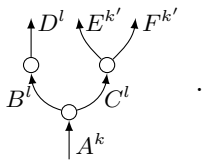
\InputIfFileExists{DotCompExample5.tikz}{}{\input{./figures/DotCompExample5.tikz}} \, . $$
Parallel composition, on which there are no restrictions, can lead to some relabelling as well; for instance, the parallel composition of (e) and (f2) is
$$ %
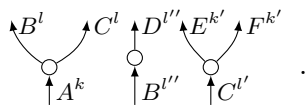
\InputIfFileExists{DotCompExample6.tikz}{}{\input{./figures/DotCompExample6.tikz}} \, . $$

The rules for suitable composition of practical isometries take a particularly simple form in index-matching diagrams. Let us define a `starting point' for an index as a node which features this index in its outputs but not in its inputs. An IODAG is proper for practical isometries if, for any index appearing in it, there is at most one starting point in the circuit for this index, and no starting point if this index appears in the global input wires of the diagram. It is proper for practical unitaries if it satisfies both this rule and a symmetric one for endpoints. For example, if we consider the diagrams of Figure \ref{fig:ODAGexamples}, (a), (b) and (c) are suitable for isometries, but not (d), as it features two starting points for the index $k$. (a) and (b) are also suitable for unitaries, but not (c), as the index $k$ is present in the global outputs of the diagram and has an endpoint.

Finally, as we have said, interpretations of IODAGs are given by assigning partitioned Hilbert spaces to wires and maps to nodes, where the maps follow the routes specified by the index-matching in the diagram. An interpretation of an IODAG thus yields a global index-matching routed map, called the \textit{meaning} of this interpretation. The meaning is obtained by composing the maps for each node in accordance with the graph, then composing with a pre-processing map, which serves to match input indices: the meaning of an interpretation of (f3), for instance, needs to include a pre-processing with a projector which will match the indices of its two input wires. Theorems \ref{th:InterpretationsMeaning}, \ref{th:SeqCompInterpretation} and \ref{th:ParrCompInterpretatiom} in Appendix \ref{app:IndexMatchingDiagrams} ensure that interpreting is a well-defined protocol, playing well with sequential and parallel compositions of IODAGS. Interpreting (a) in practical unitaries, for example, yields a circuit of the form (\ref{eq:DiamondDot}), which will therefore have a proper and unambiguous signification as a circuit of index-matching routed maps, and whose meaning will automatically be a unitary map. This ensures that our paradigmatic example of an index-matching quantum circuit is completely sound.

IODAGs could be applied to general quantum channels as well; the only difference is that each index will then become a pair of indices, to be able to denote constraints on coherence.

\correction{ Finally, let us note that index-matching quantum circuits only form a sub-framework of routed quantum circuits. A first example is that a route $\lambda = \begin{pmatrix} 1 & 0\\ 1 & 1 \end{pmatrix}$ cannot be written in terms of Kronecker-deltas; thus a routed map with this route would not be describable in the sole framework of index-matching quantum circuits. Another, more physically grounded example is that of the superposition of three trajectories, as depicted in (\ref{eq:3ChannelSuperposition}): it can be shown that (even in the unitary case) the sectorial correlations among the three wires in this diagram cannot be described using only Kronecker deltas.}

\section{Conclusion}

In this paper, we argued for the necessity of an extension to the framework of standard quantum circuits, and introduced such an extension, given by routed linear maps and routed quantum circuits. We proved that routed maps form a consistent framework, suitable for a sound and faithful diagrammatic expression in terms of circuits, and applicable to both pure and mixed quantum theory; we exemplified its use and interpretation in several cases. We showed how a recently introduced extension of quantum circuits \cite{lorenz2020} can be seen as arising naturally from a sub-framework of routed maps, and made use of this fact to provide a sound and consistent semantics to the use of this previous extension.

We believe that the tighter and more flexible description of quantum theory unlocked by the framework of routed maps could be the basis for significant advances in the understanding of the structure and possibilities of quantum theory. For example, in the two lines of research which we discussed to motivate the introduction of this framework, taking into consideration sectorial constraints and correlations led to a variety of novel results. In Ref.\ \cite{chiribella2019shannon}, this allowed to rigorously formalise the use of a superposition of two communication channels, which in turn led to the creation of a new communication \correction{paradigm exhibiting various communication advantages over what is possible in standard quantum Shannon theory, including a new generalised definition of the capacity of a quantum channel \cite{kristjansson2020resources}}. In Refs.\ \cite{Allen2017, barrett2019, lorenz2020, barrett2020cyclic}, taking sectorial constraints and correlations into account proved an unavoidable step in order to unravel the full scope of causal decompositions. The systematic and self-contained nature of the framework of routed quantum circuits could make it easier, neater and more natural to derive further such advances in the future, both in these lines of research and in others.

In the specific case of the study of causal decompositions, such decompositions were so far always found to be encompassed by what we called (practically unitary) `index-matching routed quantum circuits' in this paper. As the latter are a sub-framework of general practically unitary routed circuits, this raises an interesting question: is it possible that, in some more involved and still unproven cases, causal decompositions (of a unitary into a routed quantum circuit of practical unitaries) might not be encompassed by index-matching routed quantum circuits, but only by the more general ones? A positive answer to this question would prove that the full scope of routed quantum circuits is required in order to describe faithfully the causal structure of quantum theory.

Following on the results of Ref.\ \cite{chiribella2019shannon}, routed quantum circuits could also be used to provide clarifications on the possibility of implementing coherent control of quantum channels in more general scenarios, and to describe and prove the communication advantages which could be yielded by such a control. The possibility of formalising and harnessing quantum superpositions of \correction{quantum operations could also prove useful in various algorithms of quantum computation \cite{zhou2011adding, thompson2018, dong2019controlled}.}

Finally, the example of Ref.\ \cite{barrett2020cyclic}, in which causal decompositions involving index-matching routed quantum circuits were used to describe scenarios featuring indefinite causal order (ICO), and to prove  theorems about the structure of some of these scenarios, demonstrates that the framework of routed quantum maps could also find fruitful applications in the study of ICO -- another research direction which goes beyond the framework of standard quantum circuits \cite{chiribella2013quantum, oreshkov2012quantum} -- which gathered significant interest in the past decade. It is reasonable to think that a framework which encompasses the full scope of quantum-theoretical scenarios should be able to describe both ICO and sectorial correlations and constraints.

\begin{acknowledgments}
It is a pleasure to thank \"Amin Baumeler, Giulio Chiribella, Bob Coecke, James Hefford, Robin Lorenz and Matt Wilson for helpful discussions, advice and comments. AV is supported by the EPSRC Centre for Doctoral Training in Controlled Quantum Dynamics. HK acknowledges funding from the UK Engineering and Physical Sciences Research Council (EPSRC) through grant EP/R513295/1. This publication was made possible through the
support of the grant 61466 `The Quantum Information Structure of Spacetime (QISS)' (qiss.fr) from the John Templeton
Foundation. The opinions expressed in this publication
are those of the authors and do not necessarily reflect the views of the John Templeton Foundation.
\end{acknowledgments}

\bibliographystyle{utphys}
\bibliography{references}

\appendix

\section{Comments on the conceptual role of routes} \label{app:conceptual}

In this Appendix, we will provide a few comments on the conceptual role of the route in a routed map. This conceptual role can be understood by appealing to the notion of \textit{type}. Usually, specifying the type of a linear map $f$ means declaring its input space and its ouput space. Thus the type of a map is a structural piece of data which is prior to the specification of this map itself, with which the map itself is consistent, and which gives information about the ways in which this map can be composed: for instance, $f$ and $g$ can be composed only if $f$'s output space matches $g$'s input space. In a routed map $(\lambda, f)$, $\lambda$ should be morally understood as having the role of an \textit{additional type} for $f$, that comes as a supplement to the declaration of its partitioned input and output spaces\footnote{Let us stress, however, that this is a \textit{moral} account of routes, aimed at clarifying their conceptual role. From a purely formal point of view, the types of routed maps solely consist of their partitioned input and output spaces.}. This is indeed the case once one restricts to the theory of practical isometries or to that of routed quantum channels: routes -- and not solely input and output spaces -- have to be taken into consideration to determine which maps can be meaningfully composed.

That routes play the part of an additional type should also shed light on the seemingly disturbing fact that a given linear map $f$ can be compatible with several different routes\footnote{If $f$ follows $\lambda$, then $f$ also follows any $\tl{\lambda}$ such that $ \forall k,l, \, \tl{\lambda}^l_k \geq \lambda_k^l$.}, and on the natural question one can then ask: `what is the difference between $(\lambda, f)$ and $(\tl{\lambda},f)$ ?'. Our comments entail that this is essentially a question about the meaning that is to be ascribed to a modification of the type assigned to what is, morally, the same map. In fact, similar questions about the meaning of a type change can also arise for non-routed linear maps. Indeed, a same given linear map can also, while morally staying the same map, be ascribed a variety of output spaces (basically any space of which its `actual' output space is a subspace). In both this case and the case of a route change, the type change does not essentially modify the map, but it does modify our capacity to hold structural statements about it, and, in particular, to state what it is meaningful to compose it with.

\section{Connection with the CP* construction and other strategies} \label{app:CP*}

\correction{
In Section \ref{sec:motivation}, we showed that standard quantum circuits (i.e.\ those interpreted in either \FHilb\, or \CPM[\FHilb]) could not be used to provide an adequate description of superpositions of trajectories or causal decompositions. Here, we extend this discussion to the use of the CP* construction \cite{Coecke_2014}, a standard categorical construction yielding a theory CP*[\FHilb] that contains both quantum and classical channels: i.e., we show that neither circuits interpreted in CP*[\FHilb], nor simple constructions relying on CP*[\FHilb], can adequately model superpositions of trajectories or causal decompositions. These considerations also apply to the use of the Karoubi envelope of \CPM[\FHilb], a category slightly larger than CP*[\FHilb].

Let us first explain why it might be hoped that CP*[\FHilb] would provide a sound basis for a representation of these scenarios\footnote{We thank an anonymous QPL reviewer for raising this idea.}. The idea is that the objects of CP*[\FHilb], being defined as C* algebras, can equivalently be thought of as arising from the choice of a preferred partition of a Hilbert space; therefore, the use of CP*[\FHilb] would bypass the need for a definition of partitions `by hand' as is done in the present paper. From there, routes could be defined, not as a structure on maps, but rather as a purely diagrammatic piece of data, and it would suffice to impose that the interpretation of the diagrams in terms of morphisms in CP*[\FHilb]  be consistent with this diagrammatic information. 

We will first explain why this strategy cannot in fact be implemented using CP*[\FHilb]; then we will explain why we believe that the idea of defining routes as diagrammatic data to construct a suitable theory is at least as difficult as the approach taken in this paper.

First, CP*[\FHilb] is too restrictive to be used for the encoding of preferred partitions of Hilbert spaces. Indeed, if we define partitioned Hilbert spaces $A^k$ and $B^l$ as objects in CP*[\FHilb], then the morphisms $A^k \to B^l$ in CP*[\FHilb] can be defined as the CP linear maps from $\cl(\ch_A)$ to $\cl(\ch_B)$ that destroy coherence between the sectors of $A$ and do not allow for any coherence between those of $B$. However, we want the channels in our theory to possibly feature coherence between the sectors. This is in particular crucial in both examples of Section \ref{sec:motivation}: in superpositions of trajectories, the non-coherent case is a trivial and uninteresting one \cite{kristjansson2020resources}; and in causal decompositions, all channels have to be unitary and thus perfectly coherent. CP*[\FHilb] therefore cannot be used to model the structures we want to model.

We now comment on the more general idea of defining routes as mere diagrammatic data, rather than full-fledged morphisms, and simply asking for interpretations of diagrams to be consistent with that data. While this strategy has some advantages, it also presents significant drawbacks. First, it does not allow for the rewriting of a diagram in which, for instance, the composition of a box $f$ and of a box $g$ is replaced with a box $g \circ f$, as the `route' diagrammatic data associated to this box will be undefined -- unless compositions of routes are defined as well, which would bring one back towards a theory of routed maps as defined in this paper.

Second, it makes the problem of defining physical maps more difficult. Indeed, to obtain such a definition, we have a crucial need for a notion of practical input and output spaces, one that can only be defined from the data given by the routes themselves\footnote{The alternative would be to hardcode these `practical input and output spaces' into objects, using, for instance, Karoubi envelopes. However, going this way would come at the cost of defining a very complex pseudo-tensor product structure. Indeed, one would for instance have to express the way in which $L^k$ and $R^k$ in (\ref{eq:DiamondDot}) can be tensored in such a way as to yield not `$L^k \otimes R^{k'}$', but `$L^k \otimes R^k $'; and more generally there should be pseudo-tensor products defined for every possible case of sectorial correlations. In the general case of non-sector-preserving routed maps, these tensor products would present very exotic features: for instance, the codomain of the pseudo-tensor product of two maps would in general depend on the routes they follow. Therefore 1) doing things in this way would still require to define routes as more than diagrammatic pieces of data, and 2) even though conceivable, it would require mathematical constructions which are more involved than those in this paper.}. Here again, in the absence of the structural handles provided by the acknowledgement that routes are not just graphical objects but morphisms with their own compositions, it would be difficult to express the conditions for the composition of such physical maps to be well-defined;  expressing them would once again essentially amount to going back to a theory of routed maps.
}

\section{Categorical perspective} \label{app:categoricalperspective}

\subsection{Dagger symmetric monoidal categories} \label{app:DagSMCs}
Let us introduce, in a non-technical way, the mathematical concepts which can be used to characterise the properties of the frameworks built in the present paper. These concepts encapsulate the fact that a framework is suited for a diagrammatic representation of its maps in terms of circuits. The  structure necessary for this mimics the basic structure of quantum theory: existence of sequential and parallel compositions, of identity maps, of trivial spaces, and of hermitian conjugates, all interacting in a natural way. Any theory with these features accepts sound and intuitive diagrammatic representations of its maps in terms of circuits \cite{joyal1991, coecke_kissinger_2017}. These concepts originate from \textit{category theory}, a mathematical theory which has been at the centre of a recent re-formalisation of quantum theory \cite{abramsky2004categorical,coecke_kissinger_2017,selby2018reconstructing,tull2018categorical}. Our point here is not to present them in depth, but to provide the reader with an intuition of the simple structures that they express.

Categorical frameworks adopt the perspective of \textit{process theories}: this means a theory is described not through its states, but through its \textit{processes}, i.e.\ its dynamical transformations -- states will be recovered as special cases of processes \cite{coecke_kissinger_2017}. In the context of process theories, the main questions are about how processes can be composed together. A simple mathematical framework to describe sequential composition of processes is that of categories. A category contains two kinds of components: \textit{objects}, corresponding to what would usually be called a space; and \textit{morphisms} (or maps), with a specified domain (i.e.\ input space) and a specified codomain (i.e.\ output space), both chosen among the objects of the category. Two morphisms can be sequentially composed if the codomain of the first matches the domain of the second: i.e.\ if $f:A \to B$ and $g: B \to C$ are maps, they can be composed to form a map $g \circ f : A \to C$. In a category, sequential composition is associative, and for any object $A$ there exists an identity morphism $\textrm{id}_A : A \to A$.

Some categories, called \textit{symmetric monoidal categories} (SMC), also feature the structure for parallel composition of morphisms, in the form of an operation called the \textit{tensor product}\footnote{Note that what we call the `tensor product' in this context is not necessarily the tensor product of linear maps.}, $\otimes$. The tensor product of two objects $A$ and $B$ yields an object $A \otimes B$, and the tensor product of two morphisms $f: A \to B$ and $g: C \to D$ yields a morphism $f \otimes g: A \otimes C \to B \otimes D$. The tensor product is associative. SMCs also feature a `unit object' $I$ satisfying $A \otimes I = I \otimes A = A$, which can be thought of as the trivial space of the theory; and swap morphisms, which are, for any pair of objects $A$ and $B$, involutions from $A \otimes B$ to $B \otimes A$. These structures satisfy a set of various coherence conditions which ensure that they interplay appropriately (for instance, that sequential composition distributes over parallel composition, and so on). In an SMC, states on an object $A$ are morphisms from the unit object $I$ to $A$.

Finally, a dagger SMC is an SMC featuring an involution, called the adjoint, which, to any morphism $f: A \to B$, associates a morphism $f^\dagger: B \to A$. The adjoint satisfies various coherence conditions ensuring that it interplays consistently with the rest of the symmetric monoidal structure. Combinations of parallel and sequential compositions of maps in dagger SMCs can always be faithfully represented by so-called \textit{circuit diagrams}, in which maps are represented by boxes, and objects are represented by wires. For instance, the theory of linear maps between finite-dimensional Hilbert spaces forms a dagger SMC \cat{FHilb}. 
We refer the interested reader to Refs.\ \cite{coecke2010categories, coecke_kissinger_2017, fong2019invitation} for accessible introductions to symmetric monoidal categories.

All the theories we will be considering in this paper are dagger SMCs\footnote{In fact, all the theories presented here are dagger compact categories: on top of the dagger symmetric monoidal structure, they feature some additional structure, which can be roughly described as corresponding to the existence of a Choi-Jamiołkowski isomorphism. In particular, this entails that they can be faithfully represented by \textit{string diagrams}, an extension of circuit diagrams \cite{coecke_kissinger_2017}. As the presence of this structure is not essential to the purposes of the present work, we leave our discussion of it to another paper \cite{vanrietvelde2020categorical}.}, which makes them suitable for diagrammatic representation in terms of circuit diagrams \cite{joyal1991, coecke_kissinger_2017}.
In the following appendices we shall give the main elements of the proofs that the theories discussed in this paper form dagger SMCs; more refined proofs will be available in another paper focusing on the categorical structures which lead to the existence of routed categories \cite{vanrietvelde2020categorical}.

\subsection{Routed maps form a dagger SMC} \label{app:proofRoutedMapsSMC}

In this appendix, we prove the following theorem, which can be thought of as a rigorous version of Theorem \ref{th:PureLinearMapsSoundForCircuits}.

\begin{theorem} \label{th:RoutedMaps=dagSMC}
Partitioned spaces and routed maps form, respectively, the objects and morphisms of a dagger SMC $\cat{RoutedFHilb}$, in which:

\begin{itemize} 
\item composition is given by pairwise composition;
\item parallel composition is given on objects by $A^k \otimes B^l := (\ch_A \otimes \ch_B, \cz_A \times \cz_B, (\pi^k_A \otimes \mu^l_B)_{(k,l) \in \cz_A \times \cz_B})$, and on morphisms by the cartesian product on the routes together with the tensor product on the linear maps;
\item the trivial space is the trivial partition of $\mathbb{C}$: $1_\cat{RoutedFHilb} := (\mathbb{C}, \{ * \}, (1))$;
\item the adjoint of $(\lambda, f) : A^k \to B^l$ is $(\lambda, f)^\dagger := (\lambda^\top, f^\dagger): B^l \to A^k$;
\end{itemize}
\end{theorem}

We shall prove here the main features of \cat{RoutedFHilb}'s dagger symmetric monoidal structure: that it is closed under sequential and parallel compositions and under taking adjoints, that sequential and parallel compositions are associative, and that parallel composition distributes over sequential composition. That the various coherence conditions are also satisfied can be proven easily. In the following we will freely use some partitioned Hilbert spaces 
$A^k = (\ch_A, \cz_A, (\pi^k_A)_{k \in \cz_A})$, $B^l = (\ch_B, \cz_B, (\mu^l_B)_{l \in \cz_B})$, $C^m = (\ch_C, \cz_C, (\nu^m_C)_{m \in \cz_C})$,
$D^n = (\ch_D, \cz_D, (\eta^n_D)_{n \in \cz_D})$.

Let us first prove a useful equivalent definition of the fact that a linear map follows a route: a map $f: \ch_A \to \ch_B$ follows a route $\lambda$ if and only if $\forall k,l, \lambda^l_k = 0 \implies \mu^l_B \circ f \circ \pi^k_A = 0$. Considering first the direct implication and supposing that $f$ follows $\lambda$, one then has

\be \begin{split}
    0 &= f - f \\
    &= \left( \sum_l \mu^l_B \right) \circ f \circ \left( \sum_k \pi^k_A \right) - \sum_{k,l} \lambda^l_k \cdot \mu^l_B \circ f \circ \pi^k_A\\
    &= \sum_{k,l|\lambda^l_k = 0} \mu^l_B \circ f \circ \pi^k_A \, .
\end{split}  \ee

Hitting this equation with the $\mu^l$'s on the left and the $\pi^k$'s on the left yields: $\forall k,l: ~ \lambda^l_k = 0 \implies \mu^l_B \circ f \circ \pi^k_A = 0$. Reciprocally, supposing the latter, one has

\be \begin{split}
    f &= \left( \sum_l \mu^l_B \right) \circ f \circ \left( \sum_k \pi^k_A \right) \\
    &= \sum_{k,l|\lambda^l_k = 1} \mu^l_B \circ f \circ \pi^k_A\\
    &= \sum_{k,l} \lambda^l_k \cdot \mu^l_B \circ f \circ \pi^k_A \, .
\end{split}  \ee

We now prove that routed maps are closed under sequential composition. If we take two routed maps $(\lambda,f): A^k \to B^l$ and $(\sigma,g):B^l \to C^m$ and take indices $k,m$ such that $(\sigma \circ \lambda)^m_k = \sum_l \sigma^m_l \lambda^l_k = 0$, we have $\nu_C^m \circ g \circ f \circ \pi_A^k = \nu_C^m \circ g \circ (\sum_l \mu^l_B) \circ f \circ \pi_A^k$. Yet the fact that $\sum_l \sigma^m_l \lambda^l_k = 0$ implies that for any given $l$, one has either $\sigma^m_l = 0$ or $\lambda^l_k = 0$. The first case implies that $\nu_C^m \circ g \circ \mu^l_B =0$, and the second that $\mu_B^l \circ f \circ \pi^k_A =0$. Thus all the terms in this sum are null and $\nu_C^m \circ g \circ f \circ \pi_A^k = 0$, so $g \circ f$ follows $\sigma \circ \lambda$, so $(\sigma \circ \lambda, g \circ f)$ is a routed map.

For parallel composition, taking  $(\lambda,f): A^k \to C^m$ and $(\sigma,g):B^l \to D^n$, we have that $(\lambda \times \sigma)^{mn}_{kl} = \lambda^m_k \sigma^n_l$. So $(\lambda \times \sigma)^{mn}_{kl}=0$ implies that either $\lambda^m_k$ or $\sigma^n_l$ is null. In the first case, $\nu_C^m \circ f \circ \pi^k_A = 0$, in the second one, $\eta_D^n \circ g \circ \mu^l_B = 0$; so in both cases, $(\nu_C^m \otimes \eta_D^n) \circ (f \otimes g) \circ (\pi^k_A \otimes \mu^l_B) = (\nu_C^m \circ f \circ \pi^k_A) \otimes (\eta^n_D \circ g \circ \mu_B^l) = 0$. Therefore, $(\lambda \times \sigma, f \otimes g)$ is a routed map.

For closure under taking adjoints, take $(\lambda,f): A^k \to B^l$. Then for given $k$ and $l$, $(\lambda^\top)^k_l = 0 \implies \lambda^l_k = 0 \implies \mu^l_B \circ f \circ \pi^k_A = 0 \implies \left(\mu^l_B \circ f \circ \pi^k_A\right)^\dagger = 0 \implies \pi^k_A \circ f^\dagger \circ \mu^l_B = 0$, where in the last implication we used the fact that orthogonal projectors are self-adjoint. $(\lambda^\top,f^\dagger)$ is thus a routed map.

Finally, that parallel and sequential compositions are associative (both on objects and on morphisms) and that the former distributes over the latter is direct as these were defined pairwise from sequential and parallel compositions which possess all these properties.

\subsection{Routed CPMs form a dagger SMC} \label{app:proofRoutedCPMsSMC}

Here, we prove the analogue of Theorem \ref{th:RoutedMaps=dagSMC} for the case of routed completely positive maps. Let us start with a formal characterisation of how an orthogonal partition of a Hilbert space $\ch_A$ induces an orthogonal partition of the space $\cl(\ch_A)$ of linear operators on $\ch_A$.

\begin{theorem} \label{th:OrthPartCPM}
If $(\pi^k_A)_{k \in \cz_A}$ is an orthogonal partition of $\ch_A$, then, defining the following linear operators on $\cl(\ch_A)$,

\be 
\forall k,k', \, \,
\tilde{\pi}^{k,k'} : \rho \mapsto \pi^{k}_A \circ \rho \circ \pi^{k'}_A \, ,
\ee
$(\tilde{\pi}^{k,k'} )_{(k,k') \in \cz_A \times \cz_A}$ is an orthogonal partition of $\cl(\ch_A)$ (with respect to the Hilbert-Schmidt inner product).
\end{theorem}

The proof is direct.

We can then characterise the dagger SMC formed by routed CPMs.

\begin{theorem} \label{th:RoutedCPMs=dagSMC}
There exists a dagger SMC $\cat{RoutedCPMFHilb}$, in which:

\begin{itemize} 
\item objects are partitioned Hilbert spaces $A^{kk'} := (\cl(\ch_A), \cz_A \times \cz_A, (\tilde{\pi}^{k,k'})_{(k,k') \in \cz_A \times \cz_A})$, whose underlying space is a space of linear operators on a given Hilbert space, and whose orthogonal partition is one obtained from an orthogonal partition of this Hilbert space through the procedure of Theorem \ref{th:OrthPartCPM};
\item morphisms $A^{kk'} \to B^{ll'}$ are routed maps $(\Lambda, \cc)$, with $\Lambda$ a completely positive relation and $\cc$ a completely positive linear map;
\item composition is given by pairwise composition;
\item parallel composition is given on objects by $A^{kk'} \otimes B^{ll'} := (\cl(\ch_A \otimes \ch_B), \cz_A \times \cz_B \times \cz_A \times \cz_B, (\tilde{\pi}^{k,k'}_A \otimes \tilde{\mu}^{l,l'}_B)_{k,l,k',l'})$, and on morphisms by the cartesian product on the routes together with the tensor product on the linear maps;
\item the trivial space is the trivial partition of $\cl(\mathbb{C})$: $1_\cat{RoutedCPMFHilb} := (\cl(\mathbb{C}), \{ * \}, (1))$;
\item the adjoint of $(\Lambda, \cc)$ is $(\Lambda, \cc)^\dagger := (\Lambda^\top, \cc^\dagger)$.
\end{itemize}
\end{theorem}

The proof is very similar to \correction{the proof of Theorem} \ref{th:RoutedMaps=dagSMC}. Closure under sequential and parallel composition come from the combination of two facts: that sequential and parallel compositions of routed maps are themselves routed maps, and that sequential and parallel compositions of completely positive linear maps and relations are themselves completely positive. The first fact was proven in Appendix \ref{app:proofRoutedMapsSMC}; the second comes from the fact that completely positive morphisms can be obtained from the universal construction of Selinger \cite{Selinger2007} for any $\dagger$-compact category, and therefore form a $\dagger$-compact category themselves. The same facts entail that routed completely positive maps are closed under taking adjoints. Finally, that parallel and sequential compositions are associative (both on objects and on morphisms) and that the former distributes over the latter is direct as these were defined pairwise from sequential and parallel compositions which possess all these properties.

\section{Compositions of practical isometries} \label{app:practicalIsos}

We first prove Theorem \ref{th:IsoCompos}. Let $(\lambda, U): A^k \to B^l$ and $(\sigma, V) : B^l \to C^m$ be practical isometries, such that $\lambda$ and $\sigma$ satisfy (\ref{eq:IsoComposition}). Then the practical input set of $\sigma \circ \lambda$ is $\cs_{\sigma \circ \lambda} = \lambda^\top [\cs_\sigma] \subseteq \cs_\lambda$. The practical input space of $(\sigma, V) \circ (\lambda, U)$ is therefore $\ch_A^{\cs_{\sigma \circ \lambda}} \subseteq \ch_A^{\cs_\lambda}$. As $U$ is a partial isometry with initial subspace $\ch_A^{\cs_\lambda}$, it is in particular an isometry when restricted to $\ch_A^{\cs_{\sigma \circ \lambda}}$. Moreover, condition (\ref{eq:IsoComposition}) and the fact that $U$ follows $\lambda$ imply that $U(\ch_A^{\cs_{\sigma \circ \lambda}}) = U(\ch_A^{\lambda^\top [\cs_\sigma]}) \subseteq \ch_B^{(\lambda \circ \lambda^\top) [\cs_\sigma]} \subseteq \ch_B^{\cs_\sigma}$. Therefore, as $V$ is a partial isometry with initial subspace $\ch_B^{\cs_\sigma}$, it is in particular an isometry when restricted to $U(\ch_A^{\cs_{\sigma \circ \lambda}})$. It follows that $V \circ U$ is an isometry when restricted to $\ch_A^{\cs_{\sigma \circ \lambda}}$; so $(\sigma, V) \circ (\lambda, U)$ is a practical isometry.

We now prove that parallel composition of practical isometries always yields practical isometries. Let $(\lambda, U) : A^k \to B^l$ and $(\sigma, V) : C^m \to D^n$ be two practical isometries. Then $\cs_{\lambda \times \sigma} = \cs_\lambda \times \cs_\sigma$, so $(\ch_A \otimes \ch_B)^{\cs_{\lambda \times \sigma}} = \ch_A^{\cs_\lambda} \otimes \ch_B^{\cs_\sigma}$. The restriction  of $U \otimes V$ to $(\ch_A \otimes \ch_B)^{\cs_{\lambda \times \sigma}}$ is thus the tensor product of the restriction of $U$ to $\ch_A^{\cs_\lambda}$ with the restriction of $V$ to $\ch_B^{\cs_\sigma}$. As both of these are isometries, $(\lambda, U) \otimes (\sigma, V)$ is a practical isometry. The rest of the proof for the case of practical unitaries is similar.

\section{Computing the accessible space}\label{app:accessibleSpace}
\correction{

In this Appendix, we provide a more formal justification of the method presented in Section \ref{sec:accessiblespace} to compute the accessible space corresponding to a slice in a given routed diagram.

Let us start with a formal definition of the accessible space. We take a routed circuit and pick a slice in it. We restrict ourselves to considering only the interpretation of this diagram in the theory of relations; considering the `linear map' part of the interpretations is not important here, as we will only use the route information. We call $\cz$ the set of possible values of the indices in our slice, $\calp$ the set of possible values of the indices of the open wires at the bottom of the diagram and $\cf$ that of the open wires at the top of the diagram. We will consider what our diagram yields if we insert a given relation at this slice; given a relation $\lambda : \cz \to \cz$, we call $\ce(\lambda) : \calp \to \cf$ the interpretation of the diagram when $\lambda$ is inserted at the slice.

For any $k \in \cz$, let us define the relation $\varpi[k] : \cz \to \cz$ by  $\varpi[k]_{k'}^{k''} = \delta_{kk'}^{kk''}$. We say that $k \in \cz^\tr{non-acc}$ if $\ce(\varpi[k]) = 0$, and define $\cz^\tr{acc}$ as the complement of $\cz^\tr{non-acc}$ in $\cz$. $\cz^\tr{acc}$ corresponds to the index values that will form the accessible Hilbert space, i.e.\ we can define $\ch^\tr{acc} := \oplus_{k \in \cz^\tr{acc}} \ch^k$; indeed, the complementary set $\cz^\tr{non-acc}$ is defined as containing those values of $k$ which we know will be `killed' by the routes.

Let us now introduce a useful lemma: if, for a finite set $\cz$, we note as $\varsigma: \{*\} \to \cz$ the `full' relation defined by $\forall k \in \cz, \varsigma^k = 1$, then one has: $\forall \tau: \calp \to \cf, \tau = 0 \iff \varsigma_\cf^\top \circ \tau \circ \varsigma_\calp = 0$. The non trivial part here is the reverse implication; we can prove it by noting that the unique component of $\varsigma_\cf^\top \circ \tau \circ \varsigma_\calp$ is $\sum_{kk'} \tau^{k'}_k$; the rules of boolean calculus therefore yield: $\varsigma_\cf^\top \circ \tau \circ \varsigma_\calp = 0 \implies \forall k,k', \tau^{k'}_k =0 \implies \tau = 0$.

We thus have: $k \in \cz^\tr{acc} \iff \varsigma_\cf^\top \circ \ce(\varpi[k]) \circ \varsigma_\calp = 1 $. One can see that the steps 1, 2, 3 and 4 in the procedure of Section \ref{sec:accessiblespace} correspond to the computation of $\varsigma_\cf^\top \circ \ce(\varpi[k]) \circ \varsigma_\calp$: in particular, the absence of summation on the indices of the slice corresponds to the insertion of $\varpi[k]$, and the summation over the indices of the input and output open wires corresponds to the composition with the $\varsigma$'s. Step 5 thus recovers our formal definition of the accessible space.

}

\section{Sectorial constraints and Kraus representations} \label{app:Kraus}

We first prove Theorem \ref{th:KrausFullCoherence}. We take a routed CPM $(\Lambda,\cc): A^{kk'} \to B^{ll'}$. One can prove, in a similar way to how it was done for linear maps in Appendix \ref{app:proofRoutedMapsSMC}, that condition (\ref{eq:FollowingCPM}) is equivalent to the fact that for all $k,k',l,l'$, $\Lambda_{kk'}^{ll'} = 0 \implies \forall \rho, \mu_B^l \circ \cc ( \pi_A^k \circ \rho \circ \pi_A^{k'} ) \circ \mu_B^{l'} = 0$. Let us take $k,k',l,l'$ such that $\dot{\Lambda}_{k}^{l} = 0$, and a Kraus representation $\{K_i\}_i$ of $\cc$. If we take any states $\ket{\psi} \in \ch_A^k, \ket{\phi} \in \ch_B^l$, we have $\forall i, \bra{\phi} K_i \ket{\psi} \bra{\psi} K_i^\dagger \ket{\phi} \in \mathbb{R}^+$ and $\sum_i \bra{\phi} K_i \ket{\psi} \bra{\psi} K_i^\dagger \ket{\phi} = \bra{\phi} \cc( \ket{\psi} \bra{\psi}) \ket{\phi} = 0$, which implies $\forall i, \bra{\phi} K_i \ket{\psi} \bra{\psi} K_i^\dagger \ket{\phi} = 0$. As this is true for any $\ket{\psi} \in \ch_A^k, \ket{\phi} \in \ch_B^l$, one has $\forall i, \mu_B^l K_i \pi_A^k = 0$. Thus each of the $K_i$'s follow $\dot{\Lambda}$.

Let us prove the reverse implication for a $\Lambda$ with full coherence. If we take $\cc$ with Kraus representation $\{K_i\}_i$ such that each of the $K_i$'s follow $\dot{\Lambda}$, then for all $k,k',l,l'$, $\Lambda_{kk'}^{ll'} = \dot{\Lambda}_k^l \dot{\Lambda}_{k'}^{l'} = 0$ implies that at least one of $\dot{\Lambda}_k^l$ and $\dot{\Lambda}_{k'}^{l'}$ is null, and thus that for any given $i$, at least one of $\mu_B^l K_i \pi_A^k$ and $\mu_B^{l'} K_i \pi_A^{k'}$ is null. Therefore $\forall \rho, \mu_B^l \circ \cc ( \pi_A^k \circ \rho \circ \pi_A^{k'} ) \circ \mu_B^{l'} = \sum_i \mu_B^l \, K_i \, \pi_A^k \,  \rho \,  \pi_A^{k'} \, K_i^\dagger \, \mu_B^{l'} = 0$, so $\cc$ follows $\Lambda$.

We turn to the proof of Theorem \ref{th:KrausFullDecoherence}. Let us take a completely positive route $\Lambda$ with full decoherence, i.e.\ $\Lambda_{kk'}^{ll'} = \dot{\Lambda}_k^l \delta_{kk'} \delta^{ll'}$, and a completely positive map $\cc$ following $\Lambda$. Then $\forall \rho, \cc(\rho) = \sum_{k,l} \dot{\Lambda}_k^l \mu_B^l \circ \cc(\pi_A^k \circ \rho \circ \pi_A^k) \circ \mu_B^l$; thus if we define, for any $k,l$ such that $\dot{\Lambda}_k^l = 1$, $\cc_k^l : \rho \mapsto \mu_B^l \circ \cc(\pi_A^k \circ \rho \circ \pi_A^k) \circ \mu_B^l$, one has $\cc = \sum_{k,l| \dot{\Lambda}_k^l = 1} \cc_k^l$, and each of the $\cc_k^l$'s is a completely postive map from $\cl(\ch_A^k)$ to $\cl(\ch_B^l)$. Taking a Kraus representation for each of the $\cc_k^l$'s yields a Kraus representation of $\cc$ of the form given by Theorem \ref{th:KrausFullDecoherence}. The reverse implication is direct.

\section{A formal construction of index-matching routed maps} \label{app:IndexMatching}

\subsection{Index-matching routed maps as a category}

In this appendix, we present a formal construction of the framework of index-matching routed maps, which was introduced in a more intuitive way in Section \ref{sec:IndexMatching}.

We first need to formally define \textit{multiple indexings}: families of indices, each with its length, i.e.\ the number of different values it can take.

\begin{definition}
A (finite) \textbf{index family} is a finite set $\cx$ equipped with a `length' function $l: \cx \to \mathbb{N}$.
Given such an index family, the corresponding (finite) \textbf{multiple indexing} is the set $\bar{\cx} := \bigtimes_{x \in \cx} \llbracket 1, l(x) \rrbracket$, where $\forall n, \llbracket 1, n \rrbracket := \{ m \in \mathbb{N} | 1 \leq m \leq n \}$ and $\bigtimes$ denotes the cartesian product.
\end{definition}

$\cx$ serves as an `indexing of indices': it gives names to the different possible indices.

Routes in index-matching routed maps have to be \textit{corelations} \cite{grandis2018} (this is written with a single r). Corelations will be used to define, among the union of their input and output indices, clusters of indices which will be matched (i.e.\ will `be the same index').

\begin{definition}
Let $\cx_A$ and $\cx_B$ be two finite sets. A (finite) \textit{corelation} $\kappa : \cx_A \to  \cx_B$ is an equivalence relation on the disjoint union $ \cx_A \sqcup  \cx_B$.
\end{definition}

Finite corelations can be composed sequentially and in parallel, and form a dagger SMC \cat{FCoRel} \cite{coya2018}.

\begin{definition}
Let $\cx_A$ and $\cx_B$ be finite index families. An \textit{index-matching} from $\cx_A$ to $\cx_B$ is a corelation $\kappa: \cx_A \to \cx_B$ such that: $\forall x, x' \in \cx_A \sqcup \cx_B, x \stackrel{\kappa}{\sim} x' \implies l(x) = l(x')$.
\end{definition}

It is easy to see that the theory in which objects are index families and morphisms are index-matchings itself forms a dagger SMC \cat{FMatch}. The following ensures that index-matchings are just special cases of relations: for each index-matching between index families, there is a corresponding relation between the corresponding multiple indexings, in a consistent way. This can be considered as a formal way of defining the relation corresponding to an index-matching as made of Kronecker deltas determined by this index-matching.

\begin{definition}
For any index-matching $\kappa : \cx_A \to \cx_B$, the \textit{relation associated to} $\kappa$ is $\bar{\kappa}: \bar{\cx}_A \to \bar{\cx}_B$, defined by the following condition: an element $ \vec{k} = (k_x)_{x \in \cx_A}$ of $\bar{\cx}_A$ is \emph{\underline{not}} related by $\bar{\kappa}$ to an element $(k_x)_{x \in \cx_B}$ of $\bar{\cx}_B$ if and only if there exist $x, x' \in \cx_A \sqcup \cx_B$ such that $x \stackrel{\kappa}{\sim} x'$ and $k_x \neq k_{x'}$.
\end{definition}

Going from index families to multiple indexings, and from index-matchings to relations, is an operation which preserves the dagger SMC structure of \cat{FMatch} into that of \cat{FRel}.

\begin{theorem} \label{th:FunctorIndexMatchingToRel}
The `bar' operation, which associates to an index family its corresponding multiple indexing, and to an index-matching its associated relation, is a functor of dagger SMCs.
\end{theorem}

\begin{proof}
Let us first prove that it preserves sequential composition, i.e., $\overline{\kappa' \circ \kappa} = \bar{\kappa}' \circ \bar{\kappa}$. From $\overset{\kappa}{\sim}$ and $\overset{\kappa'}{\sim}$, one can form an equivalence relation $\sim$ on $\cx_A \sqcup \cx_B \sqcup \cx_C$, of which $\overset{\kappa' \circ \kappa}{\sim}$ is the restriction to $\cx_A \sqcup \cx_C$. Suppose $(k_x)_{x \in \cx_A} \overset{\overline{\kappa' \circ \kappa}}{\sim} (k_x)_{x \in \cx_C}$; then $\forall x, x' \in \cx_A \sqcup \cx_C, x \overset{\kappa' \circ \kappa}{\sim} x' \implies k_x = k_{x'}$. One can thus complete this by finding a family $(k_x)_{x \in \cx_B}$ such that: $\forall x, x' \in \cx_A \sqcup \cx_B \sqcup \cx_C, x \sim x' \implies k_x = k_{x'}$. Then in particular $(k_x)_{x \in \cx_A} \overset{ \bar{\kappa}}{\sim} (k_x)_{x \in \cx_B} \overset{ \bar{\kappa}'}{\sim} (k_x)_{x \in \cx_C} $, so $(k_x)_{x \in \cx_A} \overset{ \bar{\kappa}' \circ \bar{\kappa}}{\sim} (k_x)_{x \in \cx_C} $.

Reciprocally, if $(k_x)_{x \in \cx_A} \overset{ \bar{\kappa}' \circ \bar{\kappa}}{\sim} (k_x)_{x \in \cx_C}$; then there exists $(k_x)_{x \in \cx_B}$ such that $(k_x)_{x \in \cx_A} \overset{ \bar{\kappa}}{\sim} (k_x)_{x \in \cx_B} \overset{ \bar{\kappa}'}{\sim} (k_x)_{x \in \cx_C} $. If we take $x \in \cx_A, x' \in \cx_C$ such that $k_x \neq k_{x'}$, then for any $x'' \in \cx_B$, at least one of the propositions $k_x = k_{x''}$, $k_{x''} = k_{x'}$ is false, so it is not possible that $x \overset{\kappa}{\sim} x'' \overset{\kappa'}{\sim} x'$, so $x \overset{\kappa' \circ \kappa}{\not\sim} x'$. With the same reasoning, one can prove the same thing if $x$ and $x'$ are both either in $\cx_A$ or in $\cx_C$ and $k_x \neq k_{x'}$. Thus  $(k_x)_{x \in \cx_A} \overset{\overline{\kappa' \circ \kappa}}{\sim} (k_x)_{x \in \cx_C}$. From this implication and the previous one, it follows that $\overline{\kappa' \circ \kappa} = \bar{\kappa}' \circ \bar{\kappa}$.

It is then a routine check to prove that $\overline{\cx_A \times \cx_B} = \bar{\cx}_A \times \bar{\cx}_B$, $\overline{\kappa' \times \kappa} = \bar{\kappa}' \times \bar{\kappa}$, $\overline{\kappa^\top} = \bar{\kappa}^\top$, etc. \end{proof}

Thus, index-matchings can be seen as forming a subtheory of relations. This allows us to define notions for index-matchings from the notions for relations.

\begin{definition}
Let $(\ch_A, \bar{\cx}_A, (\pi^{\vec{k}})_{\vec{k} \in \bar{\cx}_A})$ and $(\ch_B, \bar{\cx}_B, (\mu^{\vec{l}})_{\vec{l} \in \bar{\cx}_B})$ be two partitioned spaces, where $\bar{\cx}_A$ and $\bar{\cx}_B$ are multiple indexings for index families $\cx_A$ and $\cx_B$, and let $\kappa : \cx_A \to \cx_B$ be an index-matching. A linear map $f : \ch_A \to \ch_B$ \textbf{follows the index-matching route} $\kappa$ if it follows its associated relation $\bar{\kappa}$. The pair $(\kappa, f)$ is then an \textbf{index-matching routed map}.
\end{definition}

The following is then direct.

\begin{theorem}
Index-matching routed maps form a dagger SMC \cat{MatchedFHilb}, which is embedded into \cat{RoutedFHilb}.
\end{theorem}

\subsection{Practical isometries and their composition}

The definitions of practical isometries and practical unitaries in \cat{RoutedFHilb} can be used in \cat{MatchedFHilb} as well. In this context, Theorem \ref{th:IsoCompos} becomes more intuitive. First, we will define formally what it means to create and delete an index.

\begin{definition}
An index \textbf{created} (resp.\ \textbf{deleted)} by an index-matching $\kappa$ is an equivalence class under $\kappa$ which only contains output (resp.\ input) elements. Each of these elements is a \textbf{representative} of the created (resp.\ deleted) index.
\end{definition}

This leads to a characterisation of those compositions which do \emph{not} satisfy Theorem \ref{th:IsoCompos}.

\begin{theorem} \label{SuitableIsoMatch}
Let $\kappa: \cx_A \to \cx_B$ and $\iota: \cx_B \to \cx_C$ be index-matchings. The composition of $\kappa$ and $\iota$ is improper for isometries if and only if there exists an index of length greater than or equal to $2$ created by $\kappa$, such that, noting $\cw \subseteq \cx_B$ as the set of representatives of this index, $\iota$ matches at least one index in $\cw$ with an index in $\cx_B  \setminus \cw$.
\end{theorem}

\begin{proof}
Let us note $\cx_B \sqcup \cx_B = \{ (i,x) | i \in \{1,2\}, x \in \cx_B \}$. It is easy to see that for $x \in \cx_B$, $(1, x) \overset{\kappa \circ \kappa^\top}{\not\sim} (2,x)$ if and only if $x$'s equivalence class under $\overset{\kappa}{\sim}$ is an index created by $\kappa$.

If the composition of $\kappa$ and $\iota$ is not proper for isometries, there exist $\vec{k} \in \cs_{\bar{\iota}}, \vec{k}' \in  \bar{\cx}_B \setminus \cs_{\bar{\iota}}$ such that $\vec{k} \overset{\bar{\kappa} \circ \bar{\kappa}^\top}{\sim} \vec{k}'$. The fact that $\vec{k}' \in  \bar{\cx}_B \setminus \cs_{\bar{\iota}}$ implies there exist $x, x' \in \cx_B$ such that $x \overset{\iota}{\sim} x'$ and ${k'}_{x} \neq {k'}_{x'}$; the fact that $\vec{k} \in \cs_{\bar{\iota}}$ implies that ${k}_{x} = {k}_{x'}$. That $(1, x) \overset{\kappa \circ \kappa^\top}{\sim} (2,x)$ and $(1, x') \overset{\kappa \circ \kappa^\top}{\sim} (2,x')$ would imply ${k'}_x = k_x = k_{x'} = {k'}_{x'}$, which would be a contradiction. Thus one of them (say, $x$) satisfies $(1, x) \overset{\kappa \circ \kappa^\top}{\not\sim} (2,x)$, so its equivalence class under $\overset{\kappa}{\sim}$ is an index created by $\kappa$. Calling this equivalence class $\cw \subseteq \cx_B$, one has $x' \in \cx_B \setminus \cw$, as $x' \in \cw$ would imply ${k'}_x = {k'}_{x'}$. Finally, as $x \overset{\iota}{\sim} x'$, $x$ and $x'$ have the same length, and ${k'}_{x} \neq {k'}_{x'}$ implies that this length is at least $2$; so the index which $\cw$ represents has length at least $2$.

Reciprocally, suppose there exists an index of length greater than or equal to $2$ created by $\kappa$, with set of representatives $\cw \subseteq \cx_B$, such that $\iota$ matches $x \in \cw$ with $x' \in \cx_B  \setminus \cw$. Then there exists a $\vec{k}' \in \bar{\cx}_B$ whose indices ${k'}_{y}$ all have value $1$, except for the $y$'s in the equivalence class of $x$ under $\overset{\kappa}{\sim}$, for which the value is ${k'}_y = 2$. We also define $\vec{k} \in \bar{\cx}_B$ whose indices all have value $1$. As the equivalence class of $x$ under $\overset{\kappa}{\sim}$ is an index created by $\kappa$, $(1, x) \overset{\kappa \circ \kappa^\top}{\not\sim} (2,x)$, so $\vec{k} \overset{\bar{\kappa} \circ \bar{\kappa}^\top}{\sim} \vec{k}'$; yet $k \in \cs_{\bar{\iota}}$ and $k' \not\in \cs_{\bar{\iota}}$, as ${k'}_{x} = 2 \neq {k'}_{x'} = 1$. Thus $\bar{\kappa} \circ \bar{\kappa}^\top (\cs_{\bar{\iota}}) \not\subseteq \cs_{\bar{\iota}}$ and the composition of $\bar{\kappa}$ and $\bar{\iota}$ is not suitable for isometries.
\end{proof}

We now just need to spell out the corresponding requirement for unitaries.

\begin{theorem}\label{SuitableUniMatch}
Let $\kappa$ and $\iota$  be matchings of indices, with $\kappa$'s codomain equal to $\iota$'s domain. The set of indices' names in the intermediary domain is noted $\cx_B$. The composition of $\kappa$ and $\iota$ is improper for unitaries if and only if at least one of the following is true:

\begin{itemize}
\item there exists an index of length greater than or equal to $2$ created by $\kappa$ such that, noting $\cw \subseteq \cx_B$ as the set of representatives of this index, $\iota$ matches at least one index in $\cw$ with an index in $\cx_B \setminus \cw$;
\item there exists an index of length greater than or equal to $2$ deleted by $\iota$ such that, noting $\cw \subseteq \cx_B$ as the set of representatives of this index, $\kappa$ matches at least one index in $\cw$ with an index in $\cx_B  \setminus \cw$.
\end{itemize}  
\end{theorem}

\section{A formal construction of index-matching quantum circuits} \label{app:IndexMatchingDiagrams}

\subsection{Definition and composition}

First, we define indexed wire systems.

\begin{definition}
An \textbf{indexed wire system} is a finite set $I$ equipped with a set of indices $K_I$, a function $p_I : K_I \to I$ (indicating the indices' placement) and an equivalence relation $\sim_I$ on $K_I$.
\end{definition}

An indexed open directed acyclic graph (IODAG) will then be a map from one indexed wire system to another, taking the form of a multi-indexed directed acyclic graph.

\begin{definition}
Let $I$ and $O$ be two indexed wire systems. An \textbf{indexed open directed acyclic graph} $\Gamma : I \to O$ consists of the following:
\begin{itemize}
\item finite sets $E_\Gamma$ (inner edges), $N_\Gamma$ (nodes), and $K_{E_\Gamma}$ (indices for the inner edges);
\item a head function $h_\Gamma : I \sqcup E_\Gamma \to N_\Gamma$, a tail function $t_\Gamma : E_\Gamma \sqcup O \to N_\Gamma$, and a placing function $p_{E \Gamma} : K_{E_\Gamma} \to E_\Gamma$;
\item an equivalence relation $\sim_\Gamma$ on $K_I \sqcup K_{E_\Gamma} \sqcup K_O$ which reduces to $\sim_I$ on $K_I$ and to $\sim_O$ on $K_O$;
\end{itemize} 
such that the directed graph formed by the edges and nodes is acyclic.
\end{definition}

Note that $h^{-1} : N_\Gamma \to \calp (I \sqcup E_\Gamma)$ and $t^{-1} : N_\Gamma \to \calp (E_\Gamma \sqcup O)$ both serve to specify, respectively, the subset of edges coming in a given node and the subset of edges going out of it. We will note $p_\Gamma := \langle p_I, p_{E \Gamma}, p_O \rangle$.

Most structural theorems for IODAGs will hold only up to isomorphism of IODAGs.

\begin{definition}
An \textbf{isomorphism of IODAGs} from $\Gamma$ to $\Gamma'$ is given by bijections $\alpha : E_\Gamma \to E_\Gamma'$, $\beta : N_\Gamma \to N_{\Gamma'}$ and $\gamma : K_{E_\Gamma} \to K_{E_{\Gamma'}}$, such that $h_{\Gamma'} \circ \langle \tr{id}_I, \alpha \rangle = \beta \circ h_\Gamma$, $t_{\Gamma'} \circ \langle \alpha, \tr{id}_O \rangle = \beta \circ t_\Gamma$, $p_{E \Gamma'} \circ \alpha = \gamma \circ p_{E \Gamma}$, and such that $\langle \tr{id}_I, \alpha, \tr{id}_O \rangle$ maps $\sim_\Gamma$ to $\sim_{\Gamma'}$.
\end{definition}

We will therefore work with equivalence classes of IODAGs under isomorphisms of IODAGs. For the sake of clarity, we will still call such an equivalence class an IODAG, and usually refer to it by specifying a representative of this class.

We can then explain how to compose IODAGs. First, we will need a way to compose equivalence relations which, contrary to the standard composition of corelations, does not forget about the intermediary set.

\begin{theorem} \label{NonForgettingComp}
Let $\sim_1$ and $\sim_2$ be equivalence relations respectively defined on $A \sqcup B$ and $B \sqcup C$, and whose restrictions to $B$ coincide. There exists a unique equivalence relation $\sim$ on $A \sqcup B \sqcup C$ such that $\sim$ reduces to $\sim_1$ on $A \sqcup B$, to $\sim_2$ on $B \sqcup C$, and to the composition of $\sim_1$ and $\sim_2$ (seen as corelations) on $A \sqcup C$. We will call $\sim$ the \textbf{non-forgetting composition} of $\sim_1$ and $\sim_2$.
\end{theorem}

\begin{proof}
To build such an equivalence relation, let us take a cospan $A \stackrel{f_A}{\to} X \stackrel{f_B}{\gets} B$ representing $\sim_1$, and a cospan $B \stackrel{g_B}{\to} Y \stackrel{g_C}{\gets} C$ representing $\sim_2$. We can take the pushout of $X \stackrel{f_B}{\gets} B \stackrel{g_B}{\to} Y$, given by $X \stackrel{i_1}{\to} Z \stackrel{i_2}{\gets} Y$. This yields an arrow $\langle i_1 \circ f_A, i_1 \circ f_B, i_2 \circ g_C \rangle : A \sqcup B \sqcup C \to Z$, which defines an equivalence relation $\sim$ on $A \sqcup B \sqcup C$. As this is the standard way to define compositions of corelations, it follows that $\sim$ reduces to the composition of $\sim_1$ and $\sim_2$ (seen as corelations) on $A \sqcup C$.

Let us prove that $\sim$ reduces to $\sim_1$ on $A \sqcup B$ and to $\sim_2$ on $B \sqcup C$. The pushout $Z$ is defined as the set of equivalence classes of $X \sqcup Y$ under the equivalence relation $\approx$ generated by the requirement: $x \approx y \iff \exists b, x = f_B(b) \land y = g_B(b)$. Yet, the fact that $\sim_1$ and $\sim_2$ coincide on $B$ implies that there exists a partition $B = \amalg_i B_i$ and families $(x_i)$, $(y_i)$ such that $\forall b \in B_i, f_B(b)= x_i \land g_B(b) = y_i$. Therefore, the equivalence classes of $\approx$ are the $\{x_i, y_i\}$ and the singletons $\{ w \}$ where $w \not\in f_B(B) \sqcup g_B(B)$. This implies that $i_1$ and $i_2$ are bijections. As $a_1$ is a bijection, two elements of $A \sqcup B$ are mapped to the same element of $X$ if and only if they are mapped to the same element of $Z$; thus $\sim$ restricts to $\sim_1$ on $A \sqcup B$. Symmetrically, it restricts to $\sim_2$ on $B \sqcup C$.

Let us finally prove uniqueness; suppose that $\sim'$ satisfies the same requirements and that there exist $d, d'$ such that $d \sim d'$ and $d \not\sim' d'$. Then, given that $\sim$ and $\sim'$ coincide on $A \sqcup B$ and on $B \sqcup C$, one must have $d \in A$ and $d' \in C$; this contradicts the fact that $\sim$ and $\sim'$ coincide on $A \sqcup C$.
\end{proof}

\begin{definition}
The sequential composition of two IODAGs $I \stackrel{\Gamma_1}{\to} J \stackrel{\Gamma_2}{\to} O$ is $\tilde{\Gamma}: I \to O$ defined by $\tilde{N} = N_1 \sqcup N_2$, $\tilde{E} = E_1 \sqcup J \sqcup E_2$, $\tilde{h} = \langle h_1, h_2 \rangle$, $\tilde{t} = \langle t_1, t_2 \rangle$, $K_{\tilde{E}} = K_{E_1} \sqcup K_J \sqcup K_{E_2}$, $\tilde{p}_{\tilde{E}} = \langle p_{E_1}, p_J, p_{E_2} \rangle$, and where $\sim_{\tilde{\Gamma}}$ is the non-forgetting composition of $\sim_1$ and $\sim_2$.
\end{definition}

\begin{theorem}
Sequential composition of IODAGs is associative.
\end{theorem}

\begin{proof}
The only non-trivial thing to check is associativity of the non-forgetting composition. This is ensured by the way we built it using cospans and pushouts in the proof of Theorem \ref{NonForgettingComp}, and the fact that pushouts are unique up to isomorphism.
\end{proof}

\begin{definition}
Given two indexed wire systems $I$ and $I'$, their parallel composition is given by  the set $I \sqcup I'$  and the structure $K_{I \sqcup I'} := K_I \sqcup K_{I'}$, $p_{I \sqcup I'} := \langle p_I, p_{I'} \rangle$ and $\sim_{I \sqcup I'} := \sim_I \sqcup \sim_{I'}$, defined by the fact that it does not relate any elements of $K_I$ and $K_{I'}$ and that it restricts to $\sim_I$ and $\sim_{I'}$ respectively on $K_I$ and $K_{I'}$.

Similarly, the parallel composition of $\Gamma : I \to O$ and $\Gamma' : I' \to O'$ is the IODAG $\Gamma \sqcup \Gamma' : I \sqcup I' \to O \sqcup O'$ given by taking disjoint unions on all of the relevant structure and defining the new equivalence relation in the same way.
\end{definition}

The following is then direct.

\begin{theorem}
The parallel composition of IODAGs is associative, and distributes over sequential composition.
\end{theorem}

Note, however, that the theory of IODAGs, thus defined, does not form a symmetric monoidal category, as it lacks identity morphisms. This can be dealt with by extending the definition of IODAGs, in order to allow for \textit{empty nodes}.

\begin{definition}
One can extend the definition of IODAGs by further equipping them with a set $\dot{N}_\Gamma \subseteq N_\Gamma$ of empty nodes, such that for a given $n \in \dot{N}_\Gamma$, $n$ has only one ingoing wire $\tr{in}(n)$, and one outgoing wire $\tr{out}(n)$, whose indices are related in a consistent way: i.e.\ there exists a bijection $\xi^\Gamma_n : p_\Gamma^{-1} (\tr{in}(n)) \to p_\Gamma^{-1} (\tr{out}(n))$ such that $\forall k \in p_\Gamma^{-1} (\tr{in}(n)), k \sim_\Gamma \xi^\Gamma_n (k)$. One can further redefine a IODAG to be an equivalence class under the rewriting operations which consist in getting rid of some empty nodes and identifying their ingoing wire with their outgoing wire.
\end{definition}

\begin{theorem}
The theory of IODAGs with possibly empty nodes is a symmetric monoidal category. 
\end{theorem}

\begin{proof}
The non-trivial part is to prove that this theory has identity morphisms and swaps. The identity morphism from $I$ to itself is given by the IODAG with no inner edges, $\abs{I}$ empty nodes, each of which connects an element of $I$ in the inputs with its counterpart in the outputs, and such that two elements $k, k' \in K_I \sqcup K_I$ are related if and only if they are related as elements of $K_I$. The swap from $I \sqcup J$ to $J \sqcup I$ is built in an analogous way.
\end{proof}

We can now single out these IODAGs which are well indexed, for an interpretation in practical isometries and for one in practical unitaries.

\begin{definition}\label{isoDot}
Let $\Gamma: I \to O$ be a IODAG. For a given node $n$ of $\Gamma$, the set of incoming indices for this node is $p_\Gamma^{-1} \circ h^{-1} (n)$, and the set of outgoing indices for this node is $p_\Gamma^{-1} \circ t^{-1} (n)$.

Let $c$ be an equivalence class of $K_I \sqcup K_{E_\Gamma} \sqcup K_O$ under $\sim_\Gamma$, a \textbf{starting point} for $c$ is a node $n$ such that $c$ has at least one representative in the set of outgoing indices of $n$, but no representatives in its incoming indices. An \textbf{endpoint} for $c$ is defined symmetrically.

$\Gamma$ is an \textbf{iso-IODAG} if each equivalence class $c$ of $K_I \sqcup K_{E_\Gamma} \sqcup K_O$ under $\sim_\Gamma$ has at most one starting point, and has no starting point if it appears in the inputs of the diagram (i.e.\ if it has a representative in $K_I$).

$\Gamma$ is a \textbf{uni-IODAG} if it is an iso-IODAG in which each equivalence class $c$ of $K_I \sqcup K_{E_\Gamma} \sqcup K_O$ under $\sim_\Gamma$ has at most one endpoint, and has no endpoint if it appears in the outputs of the diagram (i.e.\ if it has a representative in $K_O$).
\end{definition}

Being well-indexed is a property preserved by composing diagrams, sequentially and in parallel:

\begin{theorem}
Sequential and parallel compositions of iso-IODAGs are iso-IODAGs, and sequential and parallel compositions of uni-IODAGs are uni-IODAGs.
\end{theorem}

\begin{proof}
Let us take two iso-IODAGs $I \stackrel{\Gamma_1}{\to} J \stackrel{\Gamma_2}{\to} O$ and look at their sequential composition $\tilde{\Gamma}$. Take an equivalence class of indices $c$. Then the fact that $\sim_{\tilde{\Gamma}}$ reduces to $\sim_1$ on $K_I \sqcup K_J$ and to $\sim_2$ on $K_J \sqcup K_O$ implies that:

\begin{itemize}
\item if $c$ has a representative in $K_J$, then there is an equivalence class $c_1$ under $\sim_1$ and an equivalence class $c_2$ under $\sim_2$ which correspond to $c$ (i.e., an element belongs to one of these classes if and only if it belongs to $c$). As $\Gamma_2$ is an iso-IODAG and $c_2$ has a representative in $K_J$, it has no starting point in $\Gamma_2$, and thus neither does $c$. If $c$ has a representative in $K_I$, then so does $c_1$, which therefore has no starting point in $\Gamma_1$; $c$ then has no starting point in $\Gamma$. Otherwise, $c_1$ has one starting point, and thus so does $c$.
\item if $c$ has no representative in $K_J$, then its representatives are either all in $\Gamma_1$ or all in $\Gamma_2$; as both are iso-IODAGs, $c$ therefore satisfies the conditions of Definition $\ref{isoDot}$.
\end{itemize}

$\tilde{\Gamma}$ is therefore an iso-IODAG. The rest of the proof for uni-IODAGs is symmetric.

For parallel compositions, as one has $\sim_{\Gamma_1 \sqcup \Gamma_2} = \sim_1 \sqcup \sim_2$, the set of representatives of a given class is included in one of the two diagrams; thus, that the requirements of Definition \ref{isoDot} are satisfied by $\Gamma_1$ and by $\Gamma_2$ directly implies that they are satisfied by $\Gamma_1 \sqcup \Gamma_2$.
\end{proof}

\subsection{Interpretation}

Finally, it is time to turn ourselves to interpretations of IODAGs in terms of index-matching routed maps. First, the index-matching routes are, in fact, defined even before having to interpret anything: they are given by the IODAGs alone.

\begin{definition} \label{CorelInterp}
Let $\Gamma$ be an iso-IODAG. For each edge $e \in I \sqcup E_\Gamma \sqcup O$, we define its set of indices $\cx_e := p_\Gamma^{-1} (e)$. For each node $n \in N_\Gamma$, we define its corelation $\kappa_n : \amalg_{e \in h^{-1}(n)} \cx_e \to  \amalg_{e' \in t^{-1}(n)} \cx_{e'}$ by the fact that it relates two elements if and only if these are related under $\sim_\Gamma$. Furthermore, we define a pre-processing corelation $\kappa_\tr{pre}$ from $K_I$ to itself by the requirement that it relates two elements of $K_I \sqcup K_I$ if and only if, when considered as elements of $K_I$, they are related by $\sim_I$.
\end{definition}

Before interpreting, we can show that the compositions of such corelations are suitable.

\begin{lemma} \label{SuitableDot}
Given an iso-IODAG (resp.\ a uni-IODAG) $\Gamma$, let us consider the corelation obtained by composing all of its nodes' corelations according to $\Gamma$'s graph, then pre-composing the result with $\kappa^\tr{pre}$. It is equal to $\kappa_\tr{tot}$, the corelation which relates two elements of $K_I \sqcup K_O$ if and only if they are related by $\sim_\Gamma$. Furthermore, all the compositions in this construction are suitable for isometries (resp.\ for unitaries).
\end{lemma}

\begin{proof}
We call $\kappa_\tr{tot}'$ the corelation thus built; let us prove that it is equal to $\kappa_\tr{tot}$. If we take two elements $k$, $k'$ of $K_I$ which are related by $\kappa_\tr{tot}$, then they are related by $\kappa_\tr{pre}$ and therefore also by $\kappa_\tr{tot}'$; the same holds if we take two elements of $K_O$  which are related by $\kappa_\tr{tot}$. If we take $k \in K_I$ and $k' \in K_O$ related by $\kappa_\tr{tot}$, this means that they belong to a same equivalence class; as $\Gamma$ is an iso-IODAG, this implies that this equivalence class has no starting point. There is therefore necessarily a path of related indices going downwards from $k'$ to at least one index $k''$ in the inputs of the diagram, which implies that $k'$ and $k''$ are related by $\kappa_\tr{tot}'$. As $k''$ is itself related to $k$ by $\kappa_\tr{pre}$ and therefore also by $\kappa_\tr{tot}'$, transitivity allows to conclude that $k$ and $k'$ are related by $\kappa_\tr{tot}'$.

Reciprocally, if we take two elements of $K_I \sqcup K_O$ which are related by $\kappa_\tr{tot}'$, this means that either they are related by $\kappa_\tr{pre}$, or we can find a path of related indices connecting them through the graph of $\Gamma$. In the first case, they are clearly related by $\kappa_\tr{tot}$; in the second case, as each of the index-matching route maps was obtained through a restriction of $\sim_\Gamma$, and as $\sim_\Gamma$ is transitive, this means that these two indices are also related by $\sim_\Gamma$, and therefore also by $\kappa_\tr{tot}$.

Let us now prove, for the isometric case, that the compositions are suitable for isometries. One can build $\kappa_\tr{tot}'$ by foliating the graph, then composing the corelations layer by layer, starting with the pre-processing. The composition of the pre-processing with the first layer is suitable for isometries by Theorem \ref{SuitableIsoMatch}, as there are no indices created by the pre-processing. Say we have composed the pre-processing and the $m$ first layers, yielding a corelation $\kappa_m$; its composition with layer $m+1$ is suitable for isometries as well, again by Theorem \ref{SuitableIsoMatch}; indeed, an index created by $\kappa_m$ is an index which is not present in the inputs of the diagram; it therefore has a starting point in the $m$ first layers. This entails that the set of wires bearing this index is connected in the graph of $\Gamma$; therefore, all elements of the corresponding equivalence class in the outputs of the $m$ first layers are related by $\kappa_n$. The $m+1$ layer therefore cannot relate these elements with others, as it has to respect $\sim_\Gamma$. The rest of the proof in the unitary case is similar.
\end{proof}

\begin{definition}
An \textbf{interpretation} of an iso-IODAG (resp.\ a uni-IODAG) $\Gamma : I \to O$ in practically isometric (resp.\ practically unitary) IMRMs consists of the following:

\begin{itemize}
\item a function $\tr{length} : K_I \sqcup K_{E_\Gamma} \sqcup K_O \to \mathbb{N}$, satisfying $k \sim_\Gamma k' \implies \tr{length}(k) = \tr{length}(k')$;
\item a function $\tr{sys}$ which, to any $e \in I \sqcup E_\Gamma \sqcup O$, associates a partitioned Hilbert space of the form $(\ch_A, \bar{X}_e, (\pi^{\vec{k}})_{\vec{k} \in \bar{X}_e})$, where $\bar{X}_e := \bigtimes_{x \in \cx_e} \llbracket 1, \tr{length}(x) \rrbracket$;
\item a function $\tr{morph}$ which, to any $n \in N_\Gamma$, associates a practically isometric (resp. practically unitary) index-matching routed map of the form $(\kappa_n, f)$ from $\bigotimes_{e \in h^{-1}(n)} \tr{sys}(e)$ to $\bigotimes_{e' \in t^{-1}(n)} \tr{sys}(e')$.
\end{itemize}

In addition, the input and output wires of an empty node must have the same interpretation and the interpretation of the empty node must be an identity morphism.
\end{definition}

Interpreted IODAGs are called index-matching quantum circuits. The global index-matching routed map that an index-matching quantum circuit represents will be called its \textit{meaning}.

\begin{definition}
Given an interpretation $(\tr{length}, \tr{sys}, \tr{morph})$ of an iso-IODAG, we define a pre-processing map $(\kappa^\tr{pre}, \pi^\tr{pre})$ from $\bigotimes_{e \in I} \tr{sys}(e)$ to itself, where $\kappa_\tr{pre}$ was defined in Definition \ref{CorelInterp}, and $\pi^\tr{pre} := \sum_{\vec{k} \in \bigtimes_{e \in I} \bar{\cx}_e}  \bar{\kappa}_\tr{pre}^{\vec{k}, \vec{k}} \pi^{\vec{k}}$.

The \textbf{meaning} of $(\tr{length}, \tr{sys}, \tr{morph})$ is then $F \circ (\kappa_\tr{pre}, \pi^\tr{pre})$, where $F$ is the index-matching routed map obtained by composing the $\tr{morph}(n)$ according to the graph of $\Gamma$.
\end{definition}

\begin{theorem} \label{th:InterpretationsMeaning}
Given an interpretation of an iso-IODAG (resp.\ of a uni-IODAG) in practically isometric (resp. in practically unitary) IMRMs, its meaning is a practical isometry (resp.\ a practical unitary).
\end{theorem}

\begin{proof}
This follows directly from Lemma \ref{SuitableDot} and Theorems \ref{th:IsoCompos} and \ref{th:UniCompos}.
\end{proof}

Interpreting also plays well with sequential and parallel compositions.

\begin{theorem} \label{th:SeqCompInterpretation}
Let $I \stackrel{\Gamma_1}{\to} J \stackrel{\Gamma_2}{\to} O$ be two iso-IODAGs, whose sequential composition is noted $\tilde{\Gamma}$, and let $(\tr{length}_1, \tr{sys}_1, \tr{morph}_1)$ and $(\tr{length}_2, \tr{sys}_2, \tr{morph}_2)$ be respective interpretations which agree on $J$, i.e.\ $\forall k \in K_J, \tr{length}_1(k) = \tr{length}_2(k)$ and $\forall e \in E_J, \tr{sys}_1(e) = \tr{sys}_2(e)$. Then the sequential composition of their meanings is equal to the meaning of the interpretation of $\tilde{\Gamma}$ obtained by combining them.
\end{theorem}

\begin{proof}
Let $\tr{mean}_1 = F_1 \circ (\kappa^\tr{pre}_1, \pi^\tr{pre}_1)$ and $\tr{mean}_2 = F_2 \circ (\kappa^\tr{pre}_2, \pi^\tr{pre}_2)$ be the respective meanings. By Lemma \ref{SuitableDot}, the index-matching route of $\tr{mean}_1$ is $\kappa^\tr{tot}_1$. In particular, for two elements of $J$ related by $\sim_J$, they are related by $\kappa^\tr{tot}_1$; therefore, $\tr{mean}_1 = (\kappa^\tr{pre}_2, \pi^\tr{pre}_2) \circ \tr{mean}_1$. This implies that $\tr{mean}_2 \circ \tr{mean}_1 = F_2 \circ F_1 \circ (\kappa^\tr{pre}_1, \pi^\tr{pre}_1)$, which is the meaning of the corresponding interpretation of $\tilde{\Gamma}$.
\end{proof}

\begin{theorem} \label{th:ParrCompInterpretatiom}
Let $\Gamma$ and $\Gamma'$ be two iso-IODAGs, with respective interpretations $(\tr{length}_1, \tr{sys}_1, \tr{morph}_1)$ and $(\tr{length}_2, \tr{sys}_2, \tr{morph}_2)$. The parallel composition of their meanings is equal to the meaning of the interpretation $(\langle \tr{length}_1, \tr{length}_2 \rangle, \langle \tr{sys}_1, \tr{sys}_2 \rangle, \langle \tr{morph}_1, \tr{morph}_2 \rangle)$ of $\Gamma \sqcup \Gamma'$.
\end{theorem}

\begin{proof}
Direct.
\end{proof}

\end{document}

%% file: tikzstyles.tex
\usetikzlibrary{shapes.multipart}

\tikzstyle{bwSpider}=[
       rectangle split,
       rectangle split parts=2,
       rectangle split part fill={black,white},
 minimum size=3.6 mm, inner sep=-2mm, draw=black,scale=0.5,rounded corners=0.8 mm
       ]
 \tikzstyle{wbSpider}=[
       rectangle split,
       rectangle split parts=2,
       rectangle split part fill={white,black},
 minimum size=3.6 mm, inner sep=-2mm, draw=black,scale=0.5,rounded corners=0.8 mm
       ]
\tikzstyle{qWire}=[line width = 1pt, black]
\tikzstyle{cWire}=[color=gray,thin]
\tikzstyle{env}=[copoint,regular polygon rotate=0,minimum width=0.2cm, fill=black]

\tikzstyle{probs}=[shape=semicircle,fill=white,draw=black,shape border rotate=180,minimum width=1.2cm]

\tikzstyle{disc}=[circuit ee IEC,thick,ground,rotate=90,scale=1.5]

\tikzstyle{every picture}=[baseline=-0.25em,scale=0.5]
\tikzstyle{dotpic}=[] 
\tikzstyle{diredges}=[every to/.style={diredge}]
\tikzstyle{math matrix}=[matrix of math nodes,left delimiter=(,right delimiter=),inner sep=2pt,column sep=1em,row sep=0.5em,nodes={inner sep=0pt},text height=1.5ex, text depth=0.25ex]

\tikzstyle{inline text}=[text height=1.5ex, text depth=0.25ex,yshift=0.5mm]
\tikzstyle{label}=[font=\footnotesize,text height=1.5ex, text depth=0.25ex,yshift=0.5mm]
\tikzstyle{left label}=[label,anchor=east,xshift=1.5mm]
\tikzstyle{right label}=[label,anchor=west,xshift=-0.5mm]

\tikzstyle{braceedge}=[decorate,decoration={brace,amplitude=2mm,raise=-1mm}]
\tikzstyle{small braceedge}=[decorate,decoration={brace,amplitude=1mm,raise=-1mm}]

\tikzstyle{doubled}=[line width=1.6pt] 
\tikzstyle{boldedge}=[doubled,shorten <=-0.17mm,shorten >=-0.17mm]
\tikzstyle{boldedgegray}=[doubled,gray,shorten <=-0.17mm,shorten >=-0.17mm]
\tikzstyle{singleedgegray}=[gray]

\tikzstyle{semidoubled}=[line width=1.4pt] 
\tikzstyle{semiboldedgegray}=[semidoubled,gray,shorten <=-0.17mm,shorten >=-0.17mm]

\tikzstyle{boxedge}=[semiboldedgegray]

\tikzstyle{boldedgedashed}=[very thick,dashed,shorten <=-0.17mm,shorten >=-0.17mm]
\tikzstyle{vboldedgedashed}=[doubled,dashed,shorten <=-0.17mm,shorten >=-0.17mm]
\tikzstyle{left hook arrow}=[left hook-latex]
\tikzstyle{right hook arrow}=[right hook-latex]
\tikzstyle{sembracket}=[line width=0.5pt,shorten <=-0.07mm,shorten >=-0.07mm]

\tikzstyle{causal edge}=[->,thick,gray]
\tikzstyle{causal nondir}=[thick,gray]
\tikzstyle{timeline}=[thick,gray, dashed]

\tikzstyle{cedge}=[<->,thick,gray!70!white]

\tikzstyle{empty diagram}=[draw=gray!40!white,dashed,shape=rectangle,minimum width=1cm,minimum height=1cm]
\tikzstyle{empty diagram small}=[draw=gray!50!white,dashed,shape=rectangle,minimum width=0.6cm,minimum height=0.5cm]

\tikzstyle{dot}=[inner sep=0mm,minimum width=2mm,minimum height=2mm,draw,shape=circle]

\tikzstyle{leak}=[white dot, shape=regular polygon, minimum size=3.3 mm, regular polygon sides=3, outer sep=-0.2mm, regular polygon rotate=270]
\tikzstyle{proj}=[regular polygon,regular polygon sides=4,draw,scale=0.75,inner sep=-0.5pt,minimum width=6mm,fill=white]
\tikzstyle{projOut}=[regular polygon,regular polygon sides=3,draw,scale=0.75,inner sep=-0.5pt,minimum width=7.5mm,fill=white,regular polygon rotate=180]
\tikzstyle{projIn}=[regular polygon,regular polygon sides=3,draw,scale=0.75,inner sep=-0.5pt,minimum width=7.5mm,fill=white]
\tikzstyle{Vleak}=[white dot, shape=regular polygon, minimum size=3.3 mm, regular polygon sides=3, outer sep=-0.2mm, regular polygon rotate=90]
\tikzstyle{dleak}=[white dot, line width=1.6pt, shape=regular polygon, minimum size=3.3 mm, regular polygon sides=3, outer sep=-0.2mm, regular polygon rotate=270]

\tikzstyle{Wsquare}=[white dot, shape=regular polygon, rounded corners=0.8 mm, minimum size=3.3 mm, regular polygon sides=3, outer sep=-0.2mm]
\tikzstyle{Wsquareadj}=[white dot, shape=regular polygon, rounded corners=0.8 mm, minimum size=3.3 mm, regular polygon sides=3, outer sep=-0.2mm, regular polygon rotate=180]
\tikzstyle{ddot}=[inner sep=0mm, doubled, minimum width=2.5mm,minimum height=2.5mm,draw,shape=circle]

\tikzstyle{black dot}=[dot,fill=black]
\tikzstyle{white dot}=[dot,fill=white,,text depth=-0.2mm]
\tikzstyle{white Wsquare}=[Wsquare,fill=gray,,text depth=-0.2mm]
\tikzstyle{white Wsquareadj}=[Wsquareadj,fill=white,,text depth=-0.2mm]
\tikzstyle{green dot}=[white dot] 
\tikzstyle{gray dot}=[dot,fill=gray!40!white,,text depth=-0.2mm]
\tikzstyle{red dot}=[gray dot] 

\tikzstyle{black ddot}=[ddot,fill=black]
\tikzstyle{white ddot}=[ddot,fill=white]
\tikzstyle{gray ddot}=[ddot,fill=gray!40!white]

\tikzstyle{gray edge}=[gray!60!white]

\tikzstyle{small dot}=[inner sep=0.5mm,minimum width=0pt,minimum height=0pt,draw,shape=circle]

\tikzstyle{small black dot}=[small dot,fill=black]
\tikzstyle{small white dot}=[small dot,fill=white]
\tikzstyle{small gray dot}=[small dot,fill=gray!40!white]

\tikzstyle{causal dot}=[inner sep=0.4mm,minimum width=0pt,minimum height=0pt,draw=white,shape=circle,fill=gray!40!white]

\tikzstyle{phase dimensions}=[minimum size=5mm,font=\footnotesize,rectangle,rounded corners=2.5mm,inner sep=0.2mm,outer sep=-2mm]
\tikzstyle{dphase dimensions}=[minimum size=5mm,font=\footnotesize,rectangle,rounded corners=2.5mm,inner sep=0.2mm,outer sep=-2mm]

\tikzstyle{white phase dot}=[dot,fill=white,phase dimensions]
\tikzstyle{white phase ddot}=[ddot,fill=white,dphase dimensions]

\tikzstyle{white rect ddot}=[draw=black,fill=white,doubled,minimum size=5mm,font=\footnotesize,rectangle,rounded corners=2.5mm,inner sep=0.2mm]
\tikzstyle{gray rect ddot}=[draw=black,fill=gray!40!white,doubled,minimum size=6mm,font=\footnotesize,rectangle,rounded corners=3mm]

\tikzstyle{gray phase dot}=[dot,fill=gray!40!white,phase dimensions]
\tikzstyle{gray phase ddot}=[ddot,fill=gray!40!white,dphase dimensions]
\tikzstyle{grey phase dot}=[gray phase dot]
\tikzstyle{grey phase ddot}=[gray phase ddot]

\tikzstyle{small phase dimensions}=[minimum size=4mm,font=\tiny,rectangle,rounded corners=2mm,inner sep=0.2mm,outer sep=-2mm]
\tikzstyle{small dphase dimensions}=[minimum size=4mm,font=\tiny,rectangle,rounded corners=2mm,inner sep=0.2mm,outer sep=-2mm]

\tikzstyle{small gray phase dot}=[dot,fill=gray!40!white,small phase dimensions]
\tikzstyle{small gray phase ddot}=[ddot,fill=gray!40!white,small dphase dimensions]

\tikzstyle{small map}=[draw,shape=rectangle,minimum height=4mm,minimum width=4mm,fill=white]

\tikzstyle{cnot}=[fill=white,shape=circle,inner sep=-1.4pt]

\tikzstyle{asym hadamard}=[fill=white,draw,shape=NEbox,inner sep=0.6mm,font=\footnotesize,minimum height=4mm]
\tikzstyle{asym hadamard conj}=[fill=white,draw,shape=NWbox,inner sep=0.6mm,font=\footnotesize,minimum height=4mm]
\tikzstyle{asym hadamard dag}=[fill=white,draw,shape=SEbox,inner sep=0.6mm,font=\footnotesize,minimum height=4mm]

\tikzstyle{hadamard}=[fill=white,draw,inner sep=0.6mm,font=\footnotesize,minimum height=4mm,minimum width=4mm]
\tikzstyle{small hadamard}=[fill=white,draw,inner sep=0.6mm,minimum height=1.5mm,minimum width=1.5mm]
\tikzstyle{small hadamard rotate}=[small hadamard,rotate=45]
\tikzstyle{dhadamard}=[hadamard,doubled]
\tikzstyle{small dhadamard}=[small hadamard,doubled]
\tikzstyle{small dhadamard rotate}=[small hadamard rotate,doubled]
\tikzstyle{antipode}=[white dot,inner sep=0.3mm,font=\footnotesize]

\tikzstyle{scalar}=[diamond,draw,inner sep=0.5pt,font=\small]
\tikzstyle{dscalar}=[diamond,doubled, draw,inner sep=0.5pt,font=\small]

\tikzstyle{small box}=[rectangle,inline text,fill=white,draw,minimum height=5mm,yshift=-0.5mm,minimum width=5mm,font=\small]
\tikzstyle{small gray box}=[small box,fill=gray!30]
\tikzstyle{medium box}=[rectangle,inline text,fill=white,draw,minimum height=5mm,yshift=-0.5mm,minimum width=10mm,font=\small]
\tikzstyle{square box}=[small box] 
\tikzstyle{medium gray box}=[small box,fill=gray!30]
\tikzstyle{semilarge box}=[rectangle,inline text,fill=white,draw,minimum height=5mm,yshift=-0.5mm,minimum width=12.5mm,font=\small]
\tikzstyle{large box}=[rectangle,inline text,fill=white,draw,minimum height=5mm,yshift=-0.5mm,minimum width=15mm,font=\small]
\tikzstyle{large gray box}=[small box,fill=gray!30]

\tikzstyle{Bayes box}=[rectangle,fill=black,draw, minimum height=3mm, minimum width=3mm]

\tikzstyle{gray square point}=[small box,fill=gray!50]

\tikzstyle{dphase box white}=[dhadamard]
\tikzstyle{dphase box gray}=[dhadamard,fill=gray!50!white]
\tikzstyle{phase box white}=[hadamard]
\tikzstyle{phase box gray}=[hadamard,fill=gray!50!white]

\tikzstyle{point}=[regular polygon,regular polygon sides=3,draw,scale=0.75,inner sep=-0.5pt,minimum width=9mm,fill=white,regular polygon rotate=180]
\tikzstyle{point nosep}=[regular polygon,regular polygon sides=3,draw,scale=0.75,inner sep=-2pt,minimum width=9mm,fill=white,regular polygon rotate=180]
\tikzstyle{copoint}=[regular polygon,regular polygon sides=3,draw,scale=0.75,inner sep=-0.5pt,minimum width=9mm,fill=white]
\tikzstyle{dpoint}=[point,doubled]
\tikzstyle{dcopoint}=[copoint,doubled]

\tikzstyle{pointgrow}=[shape=cornerpoint,kpoint common,scale=0.75,inner sep=3pt]
\tikzstyle{pointgrow dag}=[shape=cornercopoint,kpoint common,scale=0.75,inner sep=3pt]

\tikzstyle{wide copoint}=[fill=white,draw,shape=isosceles triangle,shape border rotate=90,isosceles triangle stretches=true,inner sep=0pt,minimum width=1.5cm,minimum height=6.12mm]
\tikzstyle{wide point}=[fill=white,draw,shape=isosceles triangle,shape border rotate=-90,isosceles triangle stretches=true,inner sep=0pt,minimum width=1.5cm,minimum height=6.12mm,yshift=-0.0mm]
\tikzstyle{wide point plus}=[fill=white,draw,shape=isosceles triangle,shape border rotate=-90,isosceles triangle stretches=true,inner sep=0pt,minimum width=1.74cm,minimum height=7mm,yshift=-0.0mm]

\tikzstyle{wide dpoint}=[fill=white,doubled,draw,shape=isosceles triangle,shape border rotate=-90,isosceles triangle stretches=true,inner sep=0pt,minimum width=1.5cm,minimum height=6.12mm,yshift=-0.0mm]

\tikzstyle{tinypoint}=[regular polygon,regular polygon sides=3,draw,scale=0.55,inner sep=-0.15pt,minimum width=6mm,fill=white,regular polygon rotate=180]

\tikzstyle{white point}=[point]
\tikzstyle{white dpoint}=[dpoint]
\tikzstyle{green point}=[white point] 
\tikzstyle{white copoint}=[copoint]
\tikzstyle{gray point}=[point,fill=gray!40!white]
\tikzstyle{gray dpoint}=[gray point,doubled]
\tikzstyle{red point}=[gray point] 
\tikzstyle{gray copoint}=[copoint,fill=gray!40!white]
\tikzstyle{gray dcopoint}=[gray copoint,doubled]

\tikzstyle{white point guide}=[regular polygon,regular polygon sides=3,font=\scriptsize,draw,scale=0.65,inner sep=-0.5pt,minimum width=9mm,fill=white,regular polygon rotate=180]

\tikzstyle{black point}=[point,fill=black,font=\color{white}]
\tikzstyle{black copoint}=[copoint,fill=black,font=\color{white}]

\tikzstyle{tiny gray point}=[tinypoint,fill=gray!40!white]

\tikzstyle{diredge}=[->]
\tikzstyle{ddiredge}=[<->]
\tikzstyle{rdiredge}=[<-]
\tikzstyle{thickdiredge}=[->, very thick]
\tikzstyle{pointer edge}=[->,very thick,gray]
\tikzstyle{pointer edge part}=[very thick,gray]
\tikzstyle{dashed edge}=[dashed]
\tikzstyle{thick dashed edge}=[very thick,dashed]
\tikzstyle{thick gray dashed edge}=[thick dashed edge,gray!40]
\tikzstyle{thick map edge}=[very thick,|->]

\makeatletter
\newcommand{\boxshape}[3]{%
\pgfdeclareshape{#1}{
\inheritsavedanchors[from=rectangle] 
\inheritanchorborder[from=rectangle]
\inheritanchor[from=rectangle]{center}
\inheritanchor[from=rectangle]{north}
\inheritanchor[from=rectangle]{south}
\inheritanchor[from=rectangle]{west}
\inheritanchor[from=rectangle]{east}
\backgroundpath{
\southwest \pgf@xa=\pgf@x \pgf@ya=\pgf@y
\northeast \pgf@xb=\pgf@x \pgf@yb=\pgf@y

\@tempdima=#2
\@tempdimb=#3

\pgfpathmoveto{\pgfpoint{\pgf@xa - 5pt + \@tempdima}{\pgf@ya}}
\pgfpathlineto{\pgfpoint{\pgf@xa - 5pt - \@tempdima}{\pgf@yb}}
\pgfpathlineto{\pgfpoint{\pgf@xb + 5pt + \@tempdimb}{\pgf@yb}}
\pgfpathlineto{\pgfpoint{\pgf@xb + 5pt - \@tempdimb}{\pgf@ya}}
\pgfpathlineto{\pgfpoint{\pgf@xa - 5pt + \@tempdima}{\pgf@ya}}
\pgfpathclose
}
}}

\boxshape{NEbox}{0pt}{3pt}
\boxshape{SEbox}{0pt}{-3pt}
\boxshape{NWbox}{5pt}{0pt}
\boxshape{SWbox}{-5pt}{0pt}
\boxshape{EBox}{-3pt}{3pt}
\boxshape{WBox}{3pt}{-3pt}
\makeatother

\tikzstyle{cloud}=[shape=cloud,draw,minimum width=1.5cm,minimum height=1.5cm]

\tikzstyle{map}=[draw,shape=NEbox,inner sep=1pt,minimum height=4mm,fill=white]
\tikzstyle{dashedmap}=[draw,dashed,shape=NEbox,inner sep=2pt,minimum height=6mm,fill=white]
\tikzstyle{mapdag}=[draw,shape=SEbox,inner sep=1pt,minimum height=4mm,fill=white]
\tikzstyle{mapadj}=[draw,shape=SEbox,inner sep=2pt,minimum height=6mm,fill=white]
\tikzstyle{maptrans}=[draw,shape=SWbox,inner sep=2pt,minimum height=6mm,fill=white]
\tikzstyle{mapconj}=[draw,shape=NWbox,inner sep=2pt,minimum height=6mm,fill=white]

\tikzstyle{medium map}=[draw,shape=NEbox,inner sep=2pt,minimum height=6mm,fill=white,minimum width=7mm]
\tikzstyle{medium map dag}=[draw,shape=SEbox,inner sep=2pt,minimum height=6mm,fill=white,minimum width=7mm]
\tikzstyle{medium map adj}=[draw,shape=SEbox,inner sep=2pt,minimum height=6mm,fill=white,minimum width=7mm]
\tikzstyle{medium map trans}=[draw,shape=SWbox,inner sep=2pt,minimum height=6mm,fill=white,minimum width=7mm]
\tikzstyle{medium map conj}=[draw,shape=NWbox,inner sep=2pt,minimum height=6mm,fill=white,minimum width=7mm]
\tikzstyle{semilarge map}=[draw,shape=NEbox,inner sep=2pt,minimum height=6mm,fill=white,minimum width=9.5mm]
\tikzstyle{semilarge map trans}=[draw,shape=SWbox,inner sep=2pt,minimum height=6mm,fill=white,minimum width=9.5mm]
\tikzstyle{semilarge map adj}=[draw,shape=SEbox,inner sep=2pt,minimum height=6mm,fill=white,minimum width=9.5mm]
\tikzstyle{semilarge map dag}=[draw,shape=SEbox,inner sep=2pt,minimum height=6mm,fill=white,minimum width=9.5mm]
\tikzstyle{semilarge map conj}=[draw,shape=NWbox,inner sep=2pt,minimum height=6mm,fill=white,minimum width=9.5mm]
\tikzstyle{large map}=[draw,shape=NEbox,inner sep=2pt,minimum height=6mm,fill=white,minimum width=12mm]
\tikzstyle{large map conj}=[draw,shape=NWbox,inner sep=2pt,minimum height=6mm,fill=white,minimum width=12mm]
\tikzstyle{very large map}=[draw,shape=NEbox,inner sep=2pt,minimum height=6mm,fill=white,minimum width=17mm]

\tikzstyle{medium dmap}=[draw,doubled,shape=NEbox,inner sep=2pt,minimum height=6mm,fill=white,minimum width=7mm]
\tikzstyle{medium dmap dag}=[draw,doubled,shape=SEbox,inner sep=2pt,minimum height=6mm,fill=white,minimum width=7mm]
\tikzstyle{medium dmap adj}=[draw,doubled,shape=SEbox,inner sep=2pt,minimum height=6mm,fill=white,minimum width=7mm]
\tikzstyle{medium dmap trans}=[draw,doubled,shape=SWbox,inner sep=2pt,minimum height=6mm,fill=white,minimum width=7mm]
\tikzstyle{medium dmap conj}=[draw,doubled,shape=NWbox,inner sep=2pt,minimum height=6mm,fill=white,minimum width=7mm]
\tikzstyle{semilarge dmap}=[draw,doubled,shape=NEbox,inner sep=2pt,minimum height=6mm,fill=white,minimum width=9.5mm]
\tikzstyle{semilarge dmap trans}=[draw,doubled,shape=SWbox,inner sep=2pt,minimum height=6mm,fill=white,minimum width=9.5mm]
\tikzstyle{semilarge dmap adj}=[draw,doubled,shape=SEbox,inner sep=2pt,minimum height=6mm,fill=white,minimum width=9.5mm]
\tikzstyle{semilarge dmap dag}=[draw,doubled,shape=SEbox,inner sep=2pt,minimum height=6mm,fill=white,minimum width=9.5mm]
\tikzstyle{semilarge dmap conj}=[draw,doubled,shape=NWbox,inner sep=2pt,minimum height=6mm,fill=white,minimum width=9.5mm]
\tikzstyle{large dmap}=[draw,doubled,shape=NEbox,inner sep=2pt,minimum height=6mm,fill=white,minimum width=12mm]
\tikzstyle{large dmap conj}=[draw,doubled,shape=NWbox,inner sep=2pt,minimum height=6mm,fill=white,minimum width=12mm]
\tikzstyle{large dmap trans}=[draw,doubled,shape=SWbox,inner sep=2pt,minimum height=6mm,fill=white,minimum width=12mm]
\tikzstyle{large dmap adj}=[draw,doubled,shape=SEbox,inner sep=2pt,minimum height=6mm,fill=white,minimum width=12mm]
\tikzstyle{large dmap dag}=[draw,doubled,shape=SEbox,inner sep=2pt,minimum height=6mm,fill=white,minimum width=12mm]
\tikzstyle{very large dmap}=[draw,doubled,shape=NEbox,inner sep=2pt,minimum height=6mm,fill=white,minimum width=19.5mm]

\tikzstyle{muxbox}=[draw,shape=rectangle,minimum height=3mm,minimum width=3mm,fill=white]
\tikzstyle{dmuxbox}=[muxbox,doubled]

\tikzstyle{box}=[draw,shape=rectangle,inner sep=2pt,minimum height=6mm,minimum width=6mm,fill=white]
\tikzstyle{dbox}=[draw,doubled,shape=rectangle,inner sep=2pt,minimum height=6mm,minimum width=6mm,fill=white]
\tikzstyle{dmap}=[draw,doubled,shape=NEbox,inner sep=2pt,minimum height=6mm,fill=white]
\tikzstyle{dmapdag}=[draw,doubled,shape=SEbox,inner sep=2pt,minimum height=6mm,fill=white]
\tikzstyle{dmapadj}=[draw,doubled,shape=SEbox,inner sep=2pt,minimum height=6mm,fill=white]
\tikzstyle{dmaptrans}=[draw,doubled,shape=SWbox,inner sep=2pt,minimum height=6mm,fill=white]
\tikzstyle{dmapconj}=[draw,doubled,shape=NWbox,inner sep=2pt,minimum height=6mm,fill=white]

\tikzstyle{ddmap}=[draw,doubled,dashed,shape=NEbox,inner sep=2pt,minimum height=6mm,fill=white]
\tikzstyle{ddmapdag}=[draw,doubled,dashed,shape=SEbox,inner sep=2pt,minimum height=6mm,fill=white]
\tikzstyle{ddmapadj}=[draw,doubled,dashed,shape=SEbox,inner sep=2pt,minimum height=6mm,fill=white]
\tikzstyle{ddmaptrans}=[draw,doubled,dashed,shape=SWbox,inner sep=2pt,minimum height=6mm,fill=white]
\tikzstyle{ddmapconj}=[draw,doubled,dashed,shape=NWbox,inner sep=2pt,minimum height=6mm,fill=white]

\boxshape{sNEbox}{0pt}{3pt}
\boxshape{sSEbox}{0pt}{-3pt}
\boxshape{sNWbox}{3pt}{0pt}
\boxshape{sSWbox}{-3pt}{0pt}
\tikzstyle{smap}=[draw,shape=sNEbox,fill=white]
\tikzstyle{smapdag}=[draw,shape=sSEbox,fill=white]
\tikzstyle{smapadj}=[draw,shape=sSEbox,fill=white]
\tikzstyle{smaptrans}=[draw,shape=sSWbox,fill=white]
\tikzstyle{smapconj}=[draw,shape=sNWbox,fill=white]

\tikzstyle{dsmap}=[draw,dashed,shape=sNEbox,fill=white]
\tikzstyle{dsmapdag}=[draw,dashed,shape=sSEbox,fill=white]
\tikzstyle{dsmaptrans}=[draw,dashed,shape=sSWbox,fill=white]
\tikzstyle{dsmapconj}=[draw,dashed,shape=sNWbox,fill=white]

\boxshape{mNEbox}{0pt}{10pt}
\boxshape{mSEbox}{0pt}{-10pt}
\boxshape{mNWbox}{10pt}{0pt}
\boxshape{mSWbox}{-10pt}{0pt}
\tikzstyle{mmap}=[draw,shape=mNEbox]
\tikzstyle{mmapdag}=[draw,shape=mSEbox]
\tikzstyle{mmaptrans}=[draw,shape=mSWbox]
\tikzstyle{mmapconj}=[draw,shape=mNWbox]

\tikzstyle{mmapgray}=[draw,fill=gray!40!white,shape=mNEbox]
\tikzstyle{smapgray}=[draw,fill=gray!40!white,shape=sNEbox]

\makeatletter

\pgfdeclareshape{cornerpoint}{
\inheritsavedanchors[from=rectangle] 
\inheritanchorborder[from=rectangle]
\inheritanchor[from=rectangle]{center}
\inheritanchor[from=rectangle]{north}
\inheritanchor[from=rectangle]{south}
\inheritanchor[from=rectangle]{west}
\inheritanchor[from=rectangle]{east}
\backgroundpath{
\southwest \pgf@xa=\pgf@x \pgf@ya=\pgf@y
\northeast \pgf@xb=\pgf@x \pgf@yb=\pgf@y

\pgfmathsetmacro{\pgf@shorten@left}{\pgfkeysvalueof{/tikz/shorten left}}
\pgfmathsetmacro{\pgf@shorten@right}{\pgfkeysvalueof{/tikz/shorten right}}

\pgfpathmoveto{\pgfpoint{0.5 * (\pgf@xa + \pgf@xb)}{\pgf@ya - 5pt}}
\pgfpathlineto{\pgfpoint{\pgf@xa - 8pt + \pgf@shorten@left}{\pgf@yb - 1.5 * \pgf@shorten@left}}
\pgfpathlineto{\pgfpoint{\pgf@xa - 8pt + \pgf@shorten@left}{\pgf@yb}}
\pgfpathlineto{\pgfpoint{\pgf@xb + 8pt - \pgf@shorten@right}{\pgf@yb}}
\pgfpathlineto{\pgfpoint{\pgf@xb + 8pt - \pgf@shorten@right}{\pgf@yb - 1.5 * \pgf@shorten@right}}
\pgfpathclose
}
}

\pgfdeclareshape{cornercopoint}{
\inheritsavedanchors[from=rectangle] 
\inheritanchorborder[from=rectangle]
\inheritanchor[from=rectangle]{center}
\inheritanchor[from=rectangle]{north}
\inheritanchor[from=rectangle]{south}
\inheritanchor[from=rectangle]{west}
\inheritanchor[from=rectangle]{east}
\backgroundpath{
\southwest \pgf@xa=\pgf@x \pgf@ya=\pgf@y
\northeast \pgf@xb=\pgf@x \pgf@yb=\pgf@y

\pgfmathsetmacro{\pgf@shorten@left}{\pgfkeysvalueof{/tikz/shorten left}}
\pgfmathsetmacro{\pgf@shorten@right}{\pgfkeysvalueof{/tikz/shorten right}}

\pgfpathmoveto{\pgfpoint{0.5 * (\pgf@xa + \pgf@xb)}{\pgf@yb + 5pt}}
\pgfpathlineto{\pgfpoint{\pgf@xa - 8pt + \pgf@shorten@left}{\pgf@ya + 1.5 * \pgf@shorten@left}}
\pgfpathlineto{\pgfpoint{\pgf@xa - 8pt + \pgf@shorten@left}{\pgf@ya}}
\pgfpathlineto{\pgfpoint{\pgf@xb + 8pt - \pgf@shorten@right}{\pgf@ya}}
\pgfpathlineto{\pgfpoint{\pgf@xb + 8pt - \pgf@shorten@right}{\pgf@ya + 1.5 * \pgf@shorten@right}}
\pgfpathclose
}
}

\makeatother

\pgfkeyssetvalue{/tikz/shorten left}{0pt}
\pgfkeyssetvalue{/tikz/shorten right}{0pt}

\tikzstyle{kpoint common}=[draw,fill=white,inner sep=1pt,minimum height=4mm]
\tikzstyle{kpoint sc}=[shape=cornerpoint,kpoint common]
\tikzstyle{kpoint adjoint sc}=[shape=cornercopoint,kpoint common]
\tikzstyle{kpoint}=[shape=cornerpoint,shorten left=5pt,kpoint common]
\tikzstyle{kpoint adjoint}=[shape=cornercopoint,shorten left=5pt,kpoint common]
\tikzstyle{kpoint conjugate}=[shape=cornerpoint,shorten right=5pt,kpoint common]
\tikzstyle{kpoint transpose}=[shape=cornercopoint,shorten right=5pt,kpoint common]
\tikzstyle{kpoint symm}=[shape=cornerpoint,shorten left=5pt,shorten right=5pt,kpoint common]

\tikzstyle{wide kpoint sc}=[shape=cornerpoint,kpoint common, minimum width=1 cm]
\tikzstyle{wide kpointdag sc}=[shape=cornercopoint,kpoint common, minimum width=1 cm]

\tikzstyle{black kpoint}=[shape=cornerpoint,shorten left=5pt,kpoint common,fill=black,font=\color{white}]

\tikzstyle{black kpoint sm}=[shape=cornerpoint,shorten left=5pt,kpoint common,fill=black,font=\color{white},scale=0.75]

\tikzstyle{black kpoint adjoint}=[shape=cornercopoint,shorten left=5pt,kpoint common,fill=black,font=\color{white}]
\tikzstyle{black kpointadj}=[shape=cornercopoint,shorten left=5pt,kpoint common,fill=black,font=\color{white}]

\tikzstyle{black kpointadj sm}=[shape=cornercopoint,shorten left=5pt,kpoint common,fill=black,font=\color{white},scale=0.75]

\tikzstyle{black dkpoint}=[shape=cornerpoint,shorten left=5pt,kpoint common,fill=black, doubled,font=\color{white}]
\tikzstyle{black dkpoint adjoint}=[shape=cornercopoint,shorten left=5pt,kpoint common,fill=black, doubled,font=\color{white}]
\tikzstyle{black dkpointadj}=[shape=cornercopoint,shorten left=5pt,kpoint common,fill=black, doubled,font=\color{white}]

\tikzstyle{black dkpoint sm}=[shape=cornerpoint,shorten left=5pt,kpoint common,fill=black, doubled,font=\color{white},scale=0.75]
\tikzstyle{black dkpointadj sm}=[shape=cornercopoint,shorten left=5pt,kpoint common,fill=black, doubled,font=\color{white},scale=0.75]

\tikzstyle{kpointdag}=[kpoint adjoint]
\tikzstyle{kpointadj}=[kpoint adjoint]
\tikzstyle{kpointconj}=[kpoint conjugate]
\tikzstyle{kpointtrans}=[kpoint transpose]

\tikzstyle{big kpoint}=[kpoint, minimum width=1.2 cm, minimum height=8mm, inner sep=4pt, text depth=3mm]

\tikzstyle{wide kpoint}=[kpoint, minimum width=1 cm, inner sep=2pt]
\tikzstyle{wide kpointdag}=[kpointdag, minimum width=1 cm, inner sep=2pt]
\tikzstyle{wide kpointconj}=[kpointconj, minimum width=1 cm, inner sep=2pt]
\tikzstyle{wide kpointtrans}=[kpointtrans, minimum width=1 cm, inner sep=2pt]

\tikzstyle{wider kpoint}=[kpoint, minimum width=1.25 cm, inner sep=2pt]
\tikzstyle{wider kpointdag}=[kpointdag, minimum width=1.25 cm, inner sep=2pt]
\tikzstyle{wider kpointconj}=[kpointconj, minimum width=1.25 cm, inner sep=2pt]
\tikzstyle{wider kpointtrans}=[kpointtrans, minimum width=1.25 cm, inner sep=2pt]

\tikzstyle{gray kpoint}=[kpoint,fill=gray!50!white]
\tikzstyle{gray kpointdag}=[kpointdag,fill=gray!50!white]
\tikzstyle{gray kpointadj}=[kpointadj,fill=gray!50!white]
\tikzstyle{gray kpointconj}=[kpointconj,fill=gray!50!white]
\tikzstyle{gray kpointtrans}=[kpointtrans,fill=gray!50!white]

\tikzstyle{gray dkpoint}=[kpoint,fill=gray!50!white,doubled]
\tikzstyle{gray dkpointdag}=[kpointdag,fill=gray!50!white,doubled]
\tikzstyle{gray dkpointadj}=[kpointadj,fill=gray!50!white,doubled]
\tikzstyle{gray dkpointconj}=[kpointconj,fill=gray!50!white,doubled]
\tikzstyle{gray dkpointtrans}=[kpointtrans,fill=gray!50!white,doubled]

\tikzstyle{white label}=[draw,fill=white,rectangle,inner sep=0.7 mm]
\tikzstyle{gray label}=[draw,fill=gray!50!white,rectangle,inner sep=0.7 mm]
\tikzstyle{black label}=[draw,fill=black,rectangle,inner sep=0.7 mm]

\tikzstyle{dkpoint}=[kpoint,doubled]
\tikzstyle{wide dkpoint}=[wide kpoint,doubled]
\tikzstyle{dkpointdag}=[kpoint adjoint,doubled]
\tikzstyle{wide dkpointdag}=[wide kpointdag,doubled]
\tikzstyle{dkcopoint}=[kpoint adjoint,doubled]
\tikzstyle{dkpointadj}=[kpoint adjoint,doubled]
\tikzstyle{dkpointconj}=[kpoint conjugate,doubled]
\tikzstyle{dkpointtrans}=[kpoint transpose,doubled]

\tikzstyle{kscalar}=[kpoint common, shape=EBox, inner xsep=-1pt, inner ysep=3pt,font=\small]
\tikzstyle{kscalarconj}=[kpoint common, shape=WBox, inner xsep=-1pt, inner ysep=3pt,font=\small]

\tikzstyle{spekpoint}=[kpoint sc,minimum height=5mm,inner sep=3pt]
\tikzstyle{spekcopoint}=[kpoint adjoint sc,minimum height=5mm,inner sep=3pt]

\tikzstyle{dspekpoint}=[spekpoint,doubled]
\tikzstyle{dspekcopoint}=[spekcopoint,doubled]

 \tikzstyle{upground}=[circuit ee IEC,thick,ground,rotate=90,scale=2.5]
 \tikzstyle{downground}=[circuit ee IEC,thick,ground,rotate=-90,scale=2.5]
 \tikzstyle{bigground}=[regular polygon,regular polygon sides=3,draw=gray,scale=0.50,inner sep=-0.5pt,minimum width=10mm,fill=gray]
 \tikzstyle{maxmix}=[regular polygon,regular polygon sides=3,draw=black,xscale=0.4,yscale=0.3,inner sep=-0.5pt,minimum width=10mm,fill=gray,regular polygon rotate=180]

\tikzstyle{arrs}=[-latex,font=\small,auto]
\tikzstyle{arrow plain}=[arrs]
\tikzstyle{arrow dashed}=[dashed,arrs]
\tikzstyle{arrow bold}=[very thick,arrs]
\tikzstyle{arrow hide}=[draw=white!0,-]
\tikzstyle{arrow reverse}=[latex-]
\tikzstyle{cdnode}=[]

%% file: figures/SuperpositionOfPaths.tikz
\begin{tikzpicture}
	\begin{pgfonlayer}{nodelayer}
		\node [style=map] (0) at (-1.25, 0) {$\mathcal{A}$};
		\node [style=none] (1) at (-1.25, -2.5) {};
		\node [style=none] (2) at (-1.25, 2.5) {};
		\node [style=right label] (3) at (-1.25, 1.25) {$A$};
		\node [style=right label] (4) at (-1.25, -1.25) {$A$};
		\node [style=map] (5) at (1.25, 0) {$\mathcal{B}$};
		\node [style=none] (6) at (1.25, -2.5) {};
		\node [style=none] (7) at (1.25, 2.5) {};
		\node [style=right label] (8) at (1.25, 1.25) {$B$};
		\node [style=right label] (9) at (1.25, -1.25) {$B$};
		\node [style=large map] (10) at (0, 2.5) {$\mathcal{D}$};
		\node [style=large map] (11) at (0, -2.5) {$\mathcal{E}$};
		\node [style=none] (13) at (-1, -4) {};
		\node [style=right label] (14) at (-1, -3.75) {$M$};
		\node [style=none] (16) at (1, -4) {};
		\node [style=right label] (17) at (1, -3.75) {$C$};
		\node [style=none] (18) at (-1, -2.5) {};
		\node [style=none] (19) at (1, -2.5) {};
		\node [style=none] (20) at (-1, 4) {};
		\node [style=right label] (21) at (-1, 3.75) {$M$};
		\node [style=none] (22) at (1, 4) {};
		\node [style=right label] (23) at (1, 3.75) {$C$};
		\node [style=none] (24) at (-1, 2.5) {};
		\node [style=none] (25) at (1, 2.5) {};
	\end{pgfonlayer}
	\begin{pgfonlayer}{edgelayer}
		\draw (2.center) to (0);
		\draw (1.center) to (0);
		\draw (7.center) to (5);
		\draw (6.center) to (5);
		\draw (18.center) to (13.center);
		\draw (19.center) to (16.center);
		\draw (24.center) to (20.center);
		\draw (25.center) to (22.center);
	\end{pgfonlayer}
\end{tikzpicture}

%% file: figures/CausalDec1.tikz
\begin{tikzpicture}
	\begin{pgfonlayer}{nodelayer}
		\node [style=right label] (1) at (-1, 1) {$C$};
		\node [style=right label] (3) at (1, 2) {$D$};
		\node [style=large map] (4) at (0, 0) {$\mathcal{U}$};
		\node [style=none] (5) at (-1, -2) {};
		\node [style=right label] (6) at (-1, -1.75) {$A$};
		\node [style=none] (7) at (1, -2) {};
		\node [style=right label] (8) at (1, -1.75) {$B$};
		\node [style=none] (9) at (-1, 0) {};
		\node [style=none] (10) at (1, 0) {};
		\node [style=none] (11) at (-1, 1.75) {};
		\node [style=none] (12) at (1, 2.25) {};
		\node [style=upground] (13) at (-1, 2) {};
	\end{pgfonlayer}
	\begin{pgfonlayer}{edgelayer}
		\draw (9.center) to (5.center);
		\draw (10.center) to (7.center);
		\draw (9.center) to (11.center);
		\draw (10.center) to (12.center);
	\end{pgfonlayer}
\end{tikzpicture}

%% file: figures/CausalDec2.tikz
\begin{tikzpicture}
	\begin{pgfonlayer}{nodelayer}
		\node [style=right label] (1) at (1, 2) {$D$};
		\node [style=map] (2) at (1, 0) {$\mathcal{C}$};
		\node [style=none] (3) at (-1, -2) {};
		\node [style=right label] (4) at (-1, -1.75) {$A$};
		\node [style=none] (5) at (1, -2) {};
		\node [style=right label] (6) at (1, -1.75) {$B$};
		\node [style=none] (7) at (-1, -1) {};
		\node [style=none] (10) at (1, 2.25) {};
		\node [style=upground] (11) at (-1, -0.75) {};
	\end{pgfonlayer}
	\begin{pgfonlayer}{edgelayer}
		\draw (7.center) to (3.center);
		\draw (10.center) to (2);
		\draw (5.center) to (2);
	\end{pgfonlayer}
\end{tikzpicture}

%% file: figures/CompDec1.tikz
\begin{tikzpicture}
	\begin{pgfonlayer}{nodelayer}
		\node [style=right label] (0) at (-1, 1.75) {$C$};
		\node [style=right label] (1) at (1, 1.75) {$D$};
		\node [style=large map] (2) at (0, 0) {$\mathcal{U}$};
		\node [style=none] (3) at (-1, -2) {};
		\node [style=right label] (4) at (-1, -1.75) {$A$};
		\node [style=none] (5) at (1, -2) {};
		\node [style=right label] (6) at (1, -1.75) {$B$};
		\node [style=none] (7) at (-1, 0) {};
		\node [style=none] (8) at (1, 0) {};
		\node [style=none] (9) at (-1, 2) {};
		\node [style=none] (10) at (1, 2) {};
	\end{pgfonlayer}
	\begin{pgfonlayer}{edgelayer}
		\draw (7.center) to (3.center);
		\draw (8.center) to (5.center);
		\draw (7.center) to (9.center);
		\draw (8.center) to (10.center);
	\end{pgfonlayer}
\end{tikzpicture}

%% file: figures/CompDec2.tikz
\begin{tikzpicture}
	\begin{pgfonlayer}{nodelayer}
		\node [style=right label] (0) at (-0.75, 1.75) {$C$};
		\node [style=right label] (1) at (1.25, 1.75) {$D$};
		\node [style=none] (3) at (-1.25, -2) {};
		\node [style=right label] (4) at (-1.25, -1.75) {$A$};
		\node [style=none] (5) at (0.75, -2) {};
		\node [style=right label] (6) at (0.75, -1.75) {$B$};
		\node [style=none] (9) at (-0.75, 2) {};
		\node [style=none] (10) at (1.25, 2) {};
		\node [style=map] (11) at (0.75, -0.75) {$\cu_1$};
		\node [style=none] (12) at (1.25, -0.5) {};
		\node [style=none] (13) at (0.25, -0.5) {};
		\node [style=map] (14) at (-0.75, 0.75) {$\cu_2$};
		\node [style=none] (15) at (0.75, -1) {};
		\node [style=none] (16) at (-0.75, 1) {};
		\node [style=none] (17) at (-1.25, 0.5) {};
		\node [style=none] (18) at (-0.25, 0.5) {};
	\end{pgfonlayer}
	\begin{pgfonlayer}{edgelayer}
		\draw (9.center) to (16.center);
		\draw (17.center) to (3.center);
		\draw (15.center) to (5.center);
		\draw (12.center) to (10.center);
		\draw [in=90, out=-90] (18.center) to (13.center);
	\end{pgfonlayer}
\end{tikzpicture}

%% file: figures/Diamond1.tikz
\begin{tikzpicture}
	\begin{pgfonlayer}{nodelayer}
		\node [style=none] (0) at (-1.25, 0) {};
		\node [style=right label] (1) at (-1.25, 1.25) {$A_O$};
		\node [style=none] (2) at (1.25, 0) {};
		\node [style=right label] (3) at (1.25, 1.25) {$B_O$};
		\node [style=large map] (4) at (0, 0) {$U$};
		\node [style=none] (5) at (-1.25, -1.5) {};
		\node [style=right label] (6) at (-1.25, -1.25) {$A_I$};
		\node [style=none] (7) at (1.25, -1.5) {};
		\node [style=right label] (8) at (1.25, -1.25) {$B_I$};
		\node [style=none] (9) at (-1.25, 0) {};
		\node [style=none] (10) at (1.25, 0) {};
		\node [style=none] (11) at (-1.25, 1.5) {};
		\node [style=none] (12) at (1.25, 1.5) {};
		\node [style=none] (13) at (0, 0) {};
		\node [style=right label] (14) at (0, 1.25) {$E_O$};
		\node [style=none] (15) at (0, -1.5) {};
		\node [style=right label] (16) at (0, -1.25) {$E_I$};
		\node [style=none] (17) at (0, 0) {};
		\node [style=none] (18) at (0, 1.5) {};
	\end{pgfonlayer}
	\begin{pgfonlayer}{edgelayer}
		\draw (9.center) to (5.center);
		\draw (10.center) to (7.center);
		\draw (11.center) to (0.center);
		\draw (12.center) to (2.center);
		\draw (17.center) to (15.center);
		\draw (18.center) to (13.center);
	\end{pgfonlayer}
\end{tikzpicture}

%% file: figures/Diamond2.tikz
\begin{tikzpicture}
	\begin{pgfonlayer}{nodelayer}
		\node [style=right label] (0) at (-2.25, 3.25) {$A_O$};
		\node [style=right label] (1) at (2.25, 3.25) {$B_O$};
		\node [style=none] (2) at (-2.25, -3.5) {};
		\node [style=right label] (3) at (-2.25, -3.25) {$A_I$};
		\node [style=none] (4) at (2.25, -3.5) {};
		\node [style=right label] (5) at (2.25, -3.25) {$B_I$};
		\node [style=none] (6) at (-2.25, 3.5) {};
		\node [style=none] (7) at (2.25, 3.5) {};
		\node [style=right label] (8) at (0, 3.25) {$E_O$};
		\node [style=none] (9) at (0, -3.5) {};
		\node [style=right label] (10) at (0, -3.25) {$E_I$};
		\node [style=none] (11) at (0, 3.5) {};
		\node [style=map] (13) at (0, 2.25) {$U_4$};
		\node [style=map] (16) at (0, -2.25) {$U_1$};
		\node [style=none] (20) at (-0.25, 2) {};
		\node [style=none] (21) at (0.25, 2) {};
		\node [style=none] (26) at (0.25, -2) {};
		\node [style=none] (27) at (-0.25, -2) {};
		\node [style=map] (28) at (2, 0) {$U_3$};
		\node [style=map] (29) at (-2, 0) {$U_2$};
		\node [style=none] (30) at (2.25, 0) {};
		\node [style=none] (31) at (-2.25, 0) {};
		\node [style=none] (32) at (1.5, 0.25) {};
		\node [style=none] (33) at (1.5, -0.25) {};
		\node [style=none] (34) at (-1.5, 0.25) {};
		\node [style=none] (35) at (-1.5, -0.25) {};
		\node [style=right label] (36) at (1, 1.25) {$R'$};
		\node [style=right label] (37) at (1, -1.25) {$R$};
		\node [style=left label] (38) at (-1.25, -1.25) {$L$};
		\node [style=left label] (39) at (-1, 1.25) {$L'$};
	\end{pgfonlayer}
	\begin{pgfonlayer}{edgelayer}
		\draw (16) to (9.center);
		\draw (13) to (11.center);
		\draw (6.center) to (31.center);
		\draw (31.center) to (2.center);
		\draw (4.center) to (30.center);
		\draw (30.center) to (7.center);
		\draw [in=90, out=-90] (20.center) to (34.center);
		\draw [in=90, out=-90] (35.center) to (27.center);
		\draw [in=-90, out=90] (26.center) to (33.center);
		\draw [in=-90, out=90] (32.center) to (21.center);
	\end{pgfonlayer}
\end{tikzpicture}

%% file: figures/Diamond3.tikz
\begin{tikzpicture}
	\begin{pgfonlayer}{nodelayer}
		\node [style=right label] (0) at (-2.25, 3.25) {$A_O$};
		\node [style=right label] (1) at (2.25, 3.25) {$B_O$};
		\node [style=none] (2) at (-2.25, -3.5) {};
		\node [style=right label] (3) at (-2.25, -3.25) {$A_I$};
		\node [style=none] (4) at (2.25, -3.5) {};
		\node [style=right label] (5) at (2.25, -3.25) {$B_I$};
		\node [style=none] (6) at (-2.25, 3.5) {};
		\node [style=none] (7) at (2.25, 3.5) {};
		\node [style=right label] (8) at (0, 3.25) {$E_O$};
		\node [style=none] (9) at (0, -3.5) {};
		\node [style=right label] (10) at (0, -3.25) {$E_I$};
		\node [style=none] (11) at (0, 3.5) {};
		\node [style=map] (12) at (0, 2.25) {$U_4$};
		\node [style=map] (13) at (0, -2.25) {$U_1$};
		\node [style=none] (14) at (-0.25, 2) {};
		\node [style=none] (15) at (0.25, 2) {};
		\node [style=none] (16) at (0.25, -2) {};
		\node [style=none] (17) at (-0.25, -2) {};
		\node [style=map] (18) at (2, 0) {$U_{3}$};
		\node [style=map] (19) at (-2, 0) {$U_2$};
		\node [style=none] (20) at (2.25, 0) {};
		\node [style=none] (21) at (-2.25, 0) {};
		\node [style=none] (22) at (1.5, 0.25) {};
		\node [style=none] (23) at (1.5, -0.25) {};
		\node [style=none] (24) at (-1.5, 0.25) {};
		\node [style=none] (25) at (-1.5, -0.25) {};
		\node [style=right label] (26) at (1, 1.25) {${R'}^k$};
		\node [style=right label] (27) at (1, -1.25) {$R^k$};
		\node [style=left label] (28) at (-1.25, -1.25) {$L^k$};
		\node [style=left label] (29) at (-1, 1.25) {${L'}^k$};
	\end{pgfonlayer}
	\begin{pgfonlayer}{edgelayer}
		\draw (13) to (9.center);
		\draw (12) to (11.center);
		\draw (6.center) to (21.center);
		\draw (21.center) to (2.center);
		\draw (4.center) to (20.center);
		\draw (20.center) to (7.center);
		\draw [in=90, out=-90] (14.center) to (24.center);
		\draw [in=90, out=-90] (25.center) to (17.center);
		\draw [in=-90, out=90] (16.center) to (23.center);
		\draw [in=-90, out=90] (22.center) to (15.center);
	\end{pgfonlayer}
\end{tikzpicture}

%% file: figures/RelationGraph1.tikz
\begin{tikzpicture}
	\begin{pgfonlayer}{nodelayer}
		\node [thick] (0) at (-8, 1.5) {$1$};
		\node [very thick] (1) at (-8, 0.5) {$2$};
		\node [very thick] (2) at (-8, -0.5) {$3$};
		\node [very thick] (3) at (-8, -1.5) {$4$};
		\node [very thick] (4) at (0, 0.5) {b};
		\node [very thick] (5) at (0, -0.5) {c};
		\node [very thick] (6) at (0, -1.5) {d};
		\node [very thick] (7) at (0, 1.5) {a};
		\node [ellipse, minimum height=70pt, minimum width=30pt, dashed, draw] (12) at (-8, 0) {};
		\node [ellipse, minimum height=70pt, minimum width=30pt, dashed, draw] (13) at (0, 0) {};
		\node [thick] (15) at (-8,-3.3) {$\cz_A$};
		\node [thick] (16) at (0,-3.3) {$\cz_B$};
		\node[thick] (17) at (-4,-2) {$\lambda$};
	\end{pgfonlayer}
	\begin{pgfonlayer}{edgelayer}
		\draw [style=arrow plain] (0) to (7);
		\draw [style=arrow plain, bend left=10] (1) to (4);
		\draw [style=arrow plain, bend left=5] (1) to (5);
		\draw [style=arrow plain, bend right=10] (2) to (5);
	\end{pgfonlayer}
\end{tikzpicture}

%% file: figures/RelationGraph2.tikz
\begin{tikzpicture}
	\begin{pgfonlayer}{nodelayer}
		\node [thick] (0) at (-5, 1.5) {$1$};
		\node [very thick] (1) at (-5, 0.5) {$2$};
		\node [very thick] (2) at (-5, -0.5) {$3$};
		\node [very thick] (3) at (-5, -1.5) {$4$};
		\node [very thick] (4) at (0, 0.5) {b};
		\node [very thick] (5) at (0, -0.5) {c};
		\node [very thick] (6) at (0, -1.5) {d};
		\node [very thick] (7) at (0, 1.5) {a};
		\node [very thick] (8) at (5, 0) {y};
		\node [very thick] (9) at (5, -1) {z};
		\node [very thick] (11) at (5, 1) {x};
		\node [ellipse, minimum height=70pt, minimum width=30pt, dashed, draw] (12) at (-5, 0) {};
		\node [ellipse, minimum height=70pt, minimum width=30pt, dashed, draw] (13) at (0, 0) {};
		\node [ellipse, minimum height=50pt, minimum width=30pt, dashed, draw] (14) at (5, 0) {};
		\node [thick] (15) at (-5,-3.3) {$\cz_A$};
		\node [thick] (16) at (0,-3.3) {$\cz_B$};
		\node [thick] (17) at (5,-3.3) {$\cz_C$};
		\node[thick] (18) at (-2.5,-2) {$\lambda$};
		\node[thick] (19) at (2.5,-2) {$\sigma$};
	\end{pgfonlayer}
	\begin{pgfonlayer}{edgelayer}
		\draw [style=arrow plain] (0) to (7);
		\draw [style=arrow plain, bend left=15] (1) to (4);
		\draw [style=arrow plain] (1) to (5);
		\draw [style=arrow plain] (7) to (11);
		\draw [style=arrow plain] (4) to (8);
		\draw [style=arrow plain] (4) to (9);
		\draw [style=arrow plain] (6) to (9);
		\draw [style=arrow plain, bend right=15] (2) to (5);
	\end{pgfonlayer}
\end{tikzpicture}

%% file: figures/RelationGraph3.tikz
\begin{tikzpicture}
	\begin{pgfonlayer}{nodelayer}
		\node [thick] (0) at (-7, 1.5) {$1$};
		\node [very thick] (1) at (-7, 0.5) {$2$};
		\node [very thick] (2) at (-7, -0.5) {$3$};
		\node [very thick] (3) at (-7, -1.5) {$4$};
		\node [very thick] (8) at (0, 0) {y};
		\node [very thick] (9) at (0, -1) {z};
		\node [very thick] (11) at (0, 1) {x};
		\node [ellipse, minimum height=70pt, minimum width=30pt, dashed, draw] (12) at (-7, 0) {};
		\node [ellipse, minimum height=50pt, minimum width=30pt, dashed, draw] (14) at (0, 0) {};
		\node [thick] (15) at (-7,-3.3) {$\cz_A$};
		\node [thick] (17) at (0,-3.3) {$\cz_C$};
		\node[thick] (18) at (-3.5,-2) {$\sigma \circ \lambda$};
	\end{pgfonlayer}
	\begin{pgfonlayer}{edgelayer}
		\draw [style=arrow plain] (0) to (11);
		\draw [style=arrow plain] (1) to (8);
		\draw [style=arrow plain] (1) to (9);
	\end{pgfonlayer}
\end{tikzpicture}

%% file: figures/ChannelSuperposition.tikz
\begin{tikzpicture}
	\begin{pgfonlayer}{nodelayer}
		\node [style=map] (0) at (0.5, 1) {$V_1$};
		\node [style=none] (1) at (0.5, -1.5) {};
		\node [style=none] (2) at (0.5, 3.5) {};
		\node [style=right label] (3) at (0.5, 2.25) {$A^l$};
		\node [style=right label] (4) at (0.5, -0.25) {$A^k$};
		\node [style=map] (5) at (3, 1) {$V_2$};
		\node [style=none] (6) at (3, -1.5) {};
		\node [style=none] (7) at (3, 3.5) {};
		\node [style=right label] (8) at (3, 2.25) {$B^n$};
		\node [style=right label] (9) at (3, -0.25) {$B^m$};
		\node [style=large map] (10) at (1.75, 3.5) {$U^\dagger$};
		\node [style=large map] (11) at (1.75, -1.5) {$U$};
		\node [style=none] (13) at (0.75, -3) {};
		\node [style=right label] (14) at (0.75, -2.75) {$M$};
		\node [style=none] (16) at (2.75, -3) {};
		\node [style=right label] (17) at (2.75, -2.75) {$C$};
		\node [style=none] (18) at (0.75, -1.5) {};
		\node [style=none] (19) at (2.75, -1.5) {};
		\node [style=none] (20) at (-1, -1.5) {$\omega^{km}$};
		\node [style=none] (21) at (-1, 1) {$\delta^l_k$};
		\node [style=none] (22) at (4.5, 1) {$\delta^n_m$};
		\node [style=none] (23) at (-1, 3.5) {$\omega_{ln}$};
		\node [style=none] (24) at (0.75, 5) {};
		\node [style=right label] (25) at (0.75, 4.75) {$M$};
		\node [style=none] (26) at (2.75, 5) {};
		\node [style=right label] (27) at (2.75, 4.75) {$C$};
		\node [style=none] (28) at (0.75, 3.5) {};
		\node [style=none] (29) at (2.75, 3.5) {};
	\end{pgfonlayer}
	\begin{pgfonlayer}{edgelayer}
		\draw (2.center) to (0);
		\draw (1.center) to (0);
		\draw (7.center) to (5);
		\draw (6.center) to (5);
		\draw (18.center) to (13.center);
		\draw (19.center) to (16.center);
		\draw (28.center) to (24.center);
		\draw (29.center) to (26.center);
	\end{pgfonlayer}
\end{tikzpicture}

%% file: figures/omegaGraph.tikz
\begin{tikzpicture}
	\begin{pgfonlayer}{nodelayer}
		\node [very thick] (1) at (-8, 0) {$*$};
		\node [very thick] (4) at (0, 1) {$(0,1)$};
		\node [very thick] (5) at (0, -1) {$(1,0)$};
		\node [very thick] (6) at (0, -3) {$(1,1)$};
		\node [very thick] (7) at (0, 3) {$(0,0)$};
		\node [thick] (15) at (-8, -2) {$\{*\}$};
		\node [thick] (16) at (0, -5) {$\cz_A \times \cz_B$};
		\node [thick] (17) at (-4, -1.5) {$\omega$};
		\node [ellipse, minimum height=30pt, minimum width=30pt, dashed, draw] (18) at (-8, 0) {};
		\node [ellipse, minimum height=120pt, minimum width=50pt, dashed, draw] (19) at (0, 0) {};
	\end{pgfonlayer}
	\begin{pgfonlayer}{edgelayer}
		\draw [style=arrow plain] (1) to (4);
		\draw [style=arrow plain] (1) to (5);
	\end{pgfonlayer}
\end{tikzpicture}

%% file: figures/ChannelSuperposition2.tikz
\begin{tikzpicture}
	\begin{pgfonlayer}{nodelayer}
		\node [style=map] (0) at (0.5, 1) {$V_1$};
		\node [style=none] (1) at (0.5, -1.5) {};
		\node [style=none] (2) at (0.5, 3.5) {};
		\node [style=right label] (3) at (0.5, 2.25) {$A^k$};
		\node [style=right label] (4) at (0.5, -0.25) {$A^k$};
		\node [style=map] (5) at (3, 1) {$V_2$};
		\node [style=none] (6) at (3, -1.5) {};
		\node [style=none] (7) at (3, 3.5) {};
		\node [style=right label] (8) at (3, 2.25) {$B^m$};
		\node [style=right label] (9) at (3, -0.25) {$B^m$};
		\node [style=large map] (10) at (1.75, 3.5) {$U^\dagger$};
		\node [style=large map] (11) at (1.75, -1.5) {$U$};
		\node [style=none] (13) at (0.75, -3) {};
		\node [style=right label] (14) at (0.75, -2.75) {$M$};
		\node [style=none] (16) at (2.75, -3) {};
		\node [style=right label] (17) at (2.75, -2.75) {$C$};
		\node [style=none] (18) at (0.75, -1.5) {};
		\node [style=none] (19) at (2.75, -1.5) {};
		\node [style=none] (20) at (-1, -1.5) {$\omega^{km}$};

		\node [style=none] (23) at (-1, 3.5) {$\omega_{km}$};
		\node [style=none] (24) at (0.75, 5) {};
		\node [style=right label] (25) at (0.75, 4.75) {$M$};
		\node [style=none] (26) at (2.75, 5) {};
		\node [style=right label] (27) at (2.75, 4.75) {$C$};
		\node [style=none] (28) at (0.75, 3.5) {};
		\node [style=none] (29) at (2.75, 3.5) {};
	\end{pgfonlayer}
	\begin{pgfonlayer}{edgelayer}
		\draw (2.center) to (0);
		\draw (1.center) to (0);
		\draw (7.center) to (5);
		\draw (6.center) to (5);
		\draw (18.center) to (13.center);
		\draw (19.center) to (16.center);
		\draw (28.center) to (24.center);
		\draw (29.center) to (26.center);
	\end{pgfonlayer}
\end{tikzpicture}

%% file: slices1a.tikz
\begin{tikzpicture}
	\begin{pgfonlayer}{nodelayer}
		\node [style=map] (0) at (1.75, 1) {$~~f~~$};
		\node [style=none] (1) at (1, -1.5) {};
		\node [style=none] (2) at (1, 3.5) {};
		\node [style=none] (6) at (3.5, -1.5) {};
		\node [style=none] (7) at (2.5, 3.5) {};
		\node [style=right label] (9) at (3.5, -0.25) {$B$};
		\node [style=none] (13) at (0.5, -3) {};
		\node [style=none] (16) at (3.5, -3) {};
		\node [style=right label] (17) at (3.5, -2.75) {$Z$};
		\node [style=none] (18) at (0.5, -1.5) {};
		\node [style=none] (19) at (3.5, -1.5) {};
		\node [style=none] (20) at (-0.5, -1.5) {$\alpha^{k}_p$};
		\node [style=none] (23) at (-0.5, 3.5) {$\delta_{m}^i$};
		\node [style=none] (24) at (1, 5) {};
		\node [style=right label] (25) at (1, 4.75) {$E^i$};
		\node [style=none] (26) at (3, 5) {};
		\node [style=right label] (27) at (3, 4.75) {$F^j$};
		\node [style=none] (28) at (1, 3.5) {};
		\node [style=none] (29) at (3, 3.5) {};
		\node [style=map] (30) at (1, 3.5) {$h$};
		\node [style=map] (31) at (3, 3.5) {$~g~$};
		\node [style=none] (32) at (4.5, 3.5) {$\gamma_{n}^j$};
		\node [style=map] (33) at (1, -1.5) {$~d~$};
		\node [style=none] (35) at (3.5, 3.5) {};
		\node [style=none] (37) at (1.25, -3) {};
		\node [style=right label] (38) at (1.25, -2.75) {$Y^p$};
		\node [style=none] (39) at (1.25, -1.5) {};
		\node [style=map] (40) at (3.5, -1.5) {$e$};
		\node [style=none] (41) at (-0.5, 1) {$\lambda_{k}^{mn}$};
		\node [style=none] (43) at (1, 1) {};
		\node [style=none] (44) at (2.5, 1) {};
		\node [style=none] (45) at (1, 1) {};
		\node [style=left label] (46) at (1, 2.25) {$C^m~$};
		\node [style=left label] (47) at (2.5, 2.25) {$D^n~$};
		\node [style=left label] (48) at (1, -0.25) {$A^k~$};
		\node [style=left label] (49) at (0.5, -2.75) {};
		\node [style=left label] (50) at (0.5, -2.75) {$X~$};
		\node [style=none] (51) at (8.25, 4.5) {$\bigoplus_{ij} \eta^{ij} \ch^i_E \otimes \ch^j_F$};
		\node [style=none] (54) at (9.5, 3.25) {$\eta^{ij}:=\sum_{mnkp} \delta_m^i \gamma_n^j \lambda_k^{mn} \alpha_p^k$};
		\node [style=none] (56) at (9.5, 1.25) {$\bigoplus_{in} \eta'^{in} \ch^i_E \otimes \ch^n_D \otimes \ch_B$};
		\node [style=none] (57) at (9.75, 0) {$\eta'^{in}:=\sum_{mkjp} \delta_m^i \lambda_k^{mn} \gamma_n^j \alpha_p^k$};
		\node [style=none] (58) at (-1.25, 4.5) {};
		\node [style=none] (59) at (5.25, 4.5) {};
		\node [style=none] (60) at (-1.25, 4.25) {};
		\node [style=none] (61) at (5.25, 2.25) {};
	\end{pgfonlayer}
	\begin{pgfonlayer}{edgelayer}
		\draw (18.center) to (13.center);
		\draw (19.center) to (16.center);
		\draw (28.center) to (24.center);
		\draw (29.center) to (26.center);
		\draw (39.center) to (37.center);
		\draw (35.center) to (40);
		\draw (43.center) to (30);
		\draw (45.center) to (33);
		\draw (7.center) to (44.center);
		\draw [color=blue,style=dashed](58.center) to (59.center);
		\draw [color=magenta,style=dashed,in=-180, out=0, looseness=2.25] (60.center) to (61.center);
	\end{pgfonlayer}
\end{tikzpicture}

%% file: figures/ChannelSuperposition3.tikz
\begin{tikzpicture}
	\begin{pgfonlayer}{nodelayer}
		\node [style=map] (0) at (-1.5, 1) {$\ca$};
		\node [style=none] (1) at (0.5, -2.5) {};
		\node [style=none] (2) at (0.5, 4.5) {};
		\node [style=right label] (3) at (-2.75, 2.75) {$A^{kk'}$};
		\node [style=right label] (4) at (-2.75, -0.75) {$A^{kk'}$};
		\node [style=map] (5) at (5, 1) {$\cc$};
		\node [style=none] (6) at (2.5, -2.5) {};
		\node [style=none] (7) at (2.25, 4.5) {};
		\node [style=right label] (8) at (4.5, 2.75) {$C^{pp'}$};
		\node [style=right label] (9) at (4.5, -0.75) {$C^{pp'}$};
		\node [style=large map] (10) at (1.25, 4.5) {$\cu^\dagger$};
		\node [style=large map] (11) at (1.25, -2.5) {$\cu$};
		\node [style=none] (13) at (0.25, -4) {};
		\node [style=right label] (14) at (0.25, -3.75) {$M$};
		\node [style=none] (16) at (2.25, -4) {};
		\node [style=right label] (17) at (2.25, -3.75) {$C$};
		\node [style=none] (18) at (0.25, -2.5) {};
		\node [style=none] (19) at (2.25, -2.5) {};
		\node [style=none] (20) at (-2.5, -2.5) {$\omega^{kmp} \omega^{k'm'p'}$};
		\node [style=none] (23) at (-2.5, 4.5) {$\omega_{kmp} \omega_{k'm'p'}$};
		\node [style=none] (24) at (0.25, 6) {};
		\node [style=right label] (25) at (0.25, 5.75) {$M$};
		\node [style=none] (26) at (2.25, 6) {};
		\node [style=right label] (27) at (2.25, 5.75) {$C$};
		\node [style=none] (28) at (0.25, 4.5) {};
		\node [style=none] (29) at (2.25, 4.5) {};
		\node [style=map] (30) at (1, 1) {$\cb$};
		\node [style=none] (31) at (1, -2.5) {};
		\node [style=none] (32) at (1, 4.5) {};
		\node [style=right label] (33) at (1, 2.75) {$B^{mm'}$};
		\node [style=right label] (34) at (1, -0.75) {$B^{mm'}$};
	\end{pgfonlayer}
	\begin{pgfonlayer}{edgelayer}
		\draw [in=90, out=-90, looseness=1.25] (2.center) to (0);
		\draw [in=-90, out=90] (1.center) to (0);
		\draw [in=90, out=-90] (7.center) to (5);
		\draw [in=-90, out=90] (6.center) to (5);
		\draw (18.center) to (13.center);
		\draw (19.center) to (16.center);
		\draw (28.center) to (24.center);
		\draw (29.center) to (26.center);
		\draw (32.center) to (30);
		\draw (31.center) to (30);
	\end{pgfonlayer}
\end{tikzpicture}

%% file: figures/Copying1.tikz
\begin{tikzpicture}
	\begin{pgfonlayer}{nodelayer}
		\node [style=right label] (9) at (1, 1.75) {$C^{ll'}$};
		\node [style=large map] (11) at (0, 0) {$\cc$};
		\node [style=none] (13) at (0, -2.5) {};
		\node [style=right label] (14) at (0, -2) {$A$};
		\node [style=none] (18) at (0, 0) {};
		\node [style=none] (19) at (1, 0) {};
		\node [style=none] (20) at (-3, 0) {$\delta^{kl} \delta^{k'l'}$};
		\node [style=upground] (30) at (-0.75, 3) {};
		\node [style=none] (31) at (-0.75, 0) {};
		\node [style=left label] (34) at (-1, 1.75) {$B^{kk'}$};
		\node [style=none] (35) at (-2, 3) {$\delta_{kk'}$};
		\node [style=none] (36) at (1, 3) {};
	\end{pgfonlayer}
	\begin{pgfonlayer}{edgelayer}
		\draw (18.center) to (13.center);
		\draw (31.center) to (30);
		\draw (36.center) to (19.center);
	\end{pgfonlayer}
\end{tikzpicture}

%% file: figures/Copying2.tikz
\begin{tikzpicture}
	\begin{pgfonlayer}{nodelayer}
		\node [style=right label] (9) at (1.75, 2.25) {$C^{ll'}$};
		\node [style=large map] (11) at (1.5, 0) {$\cc'$};
		\node [style=none] (13) at (1.5, -2.5) {};
		\node [style=right label] (14) at (1.5, -2) {$A$};
		\node [style=none] (18) at (1.5, 0) {};
		\node [style=none] (20) at (-1, 0) {$\delta^{ll'}$};
		\node [style=none] (36) at (1.5, 3) {};
	\end{pgfonlayer}
	\begin{pgfonlayer}{edgelayer}
		\draw (18.center) to (13.center);
		\draw (18.center) to (36.center);
	\end{pgfonlayer}
\end{tikzpicture}

%% file: figures/IMmap1.tikz
\begin{tikzpicture}
	\begin{pgfonlayer}{nodelayer}
		\node [style=large map] (0) at (0, 0) {$U_1$};
		\node [style=none] (1) at (-0.75, 0) {};
		\node [style=none] (2) at (0.75, 0) {};
		\node [style=none] (3) at (-0.75, 2) {};
		\node [style=none] (4) at (-0.75, -2) {};
		\node [style=none] (5) at (0.75, -2) {};
		\node [style=none] (6) at (0.75, 2) {};
		\node [style=right label] (7) at (1, 1.75) {$D^k$};
		\node [style=left label] (8) at (-1, 1.75) {$C^k$};
		\node [style=left label] (9) at (-1, -1.75) {$A^k$};
		\node [style=right label] (10) at (1, -1.75) {$B^k$};
	\end{pgfonlayer}
	\begin{pgfonlayer}{edgelayer}
		\draw (1.center) to (3.center);
		\draw (1.center) to (4.center);
		\draw [in=270, out=90] (2.center) to (6.center);
		\draw (2.center) to (5.center);
	\end{pgfonlayer}
\end{tikzpicture}

%% file: figures/IMmap2.tikz
\begin{tikzpicture}
	\begin{pgfonlayer}{nodelayer}
		\node [style=large map] (0) at (0, 0) {$U_2$};
		\node [style=none] (1) at (-0.75, 0) {};
		\node [style=none] (2) at (0.75, 0) {};
		\node [style=none] (3) at (-0.75, 2) {};
		\node [style=none] (4) at (-0.75, -2) {};
		\node [style=none] (5) at (0.75, -2) {};
		\node [style=none] (6) at (0.75, 2) {};
		\node [style=right label] (7) at (1, 1.75) {$D^{l}$};
		\node [style=left label] (8) at (-1, 1.75) {$C^{l}$};
		\node [style=left label] (9) at (-1, -1.75) {$A^{k_1}$};
		\node [style=right label] (10) at (1, -1.75) {$B^{k_2}$};
	\end{pgfonlayer}
	\begin{pgfonlayer}{edgelayer}
		\draw (1.center) to (3.center);
		\draw (1.center) to (4.center);
		\draw [in=270, out=90] (2.center) to (6.center);
		\draw (2.center) to (5.center);
	\end{pgfonlayer}
\end{tikzpicture}

%% file: figures/IMmap3.tikz
\begin{tikzpicture}
	\begin{pgfonlayer}{nodelayer}
		\node [style=large map] (0) at (0, 0) {$U_3$};
		\node [style=none] (1) at (-0.75, 0) {};
		\node [style=none] (2) at (0.75, 0) {};
		\node [style=none] (3) at (-0.75, 2) {};
		\node [style=none] (4) at (-0.75, -2) {};
		\node [style=none] (5) at (0.75, -2) {};
		\node [style=none] (6) at (0.75, 2) {};
		\node [style=right label] (7) at (1, 1.75) {$D^{l_2}$};
		\node [style=left label] (8) at (-1, 1.75) {$C^{l_1}$};
		\node [style=left label] (9) at (-1, -1.75) {$A^{k_1}$};
		\node [style=right label] (10) at (1, -1.75) {$B^{k_2}$};
	\end{pgfonlayer}
	\begin{pgfonlayer}{edgelayer}
		\draw (1.center) to (3.center);
		\draw (1.center) to (4.center);
		\draw [in=270, out=90] (2.center) to (6.center);
		\draw (2.center) to (5.center);
	\end{pgfonlayer}
\end{tikzpicture}

%% file: figures/IMmap4.tikz
\begin{tikzpicture}
	\begin{pgfonlayer}{nodelayer}
		\node [style=large map] (0) at (0, -1.5) {$U$};
		\node [style=none] (1) at (-0.75, -1.5) {};
		\node [style=none] (2) at (0.75, -1.5) {};
		\node [style=none] (3) at (-0.75, 1.5) {};
		\node [style=none] (4) at (-0.75, -3) {};
		\node [style=none] (5) at (0.75, -3) {};
		\node [style=none] (6) at (0.75, 1.5) {};
		\node [style=right label] (7) at (1, 0) {$D^k$};
		\node [style=left label] (8) at (-1, 0) {$C^k$};
		\node [style=left label] (9) at (-1, -2.75) {$A^k$};
		\node [style=right label] (10) at (1, -2.75) {$B^k$};
		\node [style=large map] (11) at (0, 1.5) {$V$};
		\node [style=none] (12) at (-0.75, 1.5) {};
		\node [style=none] (13) at (0.75, 1.5) {};
		\node [style=none] (14) at (-0.75, 3) {};
		\node [style=none] (15) at (0.75, 3) {};
		\node [style=right label] (16) at (1, 3) {$F^k$};
		\node [style=left label] (17) at (-1, 3) {$E^k$};
	\end{pgfonlayer}
	\begin{pgfonlayer}{edgelayer}
		\draw (1.center) to (3.center);
		\draw (1.center) to (4.center);
		\draw [in=270, out=90] (2.center) to (6.center);
		\draw (2.center) to (5.center);
		\draw (12.center) to (14.center);
		\draw [in=270, out=90] (13.center) to (15.center);
	\end{pgfonlayer}
\end{tikzpicture}

%% file: figures/DotExample1.tikz
\begin{tikzpicture}
	\begin{pgfonlayer}{nodelayer}
		\node [style=right label] (0) at (-2, 2.75) {$A_O$};
		\node [style=right label] (1) at (2, 2.75) {$B_O$};
		\node [style=none] (2) at (-2, -3) {};
		\node [style=right label] (3) at (-2, -2.75) {$A_I$};
		\node [style=none] (4) at (2, -3) {};
		\node [style=right label] (5) at (2, -2.75) {$B_I$};
		\node [style=none] (6) at (-2, 3) {};
		\node [style=none] (7) at (2, 3) {};
		\node [style=right label] (8) at (0, 2.75) {$E_O$};
		\node [style=none] (9) at (0, -3) {};
		\node [style=right label] (10) at (0, -2.75) {$E_I$};
		\node [style=none] (11) at (0, 3) {};
		\node [style=white dot] (12) at (0, 1.75) {};
		\node [style=white dot] (13) at (0, -1.75) {};
		\node [style=white dot] (18) at (2, 0) {};
		\node [style=white dot] (19) at (-2, 0) {};
		\node [style=right label] (26) at (0.875, 1) {${R'}^k$};
		\node [style=right label] (27) at (0.5, -1.5) {$R^k$};
		\node [style=left label] (28) at (-0.75, -1.5) {$L^k$};
		\node [style=left label] (29) at (-0.7, 1) {${L'}^k$};
	\end{pgfonlayer}
	\begin{pgfonlayer}{edgelayer}
		\draw [style=arrow plain] (9.center) to (13);
		\draw [style=arrow plain] (2.center) to (19);
		\draw [style=arrow plain] (4.center) to (18);
		\draw [style=arrow plain] (13) to (18);
		\draw [style=arrow plain] (13) to (19);
		\draw [style=arrow plain] (19) to (6.center);
		\draw [style=arrow plain] (19) to (12);
		\draw [style=arrow plain] (18) to (12);
		\draw [style=arrow plain] (12) to (11.center);
		\draw [style=arrow plain] (18) to (7.center);
	\end{pgfonlayer}
\end{tikzpicture}

%% file: figures/DotExample2.tikz
\begin{tikzpicture}
	\begin{pgfonlayer}{nodelayer}
		\node [style=white dot] (0) at (0, -2.25) {};
		\node [style=white dot] (2) at (-1, 1) {};
		\node [style=white dot] (3) at (1, -1) {};
		\node [style=white dot] (4) at (0, 2.25) {};
		\node [style=none] (5) at (0, -3.25) {};
		\node [style=none] (6) at (0, 3.25) {};
		\node [style=left label] (7) at (-1, -1.5) {$B^k$};
		\node [style=right label] (8) at (0, -3) {$A$};
		\node [style=right label] (9) at (0.75, 0.75) {$F^l$};
		\node [style=right label] (10) at (-0.3, -0.075) {$D^{kl}$};
		\node [style=right label] (11) at (0, 3) {$G$};
		\node [style=left label] (12) at (-0.6, 1.6) {$E^l$};
		\node [style=right label] (13) at (0.5, -2) {$C^k$};
	\end{pgfonlayer}
	\begin{pgfonlayer}{edgelayer}
		\draw [style=arrow plain] (5.center) to (0);
		\draw [style=arrow plain] (4) to (6.center);
		\draw [style=arrow plain] (0) to (3);
		\draw [style=arrow plain, in=-60, out=150] (3) to (2);
		\draw [style=arrow plain, in=-105, out=135] (0) to (2);
		\draw [style=arrow plain] (2) to (4);
		\draw [style=arrow plain, in=-60, out=90] (3) to (4);
	\end{pgfonlayer}
\end{tikzpicture}

%% file: figures/DotExample3.tikz
\begin{tikzpicture}
	\begin{pgfonlayer}{nodelayer}
		\node [style=white dot] (0) at (0, -0.5) {};
		\node [style=white dot] (1) at (-1.5, 0.75) {};
		\node [style=none] (2) at (-1.5, 1.75) {};
		\node [style=none] (3) at (0, 1.75) {};
		\node [style=none] (4) at (1.5, 1.75) {};
		\node [style=none] (5) at (0, -1.5) {};
		\node [style=right label] (6) at (0, -1.25) {$A$};
		\node [style=right label] (7) at (-1.5, 1.5) {$E^{l}$};
		\node [style=right label] (8) at (0, 1.5) {$C^{km}$};
		\node [style=right label] (9) at (1.5, 1.5) {$D^{lm}$};
		\node [style=left label] (10) at (-1, -0.5) {$B^{kl}$};
	\end{pgfonlayer}
	\begin{pgfonlayer}{edgelayer}
		\draw [style=arrow plain, in=-90, out=150] (0) to (1);
		\draw [style=arrow plain] (1) to (2.center);
		\draw [style=arrow plain] (0) to (3.center);
		\draw [style=arrow plain, in=-90, out=30] (0) to (4.center);
		\draw [style=arrow plain] (5.center) to (0);
	\end{pgfonlayer}
\end{tikzpicture}

%% file: figures/DotExample5.tikz
\begin{tikzpicture}
	\begin{pgfonlayer}{nodelayer}
		\node [style=white dot] (0) at (-0.75, 0) {};
		\node [style=white dot] (1) at (0.75, 0) {};
		\node [style=none] (2) at (-0.75, 1.25) {};
		\node [style=none] (3) at (0.75, 1.25) {};
		\node [style=none] (4) at (-0.75, -1.25) {};
		\node [style=none] (5) at (0.75, -1.25) {};
		\node [style=right label] (6) at (-0.75, 1) {$C^k$};
		\node [style=right label] (7) at (0.75, 1) {$D^k$};
		\node [style=right label] (9) at (-0.75, -1) {$A$};
		\node [style=right label] (10) at (0.75, -1) {$B$};
	\end{pgfonlayer}
	\begin{pgfonlayer}{edgelayer}
		\draw [style=arrow plain] (4.center) to (0);
		\draw [style=arrow plain] (5.center) to (1);
		\draw [style=arrow plain] (0) to (2.center);
		\draw [style=arrow plain] (1) to (3.center);
	\end{pgfonlayer}
\end{tikzpicture}

%% file: figures/DotCompExample1.tikz
\begin{tikzpicture}
	\begin{pgfonlayer}{nodelayer}
		\node [style=white dot] (0) at (0, 0) {};
		\node [style=none] (1) at (-1, 1.25) {};
		\node [style=none] (2) at (1, 1.25) {};
		\node [style=none] (3) at (0, -1.25) {};
		\node [style=right label] (4) at (1, 1) {$C^l$};
		\node [style=right label] (5) at (-1, 1) {$B^l$};
		\node [style=right label] (6) at (0, -1) {$A^k$};
	\end{pgfonlayer}
	\begin{pgfonlayer}{edgelayer}
		\draw [style=arrow plain, in=-90, out=150, looseness=0.75] (0) to (1.center);
		\draw [style=arrow plain, in=-90, out=30] (0) to (2.center);
		\draw [style=arrow plain] (3.center) to (0);
	\end{pgfonlayer}
\end{tikzpicture}

%% file: figures/DotCompExample2.tikz
\begin{tikzpicture}
	\begin{pgfonlayer}{nodelayer}
		\node [style=none] (1) at (-1.25, -1) {};
		\node [style=none] (2) at (0.75, -1) {};
		\node [style=right label] (4) at (0.75, -0.75) {$C$};
		\node [style=right label] (5) at (-1.25, -0.75) {$B$};
		\node [style=white dot] (6) at (0.75, 0) {};
		\node [style=white dot] (7) at (-1.25, 0.25) {};
		\node [style=none] (8) at (0, 1.25) {};
		\node [style=none] (9) at (1.5, 1.25) {};
		\node [style=right label] (10) at (1.5, 1) {$F^k$};
		\node [style=right label] (11) at (0, 1) {$E^k$};
		\node [style=none] (12) at (-1.25, 1.25) {};
		\node [style=right label] (13) at (-1.25, 1) {$D$};
	\end{pgfonlayer}
	\begin{pgfonlayer}{edgelayer}
		\draw [style=arrow plain] (2.center) to (6);
		\draw [style=arrow plain] (1.center) to (7);
		\draw [style=arrow plain, in=-90, out=135] (6) to (8.center);
		\draw [style=arrow plain, in=-90, out=45] (6) to (9.center);
		\draw [style=arrow plain, in=270, out=90] (7) to (12.center);
	\end{pgfonlayer}
\end{tikzpicture}

%% file: figures/DotCompExample3.tikz
\begin{tikzpicture}
	\begin{pgfonlayer}{nodelayer}
		\node [style=none] (0) at (-1.25, -1) {};
		\node [style=none] (1) at (0.75, -1) {};
		\node [style=right label] (2) at (0.75, -0.75) {$C^{l'}$};
		\node [style=right label] (3) at (-1.25, -0.75) {$B^{l}$};
		\node [style=white dot] (4) at (0.75, 0) {};
		\node [style=white dot] (5) at (-1.25, 0.25) {};
		\node [style=none] (6) at (0, 1.25) {};
		\node [style=none] (7) at (1.5, 1.25) {};
		\node [style=right label] (8) at (1.5, 1) {$F^k$};
		\node [style=right label] (9) at (0, 1) {$E^k$};
		\node [style=none] (10) at (-1.25, 1.25) {};
		\node [style=right label] (11) at (-1.25, 1) {$D^l$};
	\end{pgfonlayer}
	\begin{pgfonlayer}{edgelayer}
		\draw [style=arrow plain] (1.center) to (4);
		\draw [style=arrow plain] (0.center) to (5);
		\draw [style=arrow plain, in=-90, out=135] (4) to (6.center);
		\draw [style=arrow plain, in=-90, out=45] (4) to (7.center);
		\draw [style=arrow plain, in=270, out=90] (5) to (10.center);
	\end{pgfonlayer}
\end{tikzpicture}

%% file: figures/DotCompExample4.tikz
\begin{tikzpicture}
	\begin{pgfonlayer}{nodelayer}
		\node [style=none] (0) at (-1, -1) {};
		\node [style=none] (1) at (1, -1) {};
		\node [style=right label] (2) at (1, -0.75) {$C^{l'}$};
		\node [style=right label] (3) at (-1, -0.75) {$B^{l'}$};
		\node [style=white dot] (4) at (1, 0) {};
		\node [style=white dot] (5) at (-1, 0.25) {};
		\node [style=none] (6) at (0.25, 1.25) {};
		\node [style=none] (7) at (1.75, 1.25) {};
		\node [style=right label] (8) at (1.75, 1) {$F^k$};
		\node [style=right label] (9) at (0.25, 1) {$E^k$};
		\node [style=none] (10) at (-1, 1.25) {};
		\node [style=right label] (11) at (-1, 1) {$D^{l'}$};
	\end{pgfonlayer}
	\begin{pgfonlayer}{edgelayer}
		\draw [style=arrow plain] (1.center) to (4);
		\draw [style=arrow plain] (0.center) to (5);
		\draw [style=arrow plain, in=-90, out=135] (4) to (6.center);
		\draw [style=arrow plain, in=-90, out=45] (4) to (7.center);
		\draw [style=arrow plain, in=270, out=90] (5) to (10.center);
	\end{pgfonlayer}
\end{tikzpicture}

%% file: figures/DotCompExample5.tikz
\begin{tikzpicture}
	\begin{pgfonlayer}{nodelayer}
		\node [style=white dot] (0) at (0, -0.75) {};
		\node [style=none] (3) at (0, -2) {};
		\node [style=right label] (6) at (0, -1.75) {$A^k$};
		\node [style=right label] (9) at (0.75, -0.5) {$C^l$};
		\node [style=left label] (10) at (-0.975, -0.5) {$B^l$};
		\node [style=white dot] (11) at (1, 0.5) {};
		\node [style=white dot] (12) at (-1, 0.5) {};
		\node [style=none] (13) at (0.25, 1.75) {};
		\node [style=none] (14) at (1.75, 1.75) {};
		\node [style=right label] (15) at (1.75, 1.5) {$F^{k'}$};
		\node [style=right label] (16) at (0.25, 1.5) {$E^{k'}$};
		\node [style=none] (17) at (-1, 1.75) {};
		\node [style=right label] (18) at (-1, 1.5) {$D^l$};
	\end{pgfonlayer}
	\begin{pgfonlayer}{edgelayer}
		\draw [style=arrow plain] (3.center) to (0);
		\draw [style=arrow plain, in=-90, out=135] (11) to (13.center);
		\draw [style=arrow plain, in=-90, out=45] (11) to (14.center);
		\draw [style=arrow plain, in=270, out=90] (12) to (17.center);
		\draw [style=arrow plain, in=-90, out=165] (0) to (12);
		\draw [style=arrow plain, in=-90, out=15] (0) to (11);
	\end{pgfonlayer}
\end{tikzpicture}

%% file: figures/DotCompExample6.tikz
\begin{tikzpicture}
	\begin{pgfonlayer}{nodelayer}
		\node [style=none] (0) at (0.5, -1) {};
		\node [style=none] (1) at (2.5, -1) {};
		\node [style=right label] (2) at (2.5, -0.75) {$C^{l'}$};
		\node [style=right label] (3) at (0.5, -0.75) {$B^{l''}$};
		\node [style=white dot] (4) at (2.5, 0) {};
		\node [style=white dot] (5) at (0.5, 0.25) {};
		\node [style=none] (6) at (1.75, 1.25) {};
		\node [style=none] (7) at (3.25, 1.25) {};
		\node [style=right label] (8) at (3.25, 1) {$F^{k'}$};
		\node [style=right label] (9) at (1.75, 1) {$E^{k'}$};
		\node [style=none] (10) at (0.5, 1.25) {};
		\node [style=right label] (11) at (0.5, 1) {$D^{l''}$};
		\node [style=white dot] (12) at (-1.75, 0) {};
		\node [style=none] (13) at (-2.75, 1.25) {};
		\node [style=none] (14) at (-0.75, 1.25) {};
		\node [style=none] (15) at (-1.75, -1) {};
		\node [style=right label] (16) at (-0.75, 1) {$C^l$};
		\node [style=right label] (17) at (-2.75, 1) {$B^l$};
		\node [style=right label] (18) at (-1.75, -0.75) {$A^k$};
	\end{pgfonlayer}
	\begin{pgfonlayer}{edgelayer}
		\draw [style=arrow plain] (1.center) to (4);
		\draw [style=arrow plain] (0.center) to (5);
		\draw [style=arrow plain, in=-90, out=135] (4) to (6.center);
		\draw [style=arrow plain, in=-90, out=45] (4) to (7.center);
		\draw [style=arrow plain, in=270, out=90] (5) to (10.center);
		\draw [style=arrow plain, in=-90, out=150, looseness=0.75] (12) to (13.center);
		\draw [style=arrow plain, in=-90, out=30] (12) to (14.center);
		\draw [style=arrow plain] (15.center) to (12);
	\end{pgfonlayer}
\end{tikzpicture}